 \documentclass[11pt]{article}
 \newcommand{\palatinoformat}{}

\newcommand{\manuscript}{tutorial}

\newcommand{\adv}{\mathcal A} 
\newcommand{\Adv}{\mathcal B}

\newcommand{\oracle}{\mathcal O}

\newcommand{\kk}{\kappa}

\newcommand{\pk}{\mathit{pk}}
\newcommand{\sk}{\mathit{sk}}

\newcommand{\Succ}{\mathsf{Succ}}

\newcommand{\negl}{\mathsf{negl}}
\newcommand{\neglK}{\negl(\kk)}

\newcommand{\nB}{\mathit{nb}}
\newcommand{\nC}{\mathit{nc}}
\newcommand{\mC}{\mathit{mc}}
\newcommand{\mB}{\mathit{mb}}

\newcommand{\votespace}{\mathcal{M}}

\newcommand{\bb}{\mathfrak{bb}}

\newcommand{\outcome}{\mathfrak v}
\newcommand{\tpf}{\mathit{pf}}

\newcommand{\GenSymb}{\mathsf{Gen}}
\newcommand{\EncSymb}{\mathsf{Enc}}
\newcommand{\DecSymb}{\mathsf{Dec}}
\newcommand{\Gen}[1][\kk]{\GenSymb(#1)}
\newcommand{\Enc}[2][\pk]{\EncSymb(#1,#2)}
\newcommand{\Dec}[2][\sk]{\DecSymb(#1,#2)}

\newcommand{\INDCPA}{\mathsf{IND}\textsf{-}\mathsf{CPA}}

\newcommand{\CNMCPA}{\mathsf{CNM}\textsf{-}\mathsf{CPA}}

\newcommand{\INDPA}{\mathsf{IND}\textsf{-}\mathsf{PA0}}
\newcommand{\gamePINDPA}[3]{\INDPA(#1, #2,#3)}
\newcommand{\gameINDPA}{\gamePINDPA{\Pi}{\adv}{\kk}}

\newcommand{\SetupSymb}{\mathsf{Setup}}
\newcommand{\VoteSymb}{\mathsf{Vote}}
\newcommand{\TallySymb}{\mathsf{Tally}}
\newcommand{\VerifySymb}{\mathsf{Verify}}

\newcommand{\Setup}[1][\kk]{\SetupSymb(#1)}
\newcommand{\Vote}[1][\pk,\allowbreak v,\allowbreak\nC,\allowbreak\kk]{\VoteSymb(#1)}
\newcommand{\Tally}[1][\sk,\allowbreak\bb,\allowbreak\nC,\allowbreak \kk]{\TallySymb(#1)}
\newcommand{\Verify}[1][\pk,\allowbreak\bb,\allowbreak\nC,\allowbreak\outcome,\allowbreak\tpf,\allowbreak\kk]{\VerifySymb(#1)}

\newcommand{\encToVoteSymb}{\mathsf{Enc2Vote}}
\newcommand{\encToVote}[1][\Pi]{\encToVoteSymb(#1)}

\newcommand{\heliosspec}{Helios'12}
\newcommand{\heliosnext}{Helios'16}

\newcommand{\Completeness}{\mathsf{Completeness}}
\newcommand{\Injectivity}{\mathsf{Injectivity}}
\newcommand{\Soundness}{\mathsf{Soundness}}

\newcommand{\IV}{\mathsf{Individual}\textrm{-}\allowbreak\mathsf{Verifiability}}
\newcommand{\UV}{\mathsf{Universal}\textrm{-}\allowbreak\mathsf{Verifiability}}

\newcommand{\IVGame}[1][\Gamma,\allowbreak\adv,\allowbreak\kk]{\IV(#1)}

\newcommand{\CompletenessGame}[1][\Gamma,\allowbreak\adv,\allowbreak\kk]{\Completeness(#1)}
\newcommand{\SoundnessGame}[1][\Gamma,\allowbreak\adv,\allowbreak\kk]{\Soundness(#1)}

\newcommand{\BallotSecrecy}{\mathsf{Ballot}\textrm{-}\allowbreak\mathsf{Secrecy}}
\newcommand{\BallotSecrecyGame}[1][\Gamma,\allowbreak\adv,\allowbreak\kk]{\BallotSecrecy(\allowbreak{}#1)}

\newcommand{\BallotIndependence}{\mathsf{IND}\textsf{-}\allowbreak\mathsf{CVA}}

\newcommand{\BallotIndependenceGame}[1][\Gamma,\allowbreak\adv,\allowbreak\kk]{\BallotIndependence(#1)}

\newcommand{\balanced}{\mathit{balanced}}

\newcommand{\correcttally}{\mathit{correct\text{-}outcome}}

\newcommand{\Exp}{\mathsf{Exp}}

\newcommand{\VerSymb}{\mathsf{Verify}}

\newcommand{\blue}[1]{\textcolor{blue}{#1}} 
\newcommand{\red}[1]{\textcolor{red}{#1}}

\newcommand{\AdvInd}{\mathsf{Ind}} 

\newcommand{\auditoutcome}{s} 

\ifdefined\acmformat
  \usepackage{colortbl}
\else
  \usepackage{palatino}
  \usepackage[margin=1in]{geometry}
  \usepackage[dvipsnames,table]{xcolor}
  \usepackage{amsfonts}
  \usepackage{amsmath}
  \usepackage{amssymb}
    \usepackage{amsthm}
  \usepackage{authblk}
  \usepackage{cite}
  \usepackage[hyphens]{url}
  \usepackage{hyperref}
  \usepackage[utf8]{inputenc}
\fi
\definecolor{light-gray}{gray}{.95}

\usepackage[normalem]{ulem}
\usepackage[vlined,commentsnumbered,linesnumbered]{algorithm2e}
\usepackage{float}\floatstyle{ruled}\restylefloat{figure}
\usepackage{multicol}
\usepackage{multirow}
\usepackage{pifont}\newcommand{\cmark}{\ding{51}}\newcommand{\xmark}{\ding{55}}
\usepackage[framemethod=default]{mdframed}
\mdfdefinestyle{guidebox}{backgroundcolor=white,innertopmargin=0.5\baselineskip,innerbottommargin=0.5\baselineskip,skipabove=\baselineskip,skipbelow=\baselineskip}
\mdfdefinestyle{infobox}{backgroundcolor=light-gray,linecolor=black!50,innertopmargin=0.5\baselineskip,innerbottommargin=0.5\baselineskip,skipabove=\baselineskip,skipbelow=\baselineskip}
\usepackage{verbatim}

\usepackage{tikz}
\usetikzlibrary{arrows,positioning}
\usetikzlibrary{arrows.meta}
\usetikzlibrary{matrix}
\usetikzlibrary{shapes,calc}

\SetKwRepeat{Do}{do}{while}

\newenvironment{inlineexperiment}[1]
 {\medskip\noindent{#1}$\;=$\\
  \begin{algorithm}[H]}
 {\end{algorithm}}

\newcommand{\llIf}[2]{{\let\par\relax\lIf{#1}{#2}}}
\newcommand{\lleIf}[3]{{\let\par\relax\lIf{#1}{#2} \lElse{#3}}}
\newcommand{\llFor}[2]{{\let\par\relax\lFor{#1}{#2}}}

\ifdefined\acmformat
  \newtheorem{guideline}{Design guideline}
  \newtheorem{remark}{Remark}
\else
  \newtheorem{theorem}{Theorem}
  \newtheorem{corollary}[theorem]{Corollary}
  \newtheorem{lemma}[theorem]{Lemma}
  \newtheorem{proposition}[theorem]{Proposition}

  \newtheorem{remark}[theorem]{Remark}
  \newtheorem{example}{Example}
  \newtheorem{definition}{Definition}
  \newtheorem{guideline}{Design guideline}
\fi

\ifdefined\acmformat
\title{Secrecy and Verifiability: An Introduction to Electronic Voting}

\author{Paul Keeler}
\affiliation{%
  \institution{University of Melbourne}
  \department{School of Mathematics and Statistics}
  \city{Melbourne}
  \country{Australia}
}
\email{paul.keeler@unimelb.edu.au}

\author{Ben Smyth}
\affiliation{%
  \institution{Independent Researcher}
  \country{United Kingdom}
}
\email{research@bensmyth.com}


\fi

\ifdefined\palatinoformat
\title{Secrecy and Verifiability: An Introduction to Electronic Voting}

\author[1]{Paul Keeler}
\author[2]{Ben Smyth}
\affil[1]{School of Mathematics and Statistics, University of Melbourne, Melbourne}
\affil[2]{Independent Researcher}

\date{\today{}}
\fi

\begin{document}

\ifdefined\acmformat
  \begin{abstract}
Democracies are built upon secure and reliable voting systems. Electronic voting systems seek to replace ballot papers and boxes with computer hardware and software. Proposed electronic election schemes have been subjected to scrutiny, with researchers spotting inherent faults and weaknesses. Inspired by physical voting systems, we argue that any electronic voting system needs two essential properties: ballot secrecy and verifiability. These properties seemingly work against each other. An election scheme that is a complete black box offers ballot secrecy, but verification of the outcome is impossible. This challenge can be tackled using standard tools from modern cryptography, reaching a balance that delivers both properties.

This tutorial makes these ideas accessible to readers outside electronic voting. We introduce fundamental concepts such as asymmetric and homomorphic encryption, which we use to describe a general electronic election scheme while keeping mathematical formalism minimal. We outline game-based cryptography, a standard approach in modern cryptography, and introduce notation for formulating elections as games. We then give precise definitions of ballot secrecy and verifiability in the framework of game-based cryptography. A principal aim is introducing modern research approaches to electronic voting.
\end{abstract}
  \maketitle
\else
  \maketitle
  \begin{abstract}
Democracies are built upon secure and reliable voting systems. Electronic voting systems seek to replace ballot papers and boxes with computer hardware and software. Proposed electronic election schemes have been subjected to scrutiny, with researchers spotting inherent faults and weaknesses. Inspired by physical voting systems, we argue that any electronic voting system needs two essential properties: ballot secrecy and verifiability. These properties seemingly work against each other. An election scheme that is a complete black box offers ballot secrecy, but verification of the outcome is impossible. This challenge can be tackled using standard tools from modern cryptography, reaching a balance that delivers both properties.

This tutorial makes these ideas accessible to readers outside electronic voting. We introduce fundamental concepts such as asymmetric and homomorphic encryption, which we use to describe a general electronic election scheme while keeping mathematical formalism minimal. We outline game-based cryptography, a standard approach in modern cryptography, and introduce notation for formulating elections as games. We then give precise definitions of ballot secrecy and verifiability in the framework of game-based cryptography. A principal aim is introducing modern research approaches to electronic voting.
\end{abstract}
\fi

\section{Introduction to voting systems}\label{sec:intro}
Any democratic state needs to be able to offer an election, but what is an election and what properties should it have to help deliver democracy? Several properties emerge from democratic principles. Two of these appear fundamentally at odds. Voters need assurance that their votes are counted accurately, yet revealing how any individual voted would undermine the democratic process itself. Understanding these properties helps us see what is needed in building an \emph{electronic voting system}, which aims to replace the traditional voting setup of ballot papers and boxes with some combination of computer hardware and software.

An election is considered, essentially, a decision-making procedure for choosing representatives~\cite{Lijphart84,Saalfeld95,Gumbel05:StealThisVote,Alvarez10:ElectronicElections}.
Intergovernmental organizations such as the United Nations recommend voting systems that ensure that voting be done by voters so that each voter has equal influence over the result~\cite{UN:HumanRights,OSCE:HumanRights,OAS:HumanRights}. 
In an election, voters must be able to cast votes with what we call \emph{free-choice}, meaning all voters can vote for any election candidates they desire without fear of repercussions. Another crucial property is that nobody can alter cast votes or create false votes, without leaving evidence of such \emph{undue influence}. Electronic voting systems are already emerging, but unfortunately, these voting systems are routinely broken in ways to reveal that they violate free-choice~\cite{Wallach04:Diebold,Rop07:NetherlandsVoting,DebraBowenCalifornia07,Halderman10:IndiaVoting,Halderman12:DCVoting,Halderman14:EstoniaVoting}
or permit undue influence~\cite{Wallach04:Diebold,ElectoralCommision07,DebraBowenCalifornia07,GermanyCourt09,JonesSimons12:VotingBook}.

\subsection{Ballot secrecy}
To give voters in democratic elections free-choice, a voting system needs to offer \emph{ballot secrecy}, preventing the voters suffering from their cast votes. The secret ballot has a long history, being used in ancient Greece and Rome. The modern system is a relatively recent concept, first implemented in the Australian colonies in 1856 during British colonial rule, and subsequently adopted by Britain, the United States, and other democracies. The seemingly simple yet revolutionary change was to make ballots identical. Previously each political party could produce its own ballot paper, whose style and size clearly revealed anybody's voting intentions. A voting system with this lasting innovation is often called the \emph{Australian system}, which demands that votes be marked on uniform ballots in isolated polling booths and then deposited in ballot boxes~\cite{Brent06:AustralianBallot}. 

By granting ballot secrecy, the Australian system can assure that only voters vote and that they do so with free-choice. In this paper, we employ a definition of ballot secrecy introduced by  Smyth~\cite{Smyth15:BallotSecrecyFull}. 
\begin{itemize}
\item Ballot secrecy. A voter's vote is not revealed to anyone. 
\end{itemize}
To be clear, this definition of ballot secrecy means that not even the organizers of the election can see the vote of any voter. We will return to this definition later in Section~\ref{sec:secrecy:def}. 

\subsection{Verifiability}
A voting system also needs \emph{verifiability}, which is a way to prove that no undue influence has occurred. To provide this traditionally, election monitors or observers can check that only one ballot paper is distributed to each eligible voter who then only deposits one paper, without papers coming from other sources. The observers can then discard incomplete or tarnished ballot papers, with the remaining papers accurately corresponding to the election outcome. But any assurance of verifiability is limited by the monitoring ability of the observers~\cite{Bjornlund04:ElectionMonitoring,Kelley12:ElectionMonitoring,Norris15:ElectionMonitoring}.

\newcommand{\informalDefinitionUV}[1][a]{#1{}nyone can check whether an outcome corresponds to votes expressed in collected ballots}
\newcommand{\InformalDefinitionUV}{\informalDefinitionUV[A]}

The property of verifiability can be considered in a more specific manner.  In the setting of electronic voting, Smyth, Frink and  Clarkson~\cite{Smyth15:ElectionVerifiability} introduced  a definition for \emph{universal verifiability}, which is designed for anyone to check that no undue influence has occurred in the election and the election outcome is the true one. 
\begin{itemize}
\item Universal verifiability. \InformalDefinitionUV.
\end{itemize}

Universal verifiability ensures that the election outcome reflects the collected votes. But it is not enough for a voter to merely cast a vote and to assume it is properly collected, because
somebody may discard or modify the vote. We argue that voters must be able to
uniquely identify \emph{their} ballots.  Consequently, Smyth, Frink and  Clarkson~\cite{Smyth15:ElectionVerifiability} also introduced the notion of \emph{individual verifiability}.

\begin{itemize}
\item Individual verifiability. A voter can check whether their ballot is collected.
\end{itemize}

\subsection{Voting: a fundamental balance}
The two essential properties of ballot secrecy and verifiability seem to compete against each other. It is easy to create a trivial voting system that adheres to one of these properties. For example, if all voters attach their names to their ballot papers (being written on a piece paper or encoded electronically), we can immediately verify that all votes were cast by the proper parties, but then ballot secrecy is gone. 

On the surface these two conflicting properties may seem to destroy the chances of any electronic voting scheme possessing both of them in fullness: too much of one renders the other one not possible. We later see how this problem has repeatedly appeared in the research on electronic voting, due to the definitions of properties being too strong. But we argue that it is possible to arrive at a subtle balance, rather than a trade-off, between these two properties, resulting in voting systems that offer both secrecy and verifiability.

\subsection{Terminology}
We use \emph{election scheme} to refer to the cryptographic protocol, and \emph{voting system} to refer to implementations. Of course a secure election scheme does not guarantee a secure voting system.

The terms \emph{secrecy} and \emph{privacy} occasionally appear as synonyms in the literature; we prefer ballot secrecy to avoid confusion with other privacy notions. We discuss related concepts such as \emph{coercion resistance} and \emph{receipt-freeness} in Section~\ref{sec:secrecy:outlook}.

\subsection*{Structure of article}
The \manuscript{} is organized as follows:
\begin{itemize}
\item Section~\ref{sec:election-preview} previews an election scheme, giving a high-level picture before the technical details.
\item Section~\ref{sec:basics} covers cryptography basics, including negligible functions and encryption.
\item Section~\ref{sec:game-crypto} introduces game-based cryptography.
\item Section~\ref{sec:notation-primer} summarizes the notation used throughout.
\item Section~\ref{sec:election} gives an overview of election schemes.
\item Section~\ref{sec:syntax} defines the formal syntax for election schemes.
\item Section~\ref{sec:voting} surveys attacks on electronic voting and their countermeasures.
\item Section~\ref{sec:verifiability} formalizes individual and universal verifiability.
\item Section~\ref{sec:secrecy} formalizes ballot secrecy and proves sufficient conditions.
\item Section~\ref{sec:casestudies} presents case studies of Helios, Helios Mixnet, and Belenios.
\item Section~\ref{sec:summary} summarizes lessons learnt and concludes.
\end{itemize}
Readers familiar with cryptographic foundations may wish to skim or skip Sections~\ref{sec:basics}--\ref{sec:notation-primer} and proceed directly to Section~\ref{sec:election}.
\clearpage

\section{An election scheme at a glance}\label{sec:election-preview}

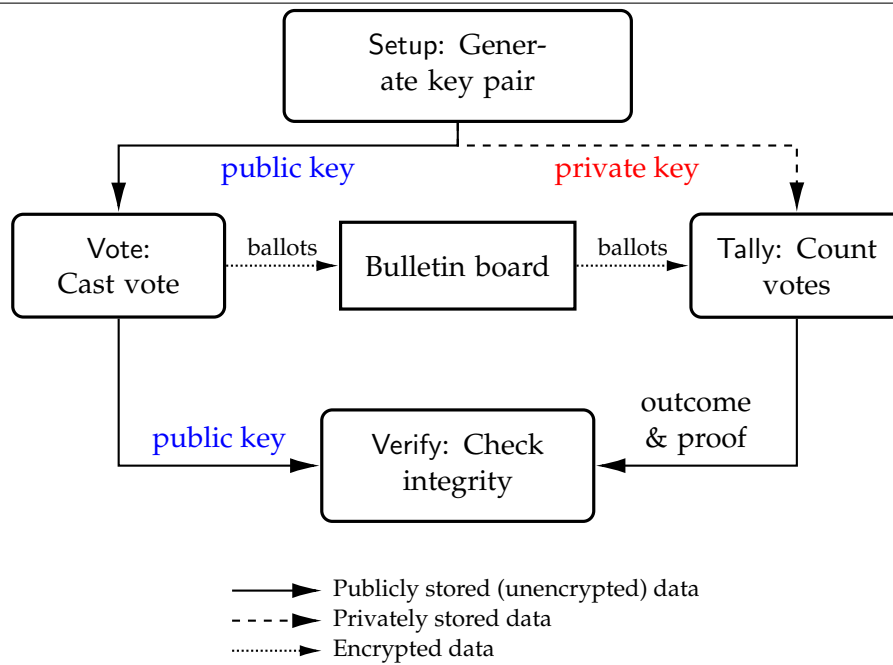
\begin{figure}
\centering
\begin{tikzpicture}[thick]  
\tikzset{
    primitive/.style={rectangle, rounded corners, draw=black, very thick, inner sep=0.8em, minimum size=3em, text centered},
    votestorage/.style={rectangle, draw=black, very thick, inner sep=0.8em, minimum size=3em, text centered},       
    arrowbig/.style={-{Latex[length=4mm,width=2mm]}, thick},
    arrowsmall/.style={-{Latex[length=3mm,width=1.5mm]}, thick},
    mylabel/.style={text width=7em, text centered} 
}
\node[primitive, text width=4cm] (setup) {$\SetupSymb$: Generate key pair};  
\node[votestorage, below=1.3cm of setup, text width=2.5cm] (bboard) {Bulletin board}; 
\node[primitive, left=1.5cm of bboard, text width=2.2cm] (vote) {$\VoteSymb$: Cast vote};  
\node[primitive, right=1.5cm of bboard, text width=2.2cm] (tally) {$\TallySymb$: Count votes};
\node[primitive, below=1.3cm of bboard, text width=3cm] (verify) {$\VerifySymb$: Check integrity};

\draw[arrowbig] (setup.south) -- ++(0,-0.3) -| node[pos=0.25, below] {\blue{public key}} (vote.north);
\draw[arrowbig, dashed] (setup.south) -- ++(0,-0.3) -| node[pos=0.25, below] {\red{private key}} (tally.north);

\draw[arrowsmall, densely dotted] (vote.east) -- node[above, font=\footnotesize] {ballots} (bboard.west); 
\draw[arrowsmall, densely dotted] (bboard.east) -- node[above, font=\footnotesize] {ballots} (tally.west); 

\draw[arrowbig] (vote.south) |- (verify.west) node[pos=0.75, above] {\blue{public key}};

\draw[arrowbig] (tally.south) |- (verify.east) node[pos=0.75, above, text width=2.5cm, align=center] {outcome\\ \& proof};

\node[below=0.9cm of verify, anchor=north] (legend) {};
\draw[arrowbig] ([xshift=-3cm]legend.north) -- ++(1.2,0) node[right, font=\footnotesize] {Publicly stored (unencrypted) data};
\draw[arrowbig, dashed] ([xshift=-3cm, yshift=-0.4cm]legend.north) -- ++(1.2,0) node[right, font=\footnotesize] {Privately stored data};
\draw[arrowsmall, densely dotted] ([xshift=-3cm, yshift=-0.8cm]legend.north) -- ++(1.2,0) node[right, font=\footnotesize] {Encrypted data};

\end{tikzpicture} 
\medskip
\caption{Overview of an election scheme with four algorithms. $\SetupSymb$ generates a key pair; the public key is used by $\VoteSymb$ to create encrypted ballots, while the private key is used by $\TallySymb$ to decrypt. Ballots flow through the bulletin board (dotted arrows indicate encrypted data). $\VerifySymb$ checks the election outcome using public information. A detailed version appears in Figure~\ref{figure-voting}.}
\label{figure-simplevoting}
\end{figure}

Before introducing the cryptographic tools that make electronic voting possible, it helps to see what an election scheme actually looks like. Many election schemes have been proposed; as a running example, we use one due to Smyth, Frink and Clarkson~\cite{Smyth15:ElectionVerifiability}, which we return to throughout the \manuscript{}. Figure~\ref{figure-simplevoting} shows its high-level structure, built from four algorithms.

The algorithm $\SetupSymb$ generates a pair of keys: a \emph{public key}, which is shared openly, and a \emph{private key}, which must be kept secret. Votes are encrypted using the public key, meaning anyone can create an encrypted vote, but only the holder of the private key can decrypt it. The algorithm $\VoteSymb$ takes a voter's choice and the public key, and produces an encrypted ballot that reveals nothing about the vote inside. These encrypted ballots are posted to a \emph{bulletin board}, which is a public, append-only log that anyone can read but nobody can secretly alter. Once voting ends, the algorithm $\TallySymb$ uses the private key to compute the election outcome and produce a proof that the count is correct. Finally, $\VerifySymb$ allows anyone to check this proof using only public information: the public key, the ballots on the bulletin board, and the claimed outcome. No private key is needed to verify.

The rest of this \manuscript{} makes all of this precise. We begin with the cryptographic building blocks, such as encryption, key generation, and related concepts, before returning to election schemes with formal definitions and security analysis.

\section{Cryptography basics}\label{sec:basics}
Delivering a voting system that promises both ballot secrecy and verifiability is a daunting task, but fortunately, the tools from modern cryptography offer a way.  We will now cover, in a rather informal fashion, some basic concepts from cryptography under the assumption that the reader has some fundamental knowledge in mathematics, logic, and algorithms. We will link concepts in cryptography to election schemes. Readers familiar with the fundamentals of cryptography can skip this section. 

Historically cryptography focused on ways of sending messages between two parties without anyone else being able to read the messages. In the cryptography literature, the two parties are often called Alice and Bob. There are other named characters appearing in cryptography narratives, with a prominent figure being Eve the eavesdropper, who wants to listen into communications between Alice and Bob. If Alice wants to send Bob a secret message, she can, for example, simply mangle or encrypt it in some way, making it seemingly unrecognizable, and then tells Bob how she did that. Then Bob can reverse Alice's encryption step by decrypting the message. 

\subsection{Symmetric encryption with private keys}
Encryption transforms a message, called the \emph{plaintext}, into an unreadable form called the \emph{ciphertext}. To decrypt, one needs a secret \emph{key}. In \emph{symmetric} (or \emph{private-key}) cryptography, the same key is used for both encryption and decryption. If we write $\EncSymb(k,m)$ for encrypting message $m$ with key $k$, then decryption satisfies $\DecSymb(k,\EncSymb(k,m))=m$. The limitation of symmetric cryptography is that Alice must somehow share the secret key with Bob while keeping it hidden from Eve, which is a challenge that asymmetric cryptography addresses.

\subsection{Asymmetric encryption with private and public keys} \label{sec:asym} 
To avoid Alice having to give Bob the private key, modern day cryptography uses a \emph{key pair} consisting of one key $\pk$ for only encrypting plaintexts and another key $\sk$ for only decrypting ciphertexts. If Bob wants to receive messages from Alice, he generates a pair of keys $(\pk,\sk)$, respectively called the \emph{public key} and the \emph{private key}. He then puts the public key $\pk$ somewhere it can be accessed by Alice (and anyone else), while keeping the private key $\sk$ secret. Alice uses the public key $\pk$ to encrypt her message and sends it to Bob who then uses his private key $\sk$, which only he knows, to decrypt the encrypted message. Alice and everyone else knows the public key $\pk$, but they never know the corresponding private key $\sk$ kept secret by Bob. It is computationally hard (or practically impossible in implementation) to infer $\sk$ from $\pk$, so only the holder of the private key $\sk$ can decrypt Alice's message.

One key for encrypting, another for decrypting. This revolutionary concept is called \emph{asymmetric} or \emph{public-key} cryptography. Mathematically, we can write asymmetric encryption as
$$c=\EncSymb(\pk,m)\,$$  
where $c$ is again a ciphertext and $\EncSymb$ is an algorithm representing asymmetric encryption. Decrypting the ciphertext $c$ with $\sk$ can then be written as 
$$\DecSymb(\sk,c)=\DecSymb(\sk,\EncSymb(\pk,m))=m\, ,$$
where $\DecSymb$ is the decryption algorithm that uses the private key $\sk$. But after using the public key for encryption, it cannot be used for decryption, implying $\DecSymb(\pk,c)\neq m\, .$ 

More formally, we can interpret the encryption algorithm $\EncSymb(\pk,m)$ as a function that maps  from the space of plaintexts $\votespace$ to the space of ciphertexts $\mathcal{C}$, while the encryption algorithm $\DecSymb(\pk,c)$ maps from  $\mathcal{C}$ to $\votespace$. 
A more precise description of asymmetric encryption is located in Section~\ref{app:sec:asym} of the appendix. 

\subsection{Examples of asymmetric cryptography}
It would be difficult to overstate the importance of asymmetric cryptography in modern computing. Since the 1970s, researchers have proposed various public-key schemes. The most famous is RSA, named for its inventors Rivest, Shamir, and Adleman~\cite{rivest1978method}. RSA exploits a fundamental asymmetry: it is easy to multiply two large prime numbers together, but incredibly hard to factor the product back into its prime components. This asymmetry, meaning easy in one direction, hard in reverse, is what makes public-key cryptography possible.

Other important schemes include Diffie--Hellman and ElGamal, based on the difficulty of the discrete logarithm problem. We will encounter ElGamal in the Helios voting system. Alternative schemes based on lattices, codes, and other structures are also an active area, particularly for post-quantum security (Section~\ref{sec:secrecy:outlook}).

\begin{example}[ElGamal encryption]\label{ex:elgamal}
We illustrate ElGamal with small numbers. Let $p = 23$ be
a prime and $g = 5$ a generator of the multiplicative group modulo~$p$.

\paragraph{Key generation.}
The key holder chooses a random private key $\sk = 6$ and computes the
public key
$$\pk = g^{\sk} \bmod p = 5^{6} \bmod 23 = 8.$$
The public parameters $(p, g, \pk) = (23, 5, 8)$ are published; the
private key $\sk = 6$ is kept secret.

\paragraph{Encryption.}
Suppose a voter wishes to encrypt a vote for candidate~$2$. In the
variant of ElGamal used for voting (sometimes called
\emph{exponential ElGamal}), the voter encrypts $g^m$ rather than $m$
directly.  With vote $m = 2$ and random value $r = 3$, the ciphertext
is the pair $(c_1, c_2)$:
\begin{align*}
c_1 &= g^{r} \bmod p = 5^{3} \bmod 23 = 10, \\
c_2 &= g^{m} \cdot \pk^{r} \bmod p
     = 5^{2} \cdot 8^{3} \bmod 23
     = 2 \cdot 6 \bmod 23
     = 12.
\end{align*}
The ciphertext $(10, 12)$ is posted to the bulletin board. A different
random~$r$ would produce a different ciphertext for the same
vote. This is why ElGamal encryption is \emph{probabilistic}, and why
two ballots for the same candidate look different.

\paragraph{Decryption.}
The key holder recovers $g^m$ using the private key:
$$g^{m} = c_2 \cdot (c_1^{\,\sk})^{-1} \bmod p
        = 12 \cdot (10^{6})^{-1} \bmod 23
        = 12 \cdot 6^{-1} \bmod 23
        = 12 \cdot 4 \bmod 23
        = 2.$$
Because the number of candidates is small, the key holder recovers~$m$
by checking a short table: $g^{1} \bmod 23 = 5$,
$g^{2} \bmod 23 = 2$, $g^{3} \bmod 23 = 10, \ldots$\;
The decrypted value matches $g^{2}$, so $m = 2$, confirming the voter voted for
candidate~$2$.

\paragraph{Homomorphic property.}
Suppose a second voter encrypts a vote for candidate~$1$ (so
$g^{m'} = g^{1} = 5$) with randomness $r' = 5$, obtaining ciphertext
$(c_1', c_2') = (20, 11)$.  Multiplying ciphertexts component-wise:
$$(c_1 \cdot c_1' \bmod 23,\; c_2 \cdot c_2' \bmod 23)
  = (10 \cdot 20 \bmod 23,\; 12 \cdot 11 \bmod 23)
  = (16, 17).$$
Decrypting this combined ciphertext yields $g^{m + m'} = g^{3} = 10$,
encoding the sum $m + m' = 3$, the total of both votes, without ever
decrypting either ballot individually.  This is how homomorphic
tallying works in Helios.
\end{example}

\subsubsection{Decisional Diffie--Hellman (DDH) assumption}
The security of both Diffie--Hellman key exchange and ElGamal encryption rests on the \emph{Decisional Diffie--Hellman (DDH) assumption}.  Informally, DDH states: given a group generator $g$ and the values $g^a$ and $g^b$ for random secret exponents $a$ and $b$, no efficient algorithm can distinguish $g^{ab}$ from a random group element.  That is, even knowing $g^a$ and $g^b$, the shared secret $g^{ab}$ looks random to any computationally bounded observer.  For ElGamal, DDH ensures that the ciphertext component $g^m \cdot \pk^r = g^m \cdot g^{\sk \cdot r}$ is indistinguishable from a random group element, hiding the vote $m$. 

\subsubsection{Perceived quantum threat}\label{ssub:perceived_quantum_threat}
A common misconception is that quantum computers would 
break all encryption. In fact, the famed quantum algorithm of 
Shor~\cite{shor1999polynomial} provides an \emph{exponential} 
speedup only for schemes whose security relies on 
number-theoretic problems, or more abstractly, problems with 
group-theoretic 
structure~\cite[Chapter~2]{lindell2017tutorials}, such as 
integer factorisation (RSA) and the discrete logarithm 
(Diffie--Hellman, ElGamal). No comparable exponential quantum 
speedup is known for lattice-based, code-based, or hash-based 
schemes. Symmetric encryption such as the Advanced Encryption Standard (AES), which secures most 
stored data in practice, is also unaffected. The best known 
quantum attack uses Grover's algorithm, which offers only a quadratic 
speedup, easily countered by doubling key lengths. The election 
schemes in this tutorial rely on ElGamal and are therefore 
vulnerable; we discuss post-quantum alternatives in 
Section~\ref{sec:secrecy:outlook}.

\subsection{Key generation}
The key generation algorithm $\GenSymb$ produces keys for an encryption scheme.

For asymmetric encryption, the key generation algorithm produces a \emph{key pair} $(\pk, \sk)$ consisting of a public key and a private key. The algorithm takes as input a \emph{security parameter} $\kk$, typically written in unary notation as $1^\kk$, which determines the size or strength of the generated keys; we write $(\pk, \sk) \leftarrow \GenSymb(\kk)$ to denote this. Larger values of $\kk$ produce longer keys that are harder to break but slower to use.

Key generation must incorporate randomness. Otherwise, if keys were predictable, an adversary could compute the same key and break the scheme. We write $(\pk, \sk) \leftarrow_R \GenSymb(\kk)$ to emphasize that generation involves random choices. For election systems, key generation is particularly sensitive: a compromised private key $\sk$ would allow an adversary to decrypt all ballots and violate ballot secrecy.

\subsection{Homomorphic encryption}\label{sec:homomorph}
A traditional problem with encrypted data, such as messages or votes, is that we need to decrypt it in order to perform operations on it, which then leaves it vulnerable to adversaries like Eve. But this issue is circumvented by the concept of \emph{homomorphic encryption}, which allows certain mathematical operations, such as addition, to be performed on the encrypted data, without ever decrypting it or even using the private key.

More formally, for all possible messages $m$ and $m'$ and public keys $\pk$ allowed in the encryption scheme, homomorphic under a binary operation $\oplus$, such as addition, means that 
$$\EncSymb(pk,m)\odot \EncSymb(pk,m')=\EncSymb(pk,m\oplus m')\,,$$  where $\oplus$ and $\odot$ are two binary operations, the first being operation for the plaintexts $m$ and $m'$, while the second is the corresponding operation for the ciphertexts, which is not necessarily the same type of operation. 
 Decrypting the ciphertext with the private key $\sk$ yields
$$\DecSymb(sk,\EncSymb(pk,m)\odot \EncSymb(pk,m'))= m\oplus m'. $$ 

Several homomorphic encryption schemes exist. Paillier~\cite{paillier1999public} introduced a scheme homomorphic under addition (where ciphertext multiplication corresponds to plaintext addition). ElGamal is homomorphic under multiplication but can be adapted for addition. These are \emph{partially} homomorphic, because they support only one operation. Gentry~\cite{gentry2009fully} showed that \emph{fully} homomorphic encryption, which allows arbitrary computations on encrypted data, is possible, though this remains too computationally expensive for practical voting systems. The schemes used in electronic voting, including those we study, rely on partially homomorphic encryption.

For a voting scheme, homomorphic encryption is extremely useful, as it is possible to count or tally the votes of the voters without ever decrypting the individual votes or even accessing the private key. For example, in a two-candidate election, voters encrypt either $0$ or $1$. These encrypted votes are homomorphically combined to derive a ciphertext encoding the total votes for candidate $1$, which is then decrypted to reveal the outcome, without ever revealing individual votes.

As discussed by Paillier~\cite{paillier1999public},  Park and Rivest~\cite{park2017towards} and others, homomorphic encryption is important for preserving secrecy in the voting system. 

\subsection{Malleable and non-malleable encryption}
An encryption scheme is \emph{malleable} if, given a ciphertext, one can produce a valid ciphertext for a related plaintext, possibly without learning the original plaintext. A \emph{non-malleable} scheme prevents this: given a ciphertext, it is infeasible to create another valid ciphertext for a related message.

Homomorphic encryption is inherently malleable (combining ciphertexts produces a valid ciphertext for the combined plaintexts), while non-malleable encryption cannot be homomorphic. For voting, this creates a tension: homomorphic encryption enables efficient tallying, but malleability can enable attacks. We return to this connection when discussing ballot privacy. 

\subsection{Threshold encryption}\label{sec:threshold}
In \emph{threshold encryption}, decryption requires collaboration among multiple parties. With $n$ total \emph{shares} (partial decryption keys), at least $t$ shares must cooperate to decrypt, since fewer than $t$ shares reveal nothing. A \emph{perfect threshold} scheme ensures that $t-1$ shares provide no information about the missing shares.

For voting, threshold encryption is essential because it distributes trust: no single authority can decrypt ballots alone. Consider an election where the private key is split among five trustees with threshold $t=3$. Even if two trustees are corrupt or compromised, they cannot decrypt any ballot without cooperation from at least one honest trustee. This prevents a single point of failure that could compromise ballot secrecy.

Threshold encryption also enables \emph{distributed key generation}, where trustees jointly generate the key pair without any single party ever knowing the complete private key. The public key is published for voters to encrypt their ballots, but decryption requires the trustees to collaborate, typically after the election closes. Katz and Lindell~\cite[Section 13.3.3]{katz2014introduction} cover threshold encryption in relation to a simple election scheme.

\subsection{Algorithms and primitives}
In cryptography, algorithms such as $\EncSymb$, $\DecSymb$, and $\GenSymb$ are called \emph{primitives} when treated as building blocks of larger schemes. We use the terms algorithm and primitive interchangeably.

There is some specific language used for describing the running of primitives or algorithms. To start an algorithm is to \emph{invoke} it, and the running of an algorithm is an \emph{invocation}. To provide input parameters to the algorithm is to \emph{instantiate} it, so an \emph{instantiation} of an algorithm is its set of inputs for one invocation. 

\subsection{Complexity approach to cryptography}
Modern cryptography is based on \emph{computational} security rather than information-theoretic security: adversaries could in principle break the system, but doing so would require non-polynomial time. Schemes are considered secure if the probability of an efficient adversary breaking them is negligible. In practical terms, security means the scheme will not be broken for hundreds or thousands of years with foreseeable technology.

\subsection{Proofs}
A \emph{proof} is an object or interaction that convinces the reader of some mathematical claim. 
When proofs are interactive, we introduce two more characters from cryptography: Peggy the prover and Victor the verifier. Peggy needs to convince Victor in a series of interactions. At the end Victor outputs a response, either accepting or rejecting the claim by Peggy, based on whether Peggy's actions have managed to convince him of the validity of a certain mathematical statement. 

Two fundamental properties characterize the quality of proof systems. \emph{Completeness} requires that if a statement is true, an honest prover can always convince the verifier to accept. \emph{Soundness} requires that if a statement is false, no cheating prover can convince the verifier to accept (except with negligible probability). In other words, completeness ensures that valid proofs succeed, while soundness ensures that invalid proofs fail. These concepts extend beyond proof systems: in the context of election verification, soundness means that the verification algorithm rejects incorrect outcomes, while completeness means it accepts correct ones. 

A \emph{zero-knowledge proof} allows Peggy to convince Victor of a statement without revealing \emph{why} it is true. A classic illustration: suppose you have a red ball and a green ball, identical except for colour, and a colour-blind friend. By repeatedly asking whether the friend switched the balls behind their back, you can convince them the colours differ; yet the friend never learns which ball is which. For voting, zero-knowledge proofs are essential: the system must convince everyone that votes were tallied correctly without revealing any individual vote.

\subsection{Fiat-Shamir transformation}\label{sec:fiat-shamir}
The \emph{Fiat-Shamir transformation} converts interactive proofs into non-interactive ones.

The key insight of the Fiat-Shamir transformation is to replace the verifier's random challenge with the output of a cryptographic hash function. In an interactive protocol, the prover sends a commitment to the verifier, the verifier responds with a random challenge, and the prover computes a response based on both the commitment and challenge. The verifier then checks the response. In the non-interactive version, instead of receiving a random challenge from the verifier, the prover computes the challenge themselves by hashing the commitment (and typically other public data). Since hash functions behave unpredictably, the prover cannot choose a commitment that will produce a favorable challenge.

More formally, suppose an interactive proof has the prover send commitment $a$, receive challenge $c$, and respond with $z$. The Fiat-Shamir transformation produces a non-interactive proof $(a, c, z)$ where $c = H(a, m)$ for a hash function $H$ and message $m$ being proved. The verifier checks both that $c = H(a, m)$ and that $(a, c, z)$ is a valid transcript.

The security of the Fiat-Shamir transformation is typically analyzed in the \emph{random oracle model}, which treats the hash function as a truly random function. While no real hash function is truly random, this idealization captures the intuition that a good hash function is unpredictable. Proofs in the random oracle model provide strong evidence of security, though they do not guarantee security when the random oracle is instantiated with a concrete hash function.

A subtle but important variant is the \emph{weak Fiat-Shamir transformation}, which hashes only the commitment $a$ rather than including the message $m$. This seemingly minor difference can introduce serious vulnerabilities in some contexts, as we will see when analyzing the Helios voting system. The weak variant allows an adversary more freedom in constructing proofs, potentially enabling attacks that the standard transformation would prevent.

For electronic voting, the Fiat-Shamir transformation is essential because it allows voters to construct non-interactive proofs that their ballots are well-formed without requiring real-time interaction with a verifier. These proofs can then be posted to a public bulletin board and verified by anyone at any time.

\subsection{Simulator}
In cryptography, an important concept related to proofs is that of a \emph{simulator}. This is essentially a fast (probabilistic polynomial-time) algorithm that can produce outputs when called upon  by the verifier without interacting with the real prover, and the outputs are indistinguishable from those resulting from interactions with the real prover. In other words, given Victor the verifier and Peggy the prover, the simulator is able to \emph{simulate} Victor's interaction with Peggy.

\subsection{Some useful functions}
We now introduce functions that appear throughout cryptographic definitions.

\subsubsection{Hash functions}
A \emph{hash function} maps data of arbitrary size to a fixed-size output. The main requirement is \emph{collision resistance}, meaning it should be computationally infeasible to find two distinct inputs that produce the same output. In the voting context, collisions could allow manipulation of election outcomes by creating duplicate ballots.

\subsubsection{Negligible functions}
A \emph{negligible function} $\negl(\kk)$ is one that decreases faster than the inverse of any polynomial as the security parameter $\kk$ grows. For example, $2^{-\kk}$ and $\kk^{-\log \kk}$ are negligible. Events occurring with negligible probability are considered infeasible in practice, because the time needed for them to occur would be more than the lifetime of our universe. A cryptographic scheme is secure if the probability of any efficient adversary breaking it is negligible; see Katz and Lindell~\cite[Section 3.1.2]{katz2014introduction} for numerical examples.

\subsection{Importance of randomness}
Why does randomness appear everywhere in cryptography? Two principal reasons. First, randomness produces unpredictable data that is secret to all parties. This is essential for key generation: if a key generation algorithm produced predictable keys, an adversary could compute the same key and break the scheme.

Second, there is always some probability that an adversary could guess a key by chance. A good cryptographic system makes this probability negligible, meaning vanishingly small as the security parameter increases. The negligible functions defined above make this notion precise: as the key length grows, the probability of a successful guess decreases faster than any polynomial, rendering brute-force attacks infeasible.

\subsection{Probabilistic algorithms}
A \emph{probabilistic algorithm} (also called a \emph{randomized algorithm}) takes a deterministic input, performs operations using random choices (such as coin flips), and produces an output influenced by that randomness. Two executions on the same input will likely produce different outputs. For voting, this is crucial: if encryption were deterministic, an adversary could encrypt each candidate's name and compare the results to voters' ballots, breaking secrecy.

\subsection{Random variables}
We write the output of a probabilistic algorithm as $A(x_1,\dots,x_n; r)$, where $x_1,\dots,x_n$ are deterministic inputs and $r$ represents the random choices. The randomness of $r$ typically stems from  a sequence of independent random bits, called \emph{coins} in cryptography (or \emph{Bernoulli random variables} in probability theory). We write $x \leftarrow_R S$ to denote sampling an element $x$ uniformly at random from set $S$. A \emph{nonce} is a number used only once, often generated randomly for authentication.

\subsection{Polynomial-time algorithms}
An algorithm is \emph{polynomial-time} if the number of steps it takes is bounded by a polynomial in the input size. A \emph{probabilistic polynomial-time (PPT)} algorithm runs in polynomial time regardless of its random choices. Security definitions in cryptography restrict adversaries to probabilistic polynomial-time algorithms, which captures efficient attackers, while excluding those requiring exponential resources.


\subsection{Cryptographic tools for electronic voting}\label{sec:crypto-voting}
The cryptographic primitives described above, such as encryption, proofs, and hash functions, are general-purpose tools. Electronic voting systems combine these primitives in specific ways to achieve ballot secrecy and verifiability. This section introduces cryptographic constructions that are particularly important for voting applications.

\subsubsection{Bulletin boards}\label{sec:bulletin-board}
The \emph{bulletin board} is a central component of electronic voting systems. It serves as a public ledger where encrypted ballots are posted and stored until tallying. The bulletin board must satisfy several properties:

\begin{itemize}
\item \emph{Public readability}: Anyone can read the contents of the bulletin board. This enables universal verifiability, as any observer can check that the announced outcome corresponds to the posted ballots.

\item \emph{Append-only}: Once a ballot is posted, it cannot be removed or modified. New ballots can only be added. This ensures that voters can verify their ballot remains on the board throughout the election.

\item \emph{Consistency}: All observers see the same bulletin board contents. Different parties cannot be shown different versions of the board.

\item \emph{Availability}: The bulletin board remains accessible throughout the election and verification period.
\end{itemize}

In our formal model, we represent the bulletin board as a set $\mathfrak{bb}$ of ballots. The set representation abstracts away implementation details while capturing the essential property that ballots are collected for tallying. In practice, bulletin boards may be implemented using web servers or distributed ledgers.

We typically assume an honest bulletin board in security definitions: ballots posted by voters appear on the board, and the board accurately reflects all posted ballots. Relaxing this assumption leads to stronger adversary models where the adversary controls ballot collection.

\subsubsection{Encoding votes}\label{sec:vote-encoding}
Election schemes need to represent votes in a form suitable for cryptographic operations. Several encoding strategies exist, each with trade-offs for efficiency and the types of elections they support.

\paragraph{Integer encoding.}
The simplest approach encodes a vote for candidate $v \in \{1, \ldots, \mathit{nc}\}$ directly as the integer $v$. This works well with homomorphic encryption schemes that support addition: the sum of encrypted votes yields an encrypted sum, but this sum alone does not reveal individual vote counts per candidate. Additional techniques (such as computing the sum of $g^{v_i}$ values and solving a discrete logarithm for small results) can extract the tally.

\paragraph{Bitstring encoding.}
A more common approach for homomorphic tallying encodes each vote as a binary vector. For an election with $\mathit{nc}$ candidates, a vote for candidate $v$ becomes a vector of length $\mathit{nc} - 1$:
\begin{itemize}
\item If $v < \mathit{nc}$: position $v$ contains $1$, all other positions contain $0$
\item If $v = \mathit{nc}$: all positions contain $0$
\end{itemize}

For example, with $\mathit{nc} = 4$ candidates:
\begin{center}
\begin{tabular}{cl}
Vote for candidate 1: & $(1, 0, 0)$ \\
Vote for candidate 2: & $(0, 1, 0)$ \\
Vote for candidate 3: & $(0, 0, 1)$ \\
Vote for candidate 4: & $(0, 0, 0)$
\end{tabular}
\end{center}

Each component is encrypted separately, producing a tuple of ciphertexts. Homomorphically combining ciphertexts column-wise across all ballots yields encrypted sums; decrypting these reveals the vote count for each candidate without decrypting individual ballots.

The Helios voting system uses bitstring encoding. This is why Helios ballots contain tuples of ciphertexts $(c_1, \ldots, c_{\mathit{nc}-1})$ rather than a single ciphertext.

\begin{example}[Bitstring encoding and homomorphic tallying]\label{ex:tallying-trace}
We trace the full pipeline for a small election with
$\nC = 3$ candidates and three voters.

\paragraph{Step 1: Encode.}
Each voter encodes their vote as a binary vector of length
$\nC - 1 = 2$:
\begin{center}
\begin{tabular}{clc}
& Vote & Encoding \\[2pt]
Voter 1: & candidate 2 & $(0, 1)$ \\
Voter 2: & candidate 1 & $(1, 0)$ \\
Voter 3: & candidate 1 & $(1, 0)$
\end{tabular}
\end{center}

\paragraph{Step 2: Encrypt.}
Each voter encrypts the components of their encoding separately and
posts the resulting ballot to the bulletin board. Writing $E(x)$ for an
encryption of~$x$, the bulletin board contains:
\begin{center}
\begin{tabular}{l@{\qquad}cc}
& Column 1 & Column 2 \\[2pt]
Voter 1: & $E(0)$ & $E(1)$ \\
Voter 2: & $E(1)$ & $E(0)$ \\
Voter 3: & $E(1)$ & $E(0)$
\end{tabular}
\end{center}
Note that each $E(\cdot)$ uses fresh randomness, so no two ciphertexts
look alike, even those encrypting the same value.

\paragraph{Step 3: Combine.}
The tallier combines ciphertexts \emph{column-wise} using the
homomorphic property, without decrypting any individual ballot:
\begin{align*}
\underbrace{E(0) \otimes E(1) \otimes E(1)}_{\text{Column 1}}
  &= E(0 + 1 + 1) = E(2), \\
  \underbrace{E(1) \otimes E(0) \otimes E(0)}_{\text{Column 2}}
  &= E(1 + 0 + 0) = E(1).
\end{align*}  

\paragraph{Step 4: Decrypt and read off the result.}
The tallier decrypts each column sum exactly once:
$$\text{Column 1} \to 2, \qquad \text{Column 2} \to 1.$$
These are the vote counts for candidates~$1$ and~$2$, respectively.
The count for candidate~$3$ is obtained by subtraction:
$3 - 2 - 1 = 0$, where $3$ is the total number of ballots.

The election outcome is the vector $(2, 1, 0)$: candidate~$1$ received
two votes, candidate~$2$ received one vote, and candidate~$3$ received
none. At no point was any individual ballot decrypted; only the
column-wise sums were. This is the mechanism used by Helios.
\end{example}

\subsubsection{Mix networks}\label{sec:mixnets}
A \emph{mix network} (or \emph{mixnet}) is a cryptographic protocol for anonymous communication, introduced by Chaum~\cite{Chaum:1981}. Mix networks break the link between voters and their votes while still allowing votes to be tallied.

\paragraph{Basic concept.}
A set of encrypted messages enters the network, and the same messages exit, re-encrypted and reordered, so that an observer who sees both inputs and outputs cannot determine which input corresponds to which output, provided at least one mix server behaves honestly.

Consider an analogy. Several people each place a sealed envelope containing a message into a box. A trusted party shakes the box thoroughly, opens all envelopes, and reads out the messages in random order. An observer learns all messages but cannot determine who wrote which one. A mix network achieves this electronically, replacing physical shuffling with cryptographic operations.

\paragraph{Structure.}
A mix network consists of a sequence of \emph{mix servers} $M_1, M_2, \ldots, M_k$. Each server receives a batch of ciphertexts, randomly reorders (shuffles) them, and \emph{re-encrypts} each one with fresh randomness. Re-encryption produces a new ciphertext that looks entirely different but still encrypts the same plaintext; this is possible with encryption schemes like ElGamal that support homomorphic re-encryption. Even if an observer knows some servers' permutations, as long as one server's permutation remains secret, the overall mapping from inputs to outputs is hidden.

\begin{figure}
\centering
\begin{tikzpicture}[thick]
\tikzset{
    mixserver/.style={rectangle, rounded corners, draw=black, very thick, inner sep=0.5em, minimum height=3.5em, minimum width=2.2cm, text width=2cm, align=center},
    ciphertext/.style={rectangle, draw=black, thick, inner sep=0.3em, minimum height=1.5em},
    arrowmix/.style={-{Latex[length=2.5mm,width=1.2mm]}, thick},
    labelstyle/.style={font=\footnotesize}
}

\node[ciphertext] (c1) at (0, 1.2) {$c_1$};
\node[ciphertext] (c2) at (0, 0) {$c_2$};
\node[ciphertext] (c3) at (0, -1.2) {$c_3$};

\node[mixserver] (mix1) at (3.2, 0) {Mix Server 1\\{\footnotesize shuffle \& re-encrypt}};
\node[mixserver] (mix2) at (7.2, 0) {Mix Server 2\\{\footnotesize shuffle \& re-encrypt}};

\node[ciphertext] (c1out) at (10.4, 1.2) {$c'_3$};
\node[ciphertext] (c2out) at (10.4, 0) {$c'_1$};
\node[ciphertext] (c3out) at (10.4, -1.2) {$c'_2$};

\draw[arrowmix] (c1.east) -- (mix1.west |- c1);
\draw[arrowmix] (c2.east) -- (mix1.west);
\draw[arrowmix] (c3.east) -- (mix1.west |- c3);

\draw[arrowmix] (mix1.east) -- (mix2.west);

\draw[arrowmix] (mix2.east |- c1out) -- (c1out.west);
\draw[arrowmix] (mix2.east) -- (c2out.west);
\draw[arrowmix] (mix2.east |- c3out) -- (c3out.west);

\node[above=0.3cm of c1, labelstyle] {Input};
\node[above=0.3cm of c1out, labelstyle] {Output};

\draw[dashed, gray] (1.3, -2) -- (1.3, 2);
\draw[dashed, gray] (9.1, -2) -- (9.1, 2);
\node[below=0.3cm of mix1, xshift=2cm, text width=5cm, align=center, labelstyle, text=gray] {Unlinkable if at least one server is honest};

\end{tikzpicture}
\caption{Operation of a mix network. Ciphertexts enter and pass through a sequence of mix servers. Each server randomly permutes its inputs and re-encrypts them, producing ciphertexts that encrypt the same plaintexts but appear different and are in an unpredictable order. An observer cannot link inputs to outputs provided at least one server keeps its permutation secret.}
\label{fig:mixnet}
\end{figure}

Figure~\ref{fig:mixnet} illustrates how ciphertexts pass through a sequence of mix servers.

\paragraph{Verifiable shuffles.}
For voting, mix servers must prove they shuffled correctly without revealing the permutation they used. Each server therefore produces a zero-knowledge proof showing that its output ciphertexts are exactly a re-encryption and permutation of its inputs: nothing added, removed, duplicated, or modified. These proofs allow anyone to verify mixing integrity while learning nothing about which input maps to which output.

\paragraph{Application to voting.}
In a mixnet-based voting system, voters encrypt their votes and post ciphertexts to the bulletin board. The mix servers sequentially shuffle and re-encrypt the ciphertexts, each providing a proof of correct shuffling. After the final shuffle, the ciphertexts are decrypted (using threshold decryption) to reveal votes, which are then tallied. Mixing breaks the link between voters and encrypted votes, so decryption reveals only shuffled votes rather than who cast each one. The shuffle proofs ensure universal verifiability. The Helios Mixnet system (Section~\ref{sec:heliosMixnet}) uses this approach.

\subsubsection{Homomorphic tallying vs.\ mixnet tallying}\label{sec:tally-approaches}
Electronic voting systems use two main approaches to compute election outcomes while preserving ballot secrecy.

\paragraph{Homomorphic tallying.}
With additively homomorphic encryption (Section~\ref{sec:homomorph}), encrypted votes can be combined without decryption:
$$\mathsf{Enc}(v_1) \otimes \mathsf{Enc}(v_2) \otimes \cdots \otimes \mathsf{Enc}(v_n) = \mathsf{Enc}(v_1 + v_2 + \cdots + v_n)$$
Only the final sum is decrypted, revealing the aggregate tally but never individual votes. This is efficient and provides strong privacy, but it is limited to elections where the outcome is a sum of votes. It cannot support complex voting methods like ranked-choice or approval voting. The basic Helios system uses homomorphic tallying.

\paragraph{Mixnet tallying.}
With mix networks, the encrypted votes are first shuffled through the mixnet, then each shuffled ciphertext is decrypted to reveal an individual vote. The revealed votes can then be tallied using any counting method. This supports arbitrary voting methods and provides flexibility for complex elections, but it reveals individual votes (though not who cast them) and requires verification of the shuffle proofs. Helios Mixnet uses this approach to support elections beyond simple plurality voting.

\begin{figure}
\centering
\begin{tikzpicture}[thick]
\tikzset{
    procbox/.style={rectangle, rounded corners, draw=black, thick, inner sep=0.6em, minimum height=2em, align=center},
    databox/.style={rectangle, draw=black, inner sep=0.4em, minimum height=1.8em},
    arrowtally/.style={-{Latex[length=2.5mm,width=1.5mm]}, thick},
    titlestyle/.style={font=\bfseries}
}

\node[titlestyle] at (-3.5, 2.5) {Homomorphic Tallying};

\node[databox] (hinput) at (-3.5, 1.5) {Encrypted votes};
\node[procbox] (hcombine) at (-3.5, 0) {Combine homomorphically};
\node[procbox] (hdec) at (-3.5, -1.5) {Decrypt once};
\node[databox] (hout) at (-3.5, -3) {Tally (sum only)};

\draw[arrowtally] (hinput) -- (hcombine);
\draw[arrowtally] (hcombine) -- (hdec);
\draw[arrowtally] (hdec) -- (hout);

\node[titlestyle] at (3.5, 2.5) {Mixnet Tallying};

\node[databox] (minput) at (3.5, 1.5) {Encrypted votes};
\node[procbox] (mshuffle) at (3.5, 0) {Shuffle (break links)};
\node[procbox] (mdec) at (3.5, -1.5) {Decrypt each};
\node[databox] (mout) at (3.5, -3) {Individual votes};

\draw[arrowtally] (minput) -- (mshuffle);
\draw[arrowtally] (mshuffle) -- (mdec);
\draw[arrowtally] (mdec) -- (mout);

\draw[dashed, gray] (0, 2.8) -- (0, -3.5);

\end{tikzpicture}
\caption{Two approaches to tallying encrypted votes. \emph{Left:} Homomorphic tallying combines ciphertexts mathematically, then decrypts once to reveal only the aggregate sum, so the individual votes are never exposed. \emph{Right:} Mixnet tallying shuffles ciphertexts to break the link between voters and ballots, then decrypts each vote individually. Homomorphic tallying reveals less information but only supports summation; mixnet tallying supports arbitrary voting methods but reveals individual votes (though not who cast them).}
\label{fig:tally-comparison}
\end{figure}

Figure~\ref{fig:tally-comparison} contrasts the two approaches.

\subsubsection{Well-definedness}\label{sec:well-defined}
An encryption scheme satisfies \emph{well-definedness} if there is an efficient way to distinguish properly formed ciphertexts from malformed ones. A ciphertext is \emph{well-formed} if it could have been produced by the encryption algorithm on some valid plaintext; otherwise it is \emph{ill-formed}.

For many encryption schemes, any bit string of appropriate length could potentially be a ciphertext, making tampering difficult to detect. A well-defined scheme provides a method to check whether a ciphertext lies in the valid ciphertext space.

Well-definedness is important for voting because:
\begin{itemize}
\item It allows detection of malformed ballots before tallying.
\item It ensures tallying algorithms behave predictably on all inputs.
\item It prevents attacks where adversaries submit crafted invalid ciphertexts that cause unexpected behavior during decryption.
\end{itemize}

Well-definedness is distinct from non-malleability. Non-malleability prevents producing a \emph{related valid} ciphertext from an existing one. Well-definedness allows detecting \emph{invalid} ciphertexts regardless of how they were produced. Voting systems typically need both properties.

\subsubsection{Ballot weeding}\label{sec:ballot-weeding}
Even with non-malleable encryption, an adversary might submit a ballot that is \emph{related} to another voter's ballot in a way that leaks information during tallying. For example, if the adversary can copy another voter's encrypted ballot and submit it as their own, the tally would reveal whether that voter's choice matched a particular candidate (by showing two votes instead of one).

\emph{Ballot weeding} is a countermeasure that removes or rejects ballots deemed to be meaningfully related to other ballots before tallying. The precise definition of \emph{meaningfully related} depends on the voting system, but typically includes duplicate ballots and ballots that can be shown to encrypt related plaintexts. Ballot weeding complements non-malleability: non-malleability makes it hard to create related ballots, while ballot weeding provides a second line of defense by detecting and removing any that slip through.

\subsection{Further reading on cryptography}
For introductions to the modern theory of cryptography, see, for example, the texts by Goldreich~\cite{goldreich2003foundations} or Katz and Lindell~\cite{katz2014introduction}. The classic text by Menezes, Van Oorschot and Vanstone~\cite{menezes1996handbook} gives a wide-encompassing treatment of applied cryptography. 
More advanced readers can also see the recent edited collection of cryptography tutorials~\cite{lindell2017tutorials}. Gentry\cite{gentry2010intro} covers the basics of fully homomorphic encryption in an approachable fashion.

\section{Game-based cryptography}\label{sec:game-crypto}
A modern approach to cryptography is to detect \emph{breaks} (that is, develop suitable attacks that find holes) in cryptography systems by using the concept of a \emph{game} or an \emph{experiment}. In game-based cryptography, a game formulates a series of interactions between a benign challenger, a malicious adversary $\adv$, and a cryptography scheme. The adversary $\adv$ wins the game by completing a \emph{task} that captures an execution of the scheme in which security is broken. In other words, the adversary $\adv$ wins by doing what should be unachievable. Formally, games in cryptography are interpreted as probabilistic (or randomized) algorithms that output Booleans, that is, answers such as true or false; 	$\top$ or $\bot$; $0$ or  $1$; and so on. Importantly, adversaries are said to be \emph{stateful}, meaning information persists across invocations of an adversary in a game, so adversaries can access earlier assignments.

In the electronic voting literature, it is standard to analyze voting systems by using methods from game-based cryptography. In the current setting, we consider a benign challenger, a malicious adversary and a voting system. The game is that the adversary seeks to break the security of the voting system by, for example, changing votes or revealing the votes of voters. For example, Smyth, Frink and  Clarkson~\cite{Smyth15:ElectionVerifiability} used this approach to study their definitions of universal and individual verifiability. 

We will now illustrate the basics of game-based cryptography. Notation used throughout this section and the rest of the \manuscript{} is collected in Section~\ref{sec:notation-primer}.

\subsection{Indistinguishability as a game}\label{sec:indgame}
For a concrete example of a game, we introduce the concept of \emph{indistinguishability}, which is when an adversary observes a single ciphertext but then is incapable of determining which of two message the ciphertext corresponds to.  In general, indistinguishability is important in cryptography. For voting systems, ballot secrecy can be expressed as the inability to distinguish between an instance of a voting system in which voters cast some votes, from another instance in which the voters cast a permutation of those votes. 

An encryption scheme is said to be \emph{perfectly indistinguishable} if no adversary $\adv$ can succeed with probability better than one half, meaning, no attacker can do any better than guessing the correct message half the time. But this concept is an idealization, as the probability will not equal exactly one half in a finite system. The probability should approach or converge exceedingly fast to one half as some security parameter $\kk$, such as the key length, approaches infinity.  

Consequently we will define a game that formalize the notion of an encryption scheme being \emph{indistinguishable}. To consider this game formally, we write $\Pi=(\EncSymb,\DecSymb,\GenSymb)$ to denote an encryption scheme with the message space $\votespace$, where we recall the respective primitives $\EncSymb$, $\DecSymb$ and $\GenSymb$ for encryption, decryption and key generation.  We now define a game called \emph{indistinguishability}, which we denote by $\AdvInd$. 
	
{\upshape
\begin{inlineexperiment}{Game $\AdvInd(\Pi, \adv,\kk) $}
The adversary  $\adv$ outputs two messages $m_0$ and $m_1$, which both belong to the space of all possible messages $\votespace$, meaning  $m_0,m_1\in \votespace$\;
A key pair $\pk$ and $\sk$ is generated using the key generation primitive $\GenSymb(\kk)$\;
A random uniform bit is also generated $\beta$, meaning a random variable $\beta\in \{0,1\}$, which is used to randomly choose a message $m_{\beta}$\;
Applying the encryption primitive $\EncSymb$ to the (randomly chosen) message $m_{\beta}$, a ciphertext $c$ is computed, written as $c\leftarrow \EncSymb(\pk,m_{\beta})$\; 
Ciphertext $c$ is given to the adversary $\adv$\;
The  adversary $\adv$ outputs a bit $\beta'$\;
The output of the game is defined to be $\top$ (true) if the bit $\beta' =\beta$, and $\bot$ (false) otherwise\;
\end{inlineexperiment}
}
We say that the adversary $\adv$ succeeds if the outcome of the above game is $\top$, which we write as $\top\leftarrow \AdvInd(\cdot)$. We write  $\Succ(\AdvInd(\cdot))=\mathrm{Pr}[\top\leftarrow \AdvInd(\cdot)]$  to denote the probability that the adversary succeeds in the game $\AdvInd$.  We can now formally define the concept of an encryption scheme being indistinguishable.
\begin{definition}[Indistinguishable]
Given the game $\AdvInd$, an encryption scheme $\Pi=(\EncSymb,\DecSymb,\GenSymb)$ with the message space $\votespace$, and a security parameter $\kk$, is indistinguishable if for every adversary $\adv$ the following inequality
$$
\Succ(\AdvInd(\cdot)) \leq  \frac{1}{2}  +\neglK \,.
$$
holds, where $\neglK$ is a negligible function in $\kk$.
\end{definition}
The probability of the adversary succeeding is at most one half plus a negligible amount. We see that as the (finite) security parameter $\kk$ approaches infinity, the probability of the adversary succeeding approaches one half (overwhelmingly fast). This can be expressed as $\Succ(\AdvInd(\cdot)) \rightarrow 1/2$ as  $\kk\rightarrow \infty$, which in the limit gives a perfectly indistinguishable encryption scheme.

\begin{mdframed}[style=infobox,frametitle={\textbf{How to read a security game}}]
A security game is \emph{not} a procedure to execute; it is a \emph{probability experiment} used to define what security means.  When reading a game:

\begin{enumerate}
\item \textbf{Identify the secret.}  There is usually a random bit $\beta$ (or similar hidden value) that the adversary must guess.

\item \textbf{Identify the adversary's view.}  What information does the adversary receive?  This is typically limited to public keys, ciphertexts, and oracle responses, and never the private key or the secret bit directly.

\item \textbf{Read the \textsf{Return} line as the winning condition.}  The game outputs $\top$ (true) when the adversary wins.  The conditions in the \textsf{Return} line define \emph{what counts as a win}; they are not steps the adversary performs.

\item \textbf{Think probabilistically.}  Every $\leftarrow_R$ introduces randomness.  The question is: over all these random choices, how often does the adversary win?  Security means the adversary cannot win much more often than by random guessing.
\end{enumerate}

\noindent
Each game defines a specific threat model.  Different games for the same scheme capture different attack scenarios (for example, passive eavesdropping vs.\ active manipulation).
\end{mdframed}

\subsection{Hierarchy of security notions}\label{sec:security-hierarchy}
The indistinguishability game from Section~\ref{sec:indgame} captures a basic notion of security. Cryptographers have defined several variants, forming a hierarchy from weaker to stronger guarantees. Understanding this hierarchy helps in selecting appropriate security assumptions for proofs. To orient the reader, the key notions and their logical relationships are
$$\text{IND-CCA} \rightarrow \text{NM-CPA} \leftrightarrow \text{IND-PA0} \rightarrow \text{IND-CPA}$$
where arrows indicate implication. The weakest notion, IND-CPA, captures security against passive eavesdroppers. The strongest, IND-CCA, implies non-malleability (NM-CPA). For election schemes, the intermediate notion IND-PA0 turns out to be the most natural, because its parallel decryption oracle models the tallying phase where all submitted ballots are decrypted together. We now define each notion.

\paragraph{IND-CPA: Indistinguishability under Chosen Plaintext Attack.}
This is the basic notion of semantic security that we introduced above. An adversary who can obtain encryptions of chosen messages still cannot distinguish encryptions of two messages of their choice. This captures security against passive eavesdroppers who observe ciphertexts but cannot interfere with the system.

\paragraph{IND-CCA: Indistinguishability under Chosen Ciphertext Attack.}
A stronger notion where the adversary additionally has access to a decryption oracle and can obtain decryptions of chosen ciphertexts (except the challenge ciphertext). This models active attackers who might trick honest parties into decrypting messages. IND-CCA security implies non-malleability.

\paragraph{IND-PA0: Indistinguishability under Parallel Attack.}
An intermediate notion, introduced by Bellare and Sahai~\cite{Bellare99:IND-k-CPA}, where the adversary receives a challenge ciphertext and can then submit a vector of ciphertexts for decryption (all at once, not adaptively), excluding the challenge itself. This notion is particularly relevant for voting because tallying is fundamentally a decryption step: the election authority decrypts all submitted ballots together to produce the outcome. The parallel decryption oracle in IND-PA0 models exactly this situation, where an adversary constructs ballots that will all be decrypted in a single batch during tallying. The formal definition appears in Section~\ref{sec:INDPA} of the appendix.

\paragraph{NM-CPA: Non-Malleability under Chosen Plaintext Attack.}
Rather than an indistinguishability requirement, this captures a qualitatively different property: given a ciphertext, the adversary cannot produce a ciphertext of a \emph{meaningfully related} plaintext. For example, given an encryption of $m$, the adversary should not be able to produce an encryption of $m+1$ without knowing $m$.

\paragraph{Relationships.}
The equivalence NM-CPA $\leftrightarrow$ IND-PA0 is a theorem of Bellare and Sahai~\cite{Bellare99:IND-k-CPA}, showing that non-malleability and this form of indistinguishability are the same notion expressed differently. This equivalence is useful because IND-PA0 is often easier to work with in proofs than the more complex definition of non-malleability.

For election schemes, we typically require IND-PA0 (equivalently, NM-CPA) because ballot secrecy requires indistinguishability, verifiability requires non-malleability, and the parallel attack model matches the voting scenario where multiple ballots are decrypted together during tallying.

\subsection{Game formulation}
Using our notation,  we can formulate a game, denoted by $\Exp(H,\mathcal{S},\adv)$, which gives an adversary $\adv$ the task of distinguishing between a function $H$ 
and a simulator $\mathcal{S}$.  (The idea of a simulator is important in game-based cryptography, as it gives a means for the adversary to generate false information with the aim of winning the game.)

 The game $\Exp(H,\mathcal{S},\adv)$ is as follows.

\begin{inlineexperiment}{Game $\Exp(H,\mathcal{S},\adv) $}
$m \leftarrow \adv();  \beta \leftarrow_R \{0,1\}$ \;
\lleIf{$\beta=0$}{$x \leftarrow H(m)$}{$x\leftarrow  S(m)$} \;
$  g \leftarrow \adv(x)$ \;
$\Return\ (g = \beta)$ \;
\end{inlineexperiment}

In words, the adversary $\adv$ first creates a message $m$ and then a random bit $\beta$ is generated. If this random bit is equal zero (so $\beta=0$), a function $H$ is applied to the message; otherwise, the simulator $S$ is applied. The adversary receives the output $x$ and must guess the value of $\beta$. The adversary wins if their guess $g$ matches $\beta$.

An adversary \emph{wins} a game by causing it to output $\top$, and \emph{loses} the game if the output is $\bot$. Since adversaries are stateful, information persists across invocations of an adversary in a game, meaning adversaries can access earlier assignments.  

We say that the adversary's \emph{success} in a game $\Exp(\cdot)$, denoted $\Succ(\Exp(\cdot))$, is the probability that the adversary wins. That is, adversary's success is
\begin{align*}
\Succ(\Exp(\cdot))& =\mathrm{Pr}[x\leftarrow \Exp(\cdot) : x = \top]\\
&=\allowbreak \mathrm{Pr}[\top \leftarrow \Exp(\cdot)]\,.
\end{align*}
Interpreting the game $\Exp(\cdot)$ as a random variable, which takes the value $\top$ or $\bot$, we can use standard probability notation $\Succ(\Exp(\cdot)) =\mathrm{Pr}[\Exp(\cdot) = \top]$. We have used the probability to formulate game success, because we focus on computational security, rather than information-theoretic  security. This means we tolerate breaks by adversaries in non-polynomial time and  breaks with negligible success, since such breaks are infeasible in practice.

\subsection{Oracles}
An important notion in cryptography is that of an oracle, which formalizes the idea of accessing certain information to help break a cryptographic scheme. The oracle is treated as a blackbox that can be given certain questions or task to aid the adversary, which can be answered or done quickly. (The notion of an oracle can be made formal by modifying a Turing machine, but such details are not important here.)

To give an example of an oracle at work, an adversary $\adv$ in a game could use an oracle to encrypt messages by using a key that is unknown to the adversary. Using specifically chosen messages, the adversary $\adv$ could infer the key by examining the ciphertexts produced by the oracle. The adversary $\adv$  can interact with the oracle as many times as needed and retain information from previous interactions. For another example, an oracle may access game parameters such as the bit $\beta$ used in the game $\Exp$.

\subsection{Games with many interactions}
The aforementioned game $\Exp$ captures a single interaction between the challenger and the adversary. But if we couple games with oracles, we can extend the games so that they capture arbitrarily many interactions. For instance, we can formulate a strengthening of $\Exp$ as follows.

{\upshape
\begin{inlineexperiment}{Game $\Exp^\oracle(H,\mathcal{S},\adv)$}
$\beta \leftarrow_R \{0,1\}$ \;
$g \leftarrow \adv^\oracle()$ \;
$\Return\ (g = \beta)$ \;
\end{inlineexperiment}
}

\noindent
where $\adv^\oracle$ denotes $\adv$'s access to oracle $\oracle$, and the oracle $\oracle$ is defined as follows:
\begin{itemize}
\item $\oracle(m)$ computes \lleIf{$\beta=0$}{$x \leftarrow H(m)$}{$x\leftarrow  S(m)$} and outputs $x$.
\end{itemize}

In words, the game $\Exp^\oracle$ first generates a random bit $\beta$, then allows the adversary to interact with the oracle $\oracle$ as many times as desired. Each time the adversary submits a message $m$ to the oracle, it receives either $H(m)$ or $S(m)$ depending on the hidden bit $\beta$. After gathering information from these interactions, the adversary outputs a guess $g$, and wins if this guess matches $\beta$.

\subsection{Anatomy of a security game}\label{sec:reading-games}
Security definitions in cryptography are typically expressed as games between a \emph{challenger} (who runs the cryptographic scheme honestly) and an \emph{adversary} (who tries to break the scheme). Understanding how to read these game definitions is essential for the technical sections that follow.

A security game has several standard components:

\paragraph{Setup phase.} The challenger initializes the cryptographic scheme, typically by generating keys. For example, $(\pk, \sk, \mB, \mC) \leftarrow \Setup$ generates a key pair and system parameters. The adversary usually receives the public key but not the private key.

\paragraph{Adversary queries.} The adversary interacts with the challenger, often through an \emph{oracle} that provides specific capabilities. For instance, an encryption oracle allows the adversary to obtain encryptions of chosen messages. We write $\adv^\oracle$ to denote adversary $\adv$ with access to oracle $\oracle$.

\paragraph{Challenge phase.} The adversary attempts to accomplish a task that should be infeasible if the scheme is secure. This often involves distinguishing between two possibilities (such as which of two messages was encrypted) or forging something (such as a valid signature).

\paragraph{Output and winning condition.} The game returns a Boolean value: $\top$ (true) if the adversary wins, $\bot$ (false) otherwise. The winning condition is specified as a predicate that must hold for the adversary to win.

\paragraph{Example: Reading the IND game.}
Consider the indistinguishability game from Section~\ref{sec:indgame}. Let us walk through it line by line:

\begin{enumerate}
\item \emph{``The adversary $\adv$ outputs two messages $m_0$ and $m_1$''}: The adversary chooses which two messages to be challenged on. This models a chosen-plaintext attack.

\item \emph{``A key pair $\pk$ and $\sk$ is generated''}: The challenger runs key generation honestly.

\item \emph{``A random uniform bit $\beta$ is generated''}: This is the secret the adversary must guess. The adversary wins by determining $\beta$.

\item \emph{``Ciphertext $c$ is computed as $c \leftarrow \EncSymb(\pk, m_\beta)$''}: The challenger encrypts one of the two messages, chosen by the random bit.

\item \emph{``Ciphertext $c$ is given to the adversary''}: The adversary sees only the ciphertext, not $\beta$.

\item \emph{``The adversary outputs a bit $\beta'$''}: The adversary's guess for which message was encrypted.

\item \emph{``Output $\top$ if $\beta' = \beta$''}: The adversary wins if they guessed correctly.
\end{enumerate}

The scheme is secure if no efficient adversary can win with probability significantly better than $1/2$ (random guessing).

\paragraph{Bounded winning probability.}
Security definitions typically require that the adversary's success probability is bounded:
$$\Succ(\mathsf{Game}(\cdot)) \leq \frac{1}{2} + \neglK$$

The term $\frac{1}{2}$ represents random guessing (for bit-guessing games). The term $\neglK$ is a negligible function, meaning the adversary's advantage over guessing vanishes as the security parameter grows. For games where random success is not $1/2$ (such as games where the adversary wins only by achieving something specific), the bound might instead be $\Succ(\mathsf{Game}(\cdot)) \leq \neglK$.

\subsection{Security proofs by reduction}\label{sec:reductions}
The standard technique for proving that a cryptographic construction is secure is \emph{proof by reduction}. The idea is to show that if an adversary $\adv$ can break the construction, then we can use $\adv$ as a subroutine to build another adversary $\Adv$ that breaks some underlying assumption believed to be hard.

\paragraph{Structure of a reduction proof.}
A reduction proof typically proceeds as follows:
\begin{enumerate}
\item \emph{Assume the contrary}: Suppose the construction is insecure, meaning there exists an efficient adversary $\adv$ that wins the security game with non-negligible advantage.

\item \emph{Construct a new adversary}: Build adversary $\Adv$ that uses $\adv$ as a black box. Adversary $\Adv$ plays the role of challenger to $\adv$ while itself being a challenger in a different game.

\item \emph{Simulation}: Show that $\Adv$ can simulate the environment that $\adv$ expects. When $\adv$ makes oracle queries, $\Adv$ must respond in a way that is indistinguishable from the real game.

\item \emph{Extraction}: When $\adv$ succeeds in its task, show how $\Adv$ can extract a solution to the underlying hard problem.

\item \emph{Conclude}: Since breaking the underlying assumption is believed to be hard, and $\Adv$'s success probability is related to $\adv$'s, we conclude that $\adv$ cannot have non-negligible advantage.
\end{enumerate}

\paragraph{Example sketch.}
To prove that election scheme $\encToVote$ satisfies ballot secrecy (Proposition~\ref{prop:encToVoteBS}), we show: if adversary $\adv$ breaks ballot secrecy, then we can construct adversary $\Adv$ that breaks IND-PA0 of the encryption scheme $\Pi$.

Adversary $\Adv$ works as follows:
\begin{itemize}
\item $\Adv$ receives a public key $\pk$ from its own challenger and forwards it to $\adv$.
\item When $\adv$ outputs votes $(v_0, v_1)$, adversary $\Adv$ forwards these to its challenger as the two messages.
\item $\Adv$ receives a challenge ciphertext $c$ (an encryption of either $v_0$ or $v_1$) and gives it to $\adv$ as the ballot.
\item $\Adv$ simulates the rest of the ballot secrecy game for $\adv$.
\item When $\adv$ outputs its guess $g$, adversary $\Adv$ outputs $g$ as its own guess.
\end{itemize}

If $\adv$ can distinguish which vote was encrypted, then $\Adv$ can distinguish which message was encrypted, contradicting the security of $\Pi$. The key insight is that $\Adv$ perfectly simulates the ballot secrecy game for $\adv$, so $\adv$'s success probability transfers to $\Adv$.

\paragraph{Why reductions matter.}
Reduction proofs provide \emph{provable security}: the security of a complex construction (like an election scheme) is reduced to the security of simpler, well-studied primitives (like encryption schemes). Rather than trying to prove absolute security (which is generally impossible), we show that breaking the construction is \emph{at least as hard as} breaking the underlying primitive.

\begin{figure}
\centering
\begin{tikzpicture}[thick]
\tikzset{
    challenger/.style={rectangle, rounded corners, draw=black, very thick, inner sep=0.8em, minimum width=5cm, align=center},
    adversary/.style={rectangle, rounded corners, draw=black, thick, inner sep=0.6em, minimum width=3.5cm, align=center},
    arrowred/.style={-{Latex[length=3mm,width=1.5mm]}, thick},
    labelstyle/.style={font=\footnotesize}
}

\node[challenger] (chall) at (0, 3.5) {Challenger\\{\footnotesize (security game for primitive)}};

\node[adversary, minimum height=3.5cm, minimum width=5.5cm, dashed] (advB) at (0, -0.5) {};
\node[above] at (advB.north) {Reduction $\mathcal{B}$};

\node[adversary] (advA) at (0, -0.8) {Adversary $\mathcal{A}$\\{\footnotesize (attacks construction)}};

\node[labelstyle, text=gray] at (0, 0.8) {simulates construction};

\draw[arrowred] ([xshift=-1.4cm]chall.south) -- ([xshift=-1.4cm]advB.north);
\node[labelstyle, left=0.1cm] at (-1.4, 2.3) {challenge};

\draw[arrowred] ([xshift=1.4cm]advB.north) -- ([xshift=1.4cm]chall.south);
\node[labelstyle, right=0.1cm] at (1.4, 2.3) {response};

\end{tikzpicture}
\caption{Structure of a security reduction. To prove a construction secure, we assume an adversary $\mathcal{A}$ can break it and build a reduction $\mathcal{B}$ that uses $\mathcal{A}$ as a subroutine. The reduction $\mathcal{B}$ plays the security game for the underlying primitive while simulating the construction for $\mathcal{A}$. If $\mathcal{A}$ succeeds, $\mathcal{B}$ uses this to win its own game. Since the primitive is assumed secure, no efficient $\mathcal{A}$ can exist.}
\label{fig:reduction}
\end{figure}

Figure~\ref{fig:reduction} illustrates the structure of a reduction proof.

\subsection{Further reading on game-based cryptography}

Bellare and Rogaway have made foundational contributions to game-based security definitions. Their work on relations between cryptographic definitions~\cite{Bellare98:RelationsCryptoDefs} clarifies how different security notions (such as IND-CPA, IND-CCA, and non-malleability) relate to one another. Their treatment of the random oracle model~\cite{Bellare93:RandomOracles} provides an important idealization used in many security proofs.

For readers interested in the application of game-based techniques to voting systems specifically, Smyth, Frink, and Clarkson~\cite{Smyth15:ElectionVerifiability} provide game-based definitions of verifiability, while Smyth~\cite{Smyth15:BallotSecrecyFull} develops game-based definitions of ballot secrecy that we use in later sections.

The edited volume by Lindell~\cite{lindell2017tutorials} contains accessible tutorials on various aspects of modern cryptography, including both game-based and simulation-based approaches. This collection is particularly useful for understanding how these two paradigms complement each other.

\subsection{Games for voting security}
In subsequent sections, we use different styles of security games to capture different properties of election schemes. \emph{Verifiability} (Section~\ref{sec:verifiability}) uses games that task an adversary to reach a bad state---producing a fraudulent tally that passes verification, or causing ballots to collide. These are \emph{reachability} games: the scheme is secure if no efficient adversary can reach the forbidden state. \emph{Ballot secrecy} (Section~\ref{sec:secrecy}) uses \emph{indistinguishability} games, where the adversary must distinguish between two scenarios (such as which of two votes was encrypted). The scheme is secure if the adversary cannot do better than guessing. Recognizing which style of game applies to which property helps when reading the formal definitions that follow.


\section{Notation}\label{sec:notation-primer}

Having introduced cryptographic primitives and security games, we now consolidate the notation used throughout this \manuscript{}. This section serves as a reference for describing cryptographic games, election schemes, and security proofs; readers familiar with this material may skip ahead and return here as needed.

\subsection{Algorithms and randomness}

Formally, games and cryptographic constructions are probabilistic (or randomized) algorithms that may output Booleans. We write 
$$A(x_1,\dots,x_n; r)$$ 
to denote the output of a probabilistic algorithm $A$, where $x_1,\dots,x_n$ are the inputs and $r$ is a collection of independent coins that are each chosen uniformly at random from a suitable space of coin values. We often adopt the shorthand $A(x_1,\dots,x_n)$ to denote $A(x_1,\dots,x_n;r)$, leaving the randomness implicit.

\subsection{Assignments}

To denote an assignment of $T$ to $x$, we write $x \leftarrow T$, meaning $x$ takes on the value of $T$. For a random assignment, we write $x\leftarrow_R S$ to express that $x$ is chosen uniformly at random from a finite set $S$. We write $x,x'\leftarrow_R S$ for the independent random assignments $x\leftarrow_R S$ and $x'\leftarrow_R S$.

We write $\mathrm{Pr}(B)$ to denote the probability of a random event $B$ happening.

\subsection{Vectors and indexing}

Let $x[i]$ denote the $i$-th component of vector $x$, while $|x|$ denotes the length of vector $x$. To express an assignment with vectors, we write 
$(x_1,\allowbreak\dots,\allowbreak x_{|T|}) \leftarrow T$ for 
$x \leftarrow T;\allowbreak x_1\leftarrow x[1];\allowbreak \dots;\allowbreak x_{|T|}\leftarrow x[|T|]$, 
when $T$ is a vector.

\subsection{Logic notation}

When defining games and algorithms, we use standard logic notation to write predicates (or Boolean-valued functions), which return true or false. A predicate is a mathematical statement that tests whether something is true or false. You can think of it as a yes-or-no question applied to particular values. For example, the predicate $(n\leq m)$ returns true if the inequality is satisfied for $n$ and $m$. 

We combine predicates with the standard logical operators: \emph{conjunction} $p \wedge q$ (true when both hold), \emph{disjunction} $p \vee q$ (true when at least one holds), \emph{negation} $\neg p$ (true when $p$ is false), and \emph{implication} $p \Rightarrow q$ (if $p$ then $q$). For example, we might write $(a\neq b) \wedge (n\leq m)$ or $(a\neq b) \vee (n\leq m)$ to combine conditions.

Predicates are also used with quantifiers to express counting conditions. For instance, the notation $(\exists^{=\ell} x : P(x))$ means that exactly $\ell$ elements satisfy the property $P$. In other words, there are precisely $\ell$ values of $x$ (no more, no fewer) for which the test $P(x)$ returns true. This counting notation appears throughout our formal definitions of election properties.

\subsection{Quantifiers}

The \emph{universal quantifier} $\forall$ (``for all'') and the \emph{existential quantifier} $\exists$ (``there exists'') have their standard meanings: $\forall x \in S : P(x)$ asserts that property $P$ holds for every element $x$ in set $S$, while $\exists x \in S : P(x)$ asserts that $P$ holds for at least one such element.

We also use a \emph{counting quantifier} $\exists^{=\ell} x \in S : P(x)$ to express that \emph{exactly} $\ell$ distinct elements $x$ in $S$ satisfy property $P(x)$. (More specifically, the notation $\exists^{=\ell}$ denotes a generalized quantifier; for more details on first-order logic, see, for example, the introduction by Schweikardt~\cite{Schweikardt05}. Variable $x$ is bounded by the predicate, while the integer $\ell$ is free.) This quantifier appears in our definition of correct election outcomes; for example, $\exists^{=\ell} b \in \bb \setminus \{\bot\} : \exists r : b = \Vote[\pk, v, \nC, \kk; r]$ expresses that exactly $\ell$ ballots on the bulletin board are valid encryptions of vote $v$.

\subsection{Set notation}

We use standard set-theoretic notation: $x \in S$ denotes membership; $S \subseteq T$ denotes that $S$ is a subset of $T$; $S \cup T$, $S \cap T$, and $S \setminus T$ denote union, intersection, and set difference respectively; $|S|$ denotes cardinality; and $\emptyset$ denotes the empty set. Set-builder notation $\{x \in S \mid P(x)\}$ denotes the set of all elements $x$ in $S$ satisfying property $P$.

\subsection{Tuples and parsing}

A \emph{tuple} is an ordered collection of elements, written $(x_1, x_2, \ldots, x_n)$. Unlike sets, tuples preserve order and allow duplicates. We say we \emph{parse} an object as a tuple when we interpret it as having component parts---for example, parsing $\pk'$ as a pair $(\pk, \votespace)$ means interpreting $\pk'$ as consisting of a public key $\pk$ and a message space $\votespace$. Parsing may fail if the object does not have the expected structure, in which case algorithms typically output the error symbol $\bot$.

With these notational conventions established, we now turn to the formal description of voting systems and their security properties.

\section{Overview of an election scheme }\label{sec:election}

Now that we have the cryptographic building blocks in hand, let us revisit the election scheme introduced in Section~\ref{sec:election-preview} in more detail. We will soon introduce suitable properties for any election scheme at a more technical level. For now, we expand on one example of an election scheme, which we then use to introduce more technical concepts and definitions. Ultimately, we want to see if this and other proposed election schemes satisfies certain properties, with a focus on ballot secrecy and verifiability. 

Recall that the election scheme proposed by Smyth, Frink and  Clarkson~\cite{Smyth15:ElectionVerifiability} uses four algorithms or primitives, denoted by  $\SetupSymb$, $\VoteSymb$, $\TallySymb$, and $\VerifySymb$, as illustrated in Figure~\ref{figure-simplevoting}. The \emph{bulletin board} is where encrypted votes are stored until the election ends. We now describe the function of each primitive in more detail and then in the next section give more specific details. 

\begin{itemize}
\item $\SetupSymb$ randomly creates a public and private key pair, where the voters will use the public key and the election authority will use the private key. For illustration purposes, we only consider one key pair, but we can easily have two or more key pairs, and each candidate (such as a political party) in the election keeps their own private key secret. $\SetupSymb$ is simply an asymmetric encryption scheme that is suitably homomorphic.

\item $\VoteSymb$ allows the voters to vote using the public key. Each voter's vote is encrypted and the $\VoteSymb$ primitive must use homomorphic asymmetric encryption with  vanishingly small probability of collisions, meaning each voter has a unique but secret encrypted vote. If more than one key pair was created by the primitive $\SetupSymb$, then each voter must use the corresponding public keys to vote.

\item $\TallySymb$ adds up the encrypted votes, which is possible with a homomorphic asymmetric encryption scheme. This primitive  outputs the election outcome as well as proof for verifying that the election outcome is the true outcome based on the encrypted votes produced by the voters using $\VoteSymb$. $\TallySymb$ does not need the private key, originally generated by $\SetupSymb$, to tally the votes, but it \emph{does} need the private key to produce a proof.

\item $\VerifySymb$ uses the public key or keys from the primitive $\SetupSymb$ and a proof generated by the primitive $\TallySymb$ to verify the election outcome, also from $\TallySymb$, is correct. $\VerifySymb$ outputs \emph{yes} or \emph{no} to answer the question whether election integrity has been kept. 
\end{itemize}
Fortunately these four primitives can be represented as four probabilistic polynomial-time algorithms. We have omitted a couple of system parameters that flow between the primitives, such as the maximum number of ballots, which can easily be made large enough to eliminate all limitations, but these are not crucial for security purposes. We will soon formally detail the specifics of these four primitives, but first we introduce some notation.

\section{Election scheme syntax}\label{sec:syntax} 
In Section~\ref{sec:election} we gave an informal overview of a proposed election scheme, illustrated in Figure~\ref{figure-simplevoting}. The \emph{tallier} is the entity (or group of entities) responsible for generating election keys and computing the final tally from collected ballots. In the schemes we consider, the tallier runs the $\SetupSymb$ and $\TallySymb$ algorithms and possesses the secret key needed for decryption. More sophisticated schemes distribute the tallier's role among multiple parties using threshold encryption (Section~\ref{sec:threshold}), but we defer such extensions to future work.

First, the tallier generates a key pair using $\SetupSymb$. Secondly, each voter constructs and casts a ballot for their vote using $\VoteSymb$. All the cast ballots are collected and recorded on a bulletin board. Third, the tallier  tallies the collected ballots and announces an outcome as a tally of the votes by using $\TallySymb$. The elected candidate is derived from the vote count by using some voting system, such as the  candidate with the most votes. (We note for the first-past-the-post voting systems,  Smyth~\cite{2017-FPTP-suffices-for-ranked-voting} has shown that the syntax can model ranked-choice voting systems too.) Finally, voters and other interested parties check that the election outcome corresponds to the votes expressed in collected ballots by using $\VerifySymb$. We now give specific syntax that captures the four steps of this voting system.
  
\begin{definition}[Election scheme~\cite{Smyth15:ElectionVerifiability}] \label{def:election}
An \emph{election scheme} is a tuple of probabilistic polynomial-time algorithms 
$(\SetupSymb,\allowbreak
	\VoteSymb,\allowbreak\TallySymb,\allowbreak\VerifySymb)$ such that:
	\begin{description}
\item $\SetupSymb$, denoted $(\pk,\sk,\mB,\mC) \leftarrow \Setup$,
is run by the tallier. The algorithm takes a security parameter $\kk$ as input and outputs a key pair $(\pk,\sk)$, a maximum number of ballots $\mB$, 
and a maximum number of candidates $\mC$.

\item  $\VoteSymb$, denoted $b\leftarrow\Vote$, is run by voters. The algorithm
takes as input a public key $\pk$, a voter's vote $v$, some number of candidates 
$\nC$, and a security parameter $\kk$. Vote $v$ should be selected from a
sequence  $1,\dots,\nC$ of candidates.
The algorithm outputs an encrypted ballot $b$ or error symbol~$\perp$. 

\item  $\TallySymb$, denoted $(\outcome, \tpf)\leftarrow\Tally$, is run by the 
tallier. The algorithm takes as input a private key $\sk$, a bulletin board $\bb$ (of encrypted ballots), 
some number of candidates $\nC$, and a security parameter $\kk$. The algorithm outputs an election outcome $\outcome$ and a non-interactive tallying
proof $\tpf$. The election outcome must be a vector of length $\nC$ 
and each index $v$ of that vector should indicate the number of votes for candidate $v$. 
Moreover, the tallying proof should demonstrate that the outcome corresponds to votes 
expressed in ballots on the bulletin board.

\item $\VerifySymb$, denoted $\auditoutcome \leftarrow \Verify$, is run to audit an election.
The algorithm takes as input a public key $\pk$, a bulletin board $\bb$,
some number of candidates $\nC$, an election outcome $\outcome$, a tallying proof 
$\tpf$, and a security parameter $\kk$. The algorithm outputs a bit $\auditoutcome$, which is $1$  
if the 
outcome should be accepted and $0$ otherwise. We require the algorithm to be 
deterministic.

\end{description}
\end{definition}
\noindent
Beyond having proper syntax, election schemes must satisfy \emph{correctness}: the probability of an incorrect election outcome should be negligible. We now state this notion formally.
\begin{definition}[Election correctness]
An election scheme satisfies correctness if there exists a negligible function $\negl$, such that for all security parameters $\kk$, integers 
$\nB$ and $\nC$, and votes $v_1,\dots,v_{\nB}\in\{1,\dots,\nC\}$, it holds that, 
given a zero-filled vector $\outcome$ of length $\nC$, such that:\\
\LinesNotNumbered
{\upshape
\noindent $\Pr[(\pk,\sk,\mB, \mC) \leftarrow \Setup$;\\
\begin{algorithm}[H]
	\For{$1 \leq i \leq \nB$}{
		$b_i \leftarrow \Vote[\pk,v_i,\nC,\kk]$;\\
		$\outcome[v_i] \leftarrow \outcome[v_i] + 1$;
	}
	$(\outcome',\tpf) \leftarrow \Tally[\sk,\{b_1,\dots,b_{\nB}\},\nC, \kk]:$
	$ ((\nB \leq \mB) \wedge (\nC\leq \mC)) \Rightarrow (\outcome = \outcome')] > 1 - \neglK$.
\end{algorithm} 
}
\end{definition}
\noindent
The above definition of correctness means that, given an election setup created by $\Setup$,
for all $\nB$ voters, each voter $i$ casts their vote $v_i$ by using the $\Vote$ primitive, and then the votes are summed up, giving an outcome $\outcome$. Then the $\Tally$ primitive also returns a vote count $\outcome'$ and a proof $\tpf$. Provided the number of ballots $\nB$ is less than the ballot maximum $\mB$ and the number of candidates $\nC$ is less than the candidate maximum $\mC$, we then say the election scheme is correct if the probability of two outcomes  $\outcome$ and $\outcome'$  being the same is greater than one minus a negligible function of the security parameter, meaning the probability of the two outcomes being different is negligible.

The syntax bounds the number of ballots $\mB$ and candidates $\mC$ to broaden the scope of the correctness definition. For example, the voting scheme Helios requires $\mB$ and $\mC$ to be less than or equal to the size of the underlying encryption scheme's message space, and represents votes as integers, rather than alphanumeric strings, for brevity. 

This syntax allows us to model voting systems, while correctness ensures that election outcomes correspond to collected votes when constructed and tallied as prescribed. We will use our syntax to express verifiability and secrecy properties of election schemes. Then we will model and study implemented electronic voting systems.

Figure~\ref{figure-voting} illustrates the data flow between the four algorithms of Definition~\ref{def:election}.

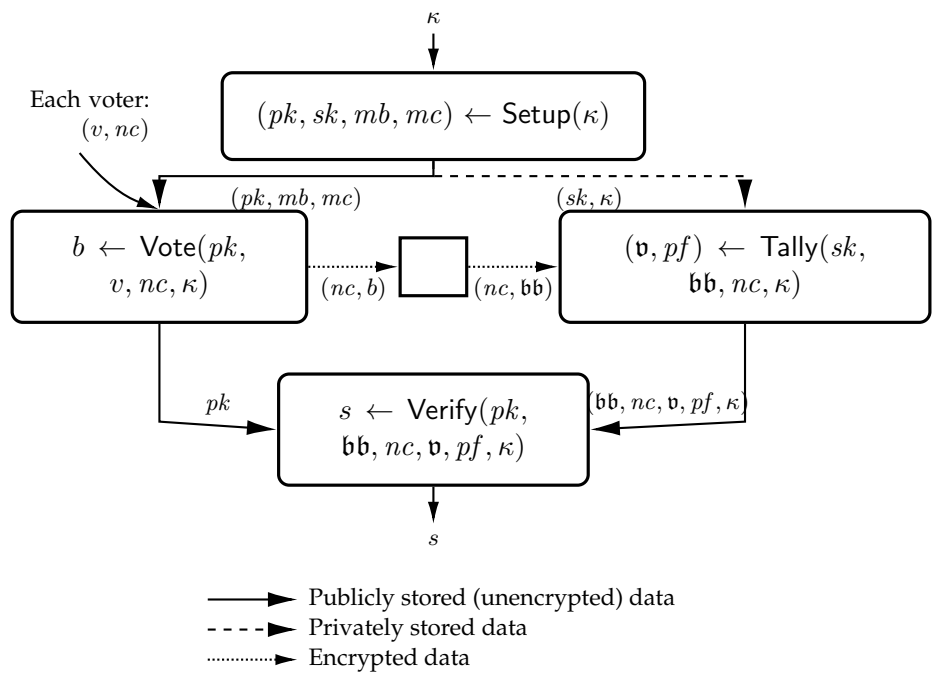
\begin{figure}
\centering
\begin{tikzpicture}[thick]  
\tikzset{
    primitive/.style={rectangle, rounded corners, draw=black, very thick, inner sep=0.8em, minimum size=3em, text centered},
    votestorage/.style={rectangle, draw=black, very thick, inner sep=0.5em, minimum size=2em, text centered},       
    arrowbig/.style={-{Latex[length=4mm,width=2mm]}, thick},
    arrowsmall/.style={-{Latex[length=3mm,width=1.5mm]}, thick},
    mylabel/.style={font=\footnotesize} 
}

\node[primitive, text width=5cm] (setup) {$(\pk,\sk,\mB,\mC) \leftarrow \Setup$};
\node[votestorage, below=1cm of setup, text width=0.5cm] (bboard) {};
\node[primitive, left=1.2cm of bboard, text width=3.3cm] (vote) {$b\leftarrow\Vote$};
\node[primitive, right=1.2cm of bboard, text width=4.3cm] (tally) {$(\outcome, \tpf)\leftarrow\Tally$};
\node[primitive, below=1cm of bboard, text width=3.5cm] (verify) {$\auditoutcome \leftarrow \Verify$};

\draw[arrowbig] (setup.south) -- ++(0,-0.2) -| node[pos=0.25, below, mylabel] {$(\pk, \mB, \mC)$} (vote.north);
\draw[arrowbig, dashed] (setup.south) -- ++(0,-0.2) -| node[pos=0.25, below, mylabel] {$(\sk,\kk)$} (tally.north);

\draw[arrowsmall, densely dotted] (vote.east) -- node[below, mylabel] {$(\nC,b)$} (bboard.west); 
\draw[arrowsmall, densely dotted] (bboard.east) -- node[below, mylabel] {$(\nC,\bb)$} (tally.west); 

\draw[arrowbig] (vote.south) -- ++(0,-1.3) -- (verify.west) node[midway, above, mylabel] {$\pk$};

\draw[arrowbig] (tally.south) -- ++(0,-1.3) -- (verify.east) node[midway, above, mylabel] {$(\bb,\nC,\outcome,\tpf,\kk)$};

\node[above=0.5cm of setup, mylabel] (kappa) {$\kk$};
\node[left=0.8cm of setup, mylabel, text width=2.5cm, align=right] (voters) {Each voter: $(v,\nC)$};
\node[below=0.5cm of verify, mylabel] (s) {$\auditoutcome$};

\draw[arrowsmall] (kappa) -- (setup.north); 
\draw[arrowsmall] (verify.south) -- (s); 
\draw[arrowsmall, bend right=15] (voters) to (vote.north); 

\node[below=1.5cm of verify, anchor=north] (legend) {};
\draw[arrowbig] ([xshift=-3cm]legend.north) -- ++(1.2,0) node[right, font=\footnotesize] {Publicly stored (unencrypted) data};
\draw[arrowbig, dashed] ([xshift=-3cm, yshift=-0.4cm]legend.north) -- ++(1.2,0) node[right, font=\footnotesize] {Privately stored data};
\draw[arrowsmall, densely dotted] ([xshift=-3cm, yshift=-0.8cm]legend.north) -- ++(1.2,0) node[right, font=\footnotesize] {Encrypted data}; 
\end{tikzpicture} 
\medskip
\caption{Detailed view of the election scheme. The four algorithms $\SetupSymb$, $\VoteSymb$, $\TallySymb$, and $\VerifySymb$ are defined in Definition~\ref{def:election}. Dashed arrows indicate secret data (the private key $\sk$); dotted arrows indicate encrypted ballots.}
\label{figure-voting}
\end{figure}	
\section{Voting}\label{sec:voting}

Now that we have formally defined election schemes and their syntax, we can examine the threats these systems face. Electronic voting systems are vulnerable to a range of security attacks that can undermine either ballot secrecy or election integrity. Understanding these attacks is essential for designing robust systems and for appreciating why the formal security properties we define in subsequent sections matter. We categorize attacks by their primary target and mechanism.

\subsection{Attacks}

\subsubsection{Coercion attacks}
In a \emph{coercion attack}~\cite{JCJ05}, an adversary forces a voter to cast a particular vote against their will. The coercer may be present during voting, may demand proof of how the voter voted, or may threaten consequences for non-compliance. Coercion attacks are particularly concerning because they can occur outside the voting system itself, making them difficult to prevent through technical means alone.

Several variants of coercion attacks exist:
\begin{itemize}
\item \emph{Forced abstention}: The coercer prevents a voter from participating in the election entirely, for example by confiscating credentials or physically preventing access to voting.

\item \emph{Randomization attack}: The coercer instructs the voter to use a specific randomization method when choosing their vote, such as flipping a coin. This allows the coercer to manipulate election outcomes statistically without knowing individual votes.

\item \emph{Simulation attack}: The coercer obtains the voter's credentials (such as private keys) after registration but before voting, then casts a vote on behalf of the voter.
\end{itemize}

\subsubsection{Vote-selling}
Closely related to coercion is \emph{vote-selling}, where a voter willingly sells their vote to a buyer. Unlike coercion, vote-selling involves a voluntary transaction, but it similarly undermines the principle that votes should reflect genuine preferences. Vote-selling requires a mechanism for the buyer to verify how the seller voted, which creates tension with ballot secrecy: strong secrecy makes verification difficult, while receipts that enable verification also enable vote-selling.

\subsubsection{Ballot manipulation attacks}
Several attacks target the ballots themselves:

\begin{itemize}
\item \emph{Ballot copying attack}: An adversary copies another voter's encrypted ballot and submits it as their own. If successful, this duplicates the copied voter's choice without the adversary knowing what that choice was. This attack exploits malleable encryption schemes.

\item \emph{Modification attack}: An adversary intercepts a ballot in transit and modifies it, changing the encoded vote. This requires the ability to meaningfully alter ciphertexts, which non-malleable encryption prevents.

\item \emph{Clash attack}~\cite{Kusters12:ClashAttacks}: Two or more voters' ballots produce the same identifier or hash value, making it impossible to distinguish them on the bulletin board. An adversary can exploit collisions to substitute ballots or cause confusion during verification.
\end{itemize}

\subsubsection{Infrastructure attacks}
\begin{itemize}
\item \emph{Man-in-the-middle attack}: An adversary positions themselves between the voter and the voting server, intercepting and potentially modifying communications in both directions. The adversary can alter votes, steal credentials, or present false information to voters.

\item \emph{Brute-force attack}: When the space of possible votes is small (as in most elections with few candidates), an adversary who obtains encrypted ballots may attempt to decrypt them by encrypting all possible votes and comparing the results. Probabilistic encryption defends against this by ensuring identical plaintexts produce different ciphertexts.
\end{itemize}

\subsection{Countermeasures}

Various countermeasures have been developed to defend against the attacks described above.

\subsubsection{Ballot weeding}
\emph{Weeding} is a technique where duplicate or malformed ballots are removed from the bulletin board before tallying. This defends against ballot copying attacks by ensuring that each valid ballot appears exactly once. Weeding can be performed by checking that each voter credential is used at most once, or by detecting duplicate ciphertexts.

\subsubsection{Non-malleable encryption}
Using \emph{non-malleable encryption}~\cite{katz2014introduction} prevents adversaries from meaningfully modifying ciphertexts. If an adversary cannot transform a ciphertext encrypting vote $v$ into a ciphertext encrypting a related vote $v'$, then modification and ballot copying attacks become infeasible.

\subsubsection{Zero-knowledge proofs}
Requiring voters to provide \emph{zero-knowledge proofs} that their ballots are well-formed (encrypting a valid vote) prevents submission of malformed ballots that might exploit vulnerabilities in the tallying process.

\subsubsection{Receipt-freeness}
A voting system is \emph{receipt-free}~\cite{BT94:ReceiptFreeVoting} if voters cannot prove to a third party how they voted. Receipt-freeness defends against both vote-selling (the buyer cannot verify the seller's compliance) and coercion (the voter can claim to have voted as instructed while actually voting freely). Achieving receipt-freeness while maintaining verifiability is challenging, as both properties involve what voters can prove.

\subsubsection{Coercion-resistance}
Stronger than receipt-freeness, \emph{coercion-resistance}~\cite{JCJ05} ensures that even if a coercer actively participates in the voting process (for example, by providing credentials or observing the voter), the voter can still cast their intended vote. This typically requires mechanisms for voters to cast ``fake'' votes that satisfy the coercer but are later replaced or ignored.

\begin{figure}[t]
\centering
\small
\begin{tabular}{@{}p{3.2cm}p{4.2cm}p{4.5cm}@{}}
\hline
\textbf{Attack} & \textbf{Property Violated} & \textbf{Countermeasure} \\
\hline
\multicolumn{3}{@{}l}{\emph{Coercion attacks}} \\[2pt]
\quad Coercion & Coercion-resistance & Coercion-resistant protocols \\
\quad Forced abstention & Coercion-resistance & Credential recovery; deniable voting \\
\quad Randomization & Coercion-resistance & Coercion-resistant protocols \\
\quad Simulation & Coercion-resistance & Credential binding \\[4pt]
\hline
Vote-selling & Receipt-freeness & Receipt-free protocols \\[4pt]
\hline
\multicolumn{3}{@{}l}{\emph{Ballot manipulation attacks}} \\[2pt]
\quad Ballot copying & Ballot secrecy; Non-malleability & Ballot weeding; NM-CPA encryption \\
\quad Modification & Universal verifiability & Non-malleable encryption; ZK proofs \\
\quad Clash attack & Individual verifiability & Collision-resistant hashing \\[4pt]
\hline
\multicolumn{3}{@{}l}{\emph{Infrastructure attacks}} \\[2pt]
\quad Man-in-the-middle & Ballot secrecy; Verifiability & Authenticated channels; ZK proofs \\
\quad Brute-force & Ballot secrecy & Probabilistic encryption \\[2pt]
\hline
\end{tabular}
\caption{Taxonomy of attacks on electronic voting systems, the security properties they violate, and corresponding countermeasures. Properties in the middle column correspond to formal definitions: Ballot secrecy (Section~\ref{sec:secrecy:def}), Universal and Individual verifiability (Section~\ref{sec:verifiability}), and Non-malleability (Section~\ref{sec:security-hierarchy}). Receipt-freeness and coercion-resistance are discussed in Section~\ref{sec:secrecy:outlook}.}
\label{fig:attacks-table}
\end{figure}

\subsection{Attacks and countermeasures summary}
Figure~\ref{fig:attacks-table} summarizes the relationships between attacks, the security properties they violate, and their countermeasures. This taxonomy illustrates why election schemes require multiple security properties working in concert: defending against one class of attacks may leave the system vulnerable to others.

A recurring suggestion is to use blockchain technology to secure elections, for example by recording ballots on a distributed ledger to prevent tampering. However, blockchains do not address the fundamental challenges of ballot secrecy and software independence. As Rivest remarked at the RSA Conference in 2020:
\begin{quote}
``Blockchain is the wrong security technology for voting. I like to think of it as bringing a combination lock to a kitchen fire.''\\
\mbox{}\hfill---Ronald Rivest~\cite{Rivest20:RSAPanel}
\end{quote}
Park, Specter, Narula and Rivest~\cite{Park20:BlockchainVoting} argue that blockchain-based voting would greatly increase the risk of undetectable, nation-scale election failures, and that any convenience gains would come at the cost of losing meaningful assurance that votes have been counted as cast.

\section{Verifiability}\label{sec:verifiability}
To ensure that election outcomes correctly correspond to the votes expressed in collected ballots, the Australian election system relies upon monitoring. 
The mere depositing of ballots into ballot boxes is enough to ensure that they are collected. 
(Though several years ago the state of Western Australia lost more than a thousand ballot papers in a federal election.)

This approach is in direct contrast with electronic election schemes, which compute election outcomes in a manner that should not 
be monitored. Otherwise such monitoring would reveal the tallier's private key, 
compromising the ballot secrecy. Furthermore, just casting a ballot is insufficient
to ensure it is collected, because an adversary may discard or modify ballots. 
Nevertheless, election schemes generate tallying proofs to provide 
evidence that the election outcomes are correctly computed, as well as allowing voters to check whether
their ballot are properly collected. These two concepts are formalized, respectively, by universal 
verifiability and individual verifiability.

Before formalizing these concepts, we briefly survey the broader landscape of verifiability concepts. Researchers have proposed different definitions of verifiability. For example, voting systems may satisfy a stronger notion of universal verifiability: \informalDefinitionUV\ 
that are authorized, except votes cast by the same voter. A related concept is \emph{unforgeability}: only voters can construct authorized ballots. (Previously Smyth~\cite{Smyth15:ElectionVerifiability} used the term \emph{eligibility verifiability} as a synonym for \emph{unforgeability}, but we prefer the latter.)

Verifiability is often decomposed into finer-grained properties~\cite{Cortier16:VerifiabilitySoK}, including \emph{cast-as-intended}, \emph{stored-as-cast}, and \emph{tallied-as-stored} verifiability, which together comprise \emph{end-to-end verifiability}. Beyond verifying that votes are counted correctly, \emph{eligibility verifiability} ensures only authorized voters can cast ballots. In this \manuscript{}, we focus on formalizing individual and universal verifiability as defined below.


\subsection{Universal verifiability}\label{sec:verifiability:uv}
Universal verifiability asserts that anyone must be able to check whether
an election outcome corresponds to the votes expressed in collected ballots. This can be expressed in terms of \emph{reachability}: whether a system, like a game, equipped with a set of rules, can reach a certain state. 

Since checks can be performed by algorithm $\VerifySymb$, it is sufficient that the algorithm accepts if and only if 
the outcome corresponds to votes expressed in the collected ballots. 
The \emph{only if} requirement is captured by \emph{soundness}, which requires 
algorithm $\VerifySymb$ to only accept correct outcomes, and the \emph{if} 
requirement is captured by completeness, which requires election outcomes 
produced by algorithm $\TallySymb$ to be accepted by algorithm $\VerifySymb$.

\paragraph{Soundness.}

Soundness captures the requirement that verification should reject incorrect outcomes. An adversary attempting to announce a fraudulent tally should be unable to produce a proof that passes verification. The formal definition tasks an adversary to produce any inputs, such as public key, bulletin board, outcome, and proof, that cause verification to accept an outcome that does not match the votes actually present in the ballots.

We formalize correct outcomes by using the function $\correcttally$. Recall that the counting quantifier $(\exists^{=\ell} x : P(x))$ expresses that exactly $\ell$ distinct values of $x$ satisfy property $P(x)$. 

Using this predicate, function $\correcttally$ is defined by
\begin{multline*}
\correcttally(\pk,\nC,\bb,\kk)[v] = \ell
\iff \\
  \exists^{=\ell} b\in \bb \setminus \{\bot\} : 
   \exists r : b=\Vote[\pk,\allowbreak v,\allowbreak\nC,\allowbreak\kk;\allowbreak r],
\end{multline*}
where $\correcttally(\pk,\nC,\bb,\kk)$ is a vector of length $\nC$ and 
$1\leq v \leq \nC$, and where $r$ denotes the random coins used by the probabilistic algorithm $\Vote$. We see that component $v$ of vector 
$
  \correcttally(\allowbreak\pk,\allowbreak \nC,\bb,\allowbreak\kk)
$ 
equals $\ell$ if and only if there exist $\ell$ ballots for vote $v$ in the bulletin 
board. The function requires that ballots be interpreted for only one candidate, 
which we can ensure with the concept of \emph{injectivity}.


\begin{definition}[$\Injectivity$~\cite{Smyth15:ElectionVerifiability,Smyth18:HeliosMixnetVerifiability}]
An election scheme $(\SetupSymb,\allowbreak\VoteSymb,\allowbreak\TallySymb,\allowbreak\VerSymb)$ satisfies $\Injectivity$, if for all probabilistic polynomial-time adversaries $\adv$, security parameters $\kk$
and computations 
$
  (\pk,\nC,v,v') \leftarrow \adv(\kk);
  b\leftarrow \Vote[\pk,v,\nC,\kk];\allowbreak
  b'\leftarrow \Vote[\pk,v',\nC,\kk]
$
such that $(v\neq v') \wedge( b\neq \bot )\wedge (b'\neq \bot)$, we have
$b\neq b'$.
\end{definition}

\noindent
We use a definition of injectivity that ensures that a ballot for vote $v$ can never
be interpreted for another vote $v'$, meaning the votes expressed in 
ballots correspond to unique outcomes.

Equipped with a notion of correct outcomes, we can formalize \emph{soundness}.
\begin{definition}[$\Soundness$~\cite{Smyth15:ElectionVerifiability}]\label{def:UVSoundness}
Let $\Gamma = (\SetupSymb,\allowbreak\VoteSymb,\allowbreak\TallySymb,\allowbreak\VerSymb)$
be an election scheme, $\adv$ be an adversary, $\kk$ be a security parameter, and 
$\SoundnessGame$ be the following game.

{\upshape
\begin{inlineexperiment}{$\SoundnessGame$}
$(\pk,\bb,\nC,\outcome,\tpf) \leftarrow \adv(\kk)$\label{alg:SoundnessGame:adv}\;
\Return~$(\VerSymb(\pk,\bb,\nC, \outcome,\tpf,\kk) = 1) \mathrel\wedge (\outcome \not= \correcttally(\pk,\nC,\bb, \kk))$\label{alg:SoundnessGame:ret}\;
\end{inlineexperiment}
}

\noindent
We say election scheme $\Gamma$ satisfies $\Soundness$, if $\Gamma$ satisfies injectivity and 
for all probabilistic polynomial-time adversaries $\adv$, there exists a 
negligible function $\negl$, such that for all security parameters $\kk$, 
we have $\Succ(\SoundnessGame) \leq \neglK$.
\end{definition}

\noindent
In Definition~\ref{def:UVSoundness}, soundness is defined as a game that tasks the adversary to compute inputs to algorithm $\VerifySymb$ (Line~\ref{alg:SoundnessGame:adv}), including an election outcome and some ballots, that cause 
the algorithm to accept when the outcome does not correspond to the votes expressed in 
those ballots (Line~\ref{alg:SoundnessGame:ret}). When algorithm $\VerifySymb$ only accepts election outcomes that correspond to votes expressed in collected ballots, we say that an election scheme satisfies $\Soundness$. 

\begin{mdframed}[style=guidebox]
\begin{guideline}
Verification must only accept outcomes that correspond to votes expressed in
collected ballots.
\end{guideline}
\end{mdframed}

\paragraph{Completeness.}

Completeness is the dual of soundness: it ensures that honestly-computed outcomes are always accepted. Without completeness, a correct election could be rejected, allowing a malicious party to void legitimate results. The formal definition tasks an adversary to produce a bulletin board for which the honestly-computed tally fails verification.

\begin{definition}[$\Completeness$~\cite{Smyth15:ElectionVerifiability}]\label{def:UVComplete}
\begin{sloppypar}
Let $\Gamma = (\SetupSymb,\allowbreak\VoteSymb,\allowbreak\TallySymb,\allowbreak\VerSymb)$ 
be an election scheme, $\adv$ be an adversary, $\kk$ be a security parameter, and 
$\CompletenessGame$ be the following game.
\end{sloppypar}

{\upshape
\begin{inlineexperiment}{$\CompletenessGame$}
  $(\pk,\allowbreak\sk,\allowbreak\mB,\allowbreak \mC) \leftarrow\allowbreak \Setup$\label{alg:CompletenessGame:setup}\;
  $(\bb,\nC)\leftarrow\adv(\pk,\kk)$\label{alg:CompletenessGame:adv}\;
  $(\outcome,\allowbreak\tpf) \leftarrow\allowbreak \Tally$\label{alg:CompletenessGame:tally}\;
  \Return $(\VerSymb(\pk,\allowbreak\bb,\allowbreak\nC,\allowbreak\outcome,\allowbreak\tpf,\allowbreak\kk) \not=\allowbreak 1) \wedge (|\bb| \leq \mB) \wedge\allowbreak (\nC\leq \mC)\allowbreak$\label{alg:CompletenessGame:ret}\;
\end{inlineexperiment}
}

\noindent
We say election scheme $\Gamma$ satisfies $\Completeness$, if for all probabilistic polynomial-time 
adversaries $\adv$, there exists a negligible function $\negl$, such that 
for all security parameters $\kk$, we have
$
\Succ(\CompletenessGame) \leq \neglK
$.
\end{definition}

\noindent
In Definition~\ref{def:UVComplete} we formalize completeness as a game that, given the key pair is computed by
algorithm $\SetupSymb$ (Line~\ref{alg:CompletenessGame:setup}), tasks the adversary to compute a bulletin board and some number of candidates (Line~\ref{alg:CompletenessGame:adv}) such that the corresponding election outcome computed by algorithm $\TallySymb$ (Line~\ref{alg:CompletenessGame:tally}) is rejected by algorithm
$\VerifySymb$ (Line~\ref{alg:CompletenessGame:ret}).

An election scheme satisfies $\Completeness$ when algorithm $\VerifySymb$ accepts 
outcomes computed by algorithm $\TallySymb$, for key pairs computed by algorithm $\SetupSymb$.
It follows that completeness implies an aspect of accountability. Indeed, if verification fails, then
the tallier is responsible for that failure, in particular, they must have incorrectly computed 
their key pair or the election outcome.

\begin{mdframed}[style=guidebox]
\begin{guideline}
Tallying must produce outcomes that will be accepted during verification.
\end{guideline}
\end{mdframed}

\noindent
We formalize universal verifiability by combining the above notions.

\begin{definition}[$\UV$~\cite{Smyth15:ElectionVerifiability,Smyth18:HeliosMixnetVerifiability}]\label{def:UV}
An election scheme $\Gamma$ satisfies $\UV$, if 
$\Soundness$ and $\Completeness$ are satisfied.
\end{definition}

\subsection{Individual verifiability}\label{sec:verifiability:iv}

Individual verifiability ensures that voters can confirm their ballot appears in the election record. Without this property, a malicious system could silently discard votes. The key challenge is that merely seeing a ballot on the bulletin board is insufficient. Rather, a voter must be confident it is \emph{their} ballot, not a coincidentally identical one constructed by someone else. This requires ballots to be uniquely identifiable.

Individual verifiability asserts that voter Alice must be able to check whether
her ballot is among those that have been collected. Given that ballots should be collected and recorded on a bulletin board, which must be available to everyone, then it is enough for voters to be able to check that their ballots (which they constructed) are on the bulletin board. This means that voters must be able to check that their ballots have not been omitted from
the bulletin board. Nevertheless, this alone is not sufficient, because the existence of
a ballot identical to an individual voter's ballot does not imply that the ballot constructed by the voter exists. In other words, if a ballot identical to voter Alice's exists, it does not mean Alice's vote exists, because Alice's ballot might have been constructed by another voter.

Consequently, individual verifiability requires that voters must be able to uniquely identify their ballots, so that the ballots have not collided. We now give a definition of individual verifiability that captures this requirement.
\begin{definition}[Individual verifiability~\cite{Smyth15:ElectionVerifiability}]\label{def:iv}
Let $\Gamma = (\SetupSymb,\allowbreak\VoteSymb,\allowbreak\TallySymb,\allowbreak\VerSymb)$ be an election scheme, $\adv$ be an adversary, $\kk$ be a security parameter, and $\IVGame$ be the following game.

{\upshape
\begin{inlineexperiment}{$\IVGame$}
  $(\pk, \nC,v,v') \leftarrow \adv(\kk)$\label{alg:IVGame:adv}\;
  $b\leftarrow\Vote[\pk,\nC,v,\kk]$\label{alg:IVGame:ballot}\;
  $b'\leftarrow\Vote[\pk,\nC,v',\kk]$\label{alg:IVGame:ballotB}\;
  \Return $(b = b') \wedge (b \not= {\perp} ) \wedge (b' \not= {\perp}) $\label{alg:IVGame:ret}\;
\end{inlineexperiment}
}

\noindent
We say election scheme $\Gamma$ satisfies $\IV$, if for all probabilistic polynomial-time 
adversaries $\adv$, there exists a negligible function $\negl$, such that 
for all security parameters $\kk$, the inequality $$\Succ(\IVGame)\leq \neglK \,,$$
holds.
\end{definition}

\noindent
In Definition~\ref{def:iv}, we formalize individual verifiability as a game that tasks the adversary to compute inputs to algorithm $\VoteSymb$ (Line~\ref{alg:IVGame:adv}) that cause the algorithm to output ballots (Lines~\ref{alg:IVGame:ballot} \&~\ref{alg:IVGame:ballotB})
that collide (Line~\ref{alg:IVGame:ret}).

When algorithm $\VoteSymb$ generates uniquely identifiable ballots, meaning ballots that do not collide, then we say an election scheme satisfies $\IV$.
Elaborating further, the respective concepts of correctness, individual verifiability and injectivity all hinge upon ballots not colliding, but they do so under different assumptions. 
Correctness requires that ballots do not collide, with overwhelming 
probability, for public keys computed by algorithm $\SetupSymb$; 
$\Injectivity$ requires that ballots for distinct votes never collide; 
and $\IV$ requires that ballots do not collide with overwhelming 
probability. Consequently, $\IV$ implies that ballots do not collide in the 
context of correctness. But $\IV$ and $\Injectivity$ are in a sense orthogonal to each other.
Specifically, $\IV$ allows collisions with negligible probability, whereas $\Injectivity$ allows collisions between ballots for the same vote.

The recurring requirement of no ballot collisions motivates another design guideline.
\begin{mdframed}[style=guidebox]
\begin{guideline}
Ballots must be distinct.
\end{guideline}
\end{mdframed}

\section{Ballot secrecy}\label{sec:secrecy}

As noted in Section~\ref{sec:intro}, we prefer the term \emph{ballot secrecy} over \emph{privacy} to avoid confusion with other privacy notions.

The Australian system relies upon uniform ballots to ensure voters' 
votes are not revealed. This uniformity ensures that ballots are 
indistinguishable during distribution, whereas the isolated polling 
booths ensures votes are not revealed while marking. Folded ballots are indistinguishable during collection and indistinguishability of markings can ensure votes are not revealed while tallying. In short, the Australian system derives ballot secrecy from physical characteristics
of the world. By comparison, electronic election schemes cannot rely on such physical
characteristics, so they must rely on cryptography to ensure voters' votes are not revealed. 
(Admittedly, some electronic voting systems do rely upon physical characteristics. For instance, MarkPledge~\cite{Neff04:MarkPledge}, Pret \`{a} Voter~\cite{ChaumRyanSchneider2005},  
 Remotegrity~\cite{Zagorski13}, Scantegrity II~\cite{Chaum08II} and 
  Three Ballot~\cite{Rivest07:ThreeBallot}  
  all use features implemented with paper, such as scratch-off surfaces and detachable columns.
  But these systems fall outside the scope of our election scheme
  syntax.) 

Some scenarios inevitably reveal voters'
votes: Unanimous election outcomes reveal how everyone voted and, more generally,
outcomes can be coupled with partial knowledge 
of voters' votes to deduce voters' votes.
For example, suppose Alice, Bob and Mallory participate in a referendum and the outcome has frequency two for `yes' and
one for `no.' Mallory and Alice can deduce Bob's vote by pooling knowledge of their
own votes. Similarly, Mallory and Bob can deduce Alice's vote. Furthermore, Mallory
can deduce that Alice and Bob both voted yes, if she voted no. For simplicity,
our informal definition of ballot secrecy (\S\ref{sec:intro}) deliberately omitted
side-conditions that exclude these inevitable revelations and that are necessary for satisfiability.
We now refine that definition as follows:
\begin{quote}
A voter's vote is not revealed to anyone, except when  the vote can be deduced from the election outcome and any partial knowledge of voters' votes.
\end{quote}

\noindent
This refinement ensures the aforementioned examples are not violations of
ballot secrecy.  By comparison, if Mallory votes yes and she can deduce the vote of
Alice, without knowledge of Bob's vote, then ballot secrecy is violated.

\subsection{Security definition}
\label{sec:secrecy:def}
The core idea behind ballot secrecy is
an \emph{indistinguishability test}. Imagine two parallel elections
that are identical in every way, except that two voters, say Alice and
Bob, swap their votes. In one election Alice votes for candidate~$1$
and Bob for candidate~$2$; in the other, Alice votes for candidate~$2$
and Bob for candidate~$1$. Every other voter casts the same vote in
both elections. Because the same collection of votes is cast in total,
the election outcome is identical in both cases.

Then we ask, can anyone, by examining the encrypted ballots, the
tally, and the proof, determine which of the two elections actually
took place? If not, then the scheme keeps ballots secret: no one can
link a specific vote to a specific voter, even with access to all
public election data. The formal game below captures exactly this: a
secret coin flip determines which election is run, an adversary sees
everything that is made public, and ballot secrecy holds if no
efficient adversary (meaning a probabilistic polynomial-time one) can guess the coin with probability meaningfully
better than~$\frac{1}{2}$.

We formalize ballot secrecy (Definition~\ref{def:ballotSecrecy}) as a game  that tasks the 
adversary to: select two lists of votes; construct a bulletin board from ballots for 
votes in one of those lists, which list is decided by a coin flip; and (non-trivially) determine the result of the coin flip from the resulting election outcome and tallying proof. That is, the game tasks 
the adversary to distinguish between an instance of the voting system for one list of votes, 
from another instance with the other list of votes, when the votes cast from each list
are permutations of each other (hence, the distinction is non-trivial). 
The game proceeds as follows: The challenger generates a key pair (Line~\ref{alg:BallotSecrecyGame:setup}), the adversary chooses some number of candidates 
(Line~\ref{alg:BallotSecrecyGame:adv:nc}), and the challenger flips a coin
(Line~\ref{alg:BallotSecrecyGame:beta}) and initialises a set to record lists 
of votes (Line~\ref{alg:BallotSecrecyGame:L}).
The adversary computes a bulletin board from ballots for votes in one of two possible lists
(Line~\ref{alg:BallotSecrecyGame:adv:bb}), where the lists are chosen by the 
adversary, the choice between lists is determined by the coin flip, and the 
ballots (for votes in one of the lists) are 
constructed by an oracle (further ballots may be constructed by
the adversary).
The challenger tallies the bulletin board to derive the election outcome and tallying 
proof (Line~\ref{alg:BallotSecrecyGame:tally}), which are given to the adversary
and the adversary is tasked with determining the result of the coin flip
(Line~\ref{alg:BallotSecrecyGame:adv:guess} \& \ref{alg:BallotSecrecyGame:ret}). 

The formal definition includes a \emph{balanced condition} that prevents trivial attacks. An adversary querying the oracle with input $(1, 2)$ receives a ballot $b$ that encrypts either vote~$1$ or vote~$2$, depending on the secret bit $\beta$. Placing only this ballot on the bulletin board would let the adversary win by simply observing whether candidate~$1$ or~$2$ received a vote. The balanced condition excludes such trivial wins by requiring that the votes for $\beta = 0$ be a permutation of the votes for $\beta = 1$, ensuring both cases produce the same election outcome.

As an example, an adversary making queries $(1,2)$, $(2,1)$, and $(3,3)$ receives ballots $b_1$, $b_2$, $b_3$. A bulletin board containing all three is balanced: when $\beta=0$ the votes are $1,2,3$, and when $\beta=1$ they are $2,1,3$---both yield one vote per candidate. The adversary must exploit some weakness in the voting scheme itself to win, not merely observe the outcome. Omitting $b_2$ and $b_3$ would make the board unbalanced (one vote for candidate~$1$ versus one for candidate~$2$), and such configurations are excluded.

\begin{definition}[$\BallotSecrecy$~\cite{Smyth15:BallotSecrecyFull}]\label{def:ballotSecrecy}
Let $\Gamma = (\SetupSymb,\VoteSymb,\TallySymb,\VerifySymb)$ be an election scheme, 
$\adv$ be an adversary, $\kk$ be a security parameter, and $\BallotSecrecyGame$ be 
the following game.

{\upshape
\begin{inlineexperiment}{$\BallotSecrecyGame$}
$(\pk,\sk,\mB,\mC) \leftarrow \Setup$\label{alg:BallotSecrecyGame:setup}\;
$\nC \leftarrow \adv(\pk,\kk)$\label{alg:BallotSecrecyGame:adv:nc}\;
$\beta \leftarrow_R \{0,1\}$\label{alg:BallotSecrecyGame:beta}\;
$L \leftarrow \emptyset$\label{alg:BallotSecrecyGame:L}\;
$\bb \leftarrow \adv^{\oracle}()$\label{alg:BallotSecrecyGame:adv:bb}\;
$(\outcome, \tpf) \leftarrow \Tally$\label{alg:BallotSecrecyGame:tally}\;
$g\leftarrow \adv(\outcome, \tpf)$\label{alg:BallotSecrecyGame:adv:guess}\;
\Return $(g = \beta )\wedge\allowbreak \balanced(\bb,\nC,L) 
	 \wedge\allowbreak (1\leq\nC\leq \mC )\wedge\allowbreak (|\bb| \leq \mB)	$\label{alg:BallotSecrecyGame:ret}\;
\end{inlineexperiment}
}

\noindent
Predicate $\balanced(\bb,\nC,L)$ holds when: for all votes  
$v\in\{1,\dots,\nC\}$ we have 
$$|\{ b \mid b \in \bb \wedge \exists v_1 \mathrel. (b,v,v_1) \in L\}|
  = |\{b \mid b \in \bb \wedge \exists v_0 \mathrel. (b,v_0,v) \in L\}|\, .$$ Oracle $\oracle$ is defined as follows:

\begin{itemize}
\item $\oracle{(v_0,v_1)}$ computes 
$b\leftarrow \Vote[\pk,v_\beta,\nC,\kk];  L \leftarrow L \cup \{ (b, v_0, v_1)\}$ and 
outputs $b$, where $v_0,v_1\in \{1, ...,\nC \}$.
\end{itemize}

\begin{sloppypar}
\noindent
We say election scheme $\Gamma$ satisfies $\BallotSecrecy$, if for all 
probabilistic polynomial-time 
adversaries $\adv$, there exists a negligible 
function 
$\negl$, such that for all security parameters $\kk$, the inequality $$\Succ(\BallotSecrecyGame) \leq 
\frac{1}{2} + \neglK\,,$$ holds.
\end{sloppypar}
\end{definition} 

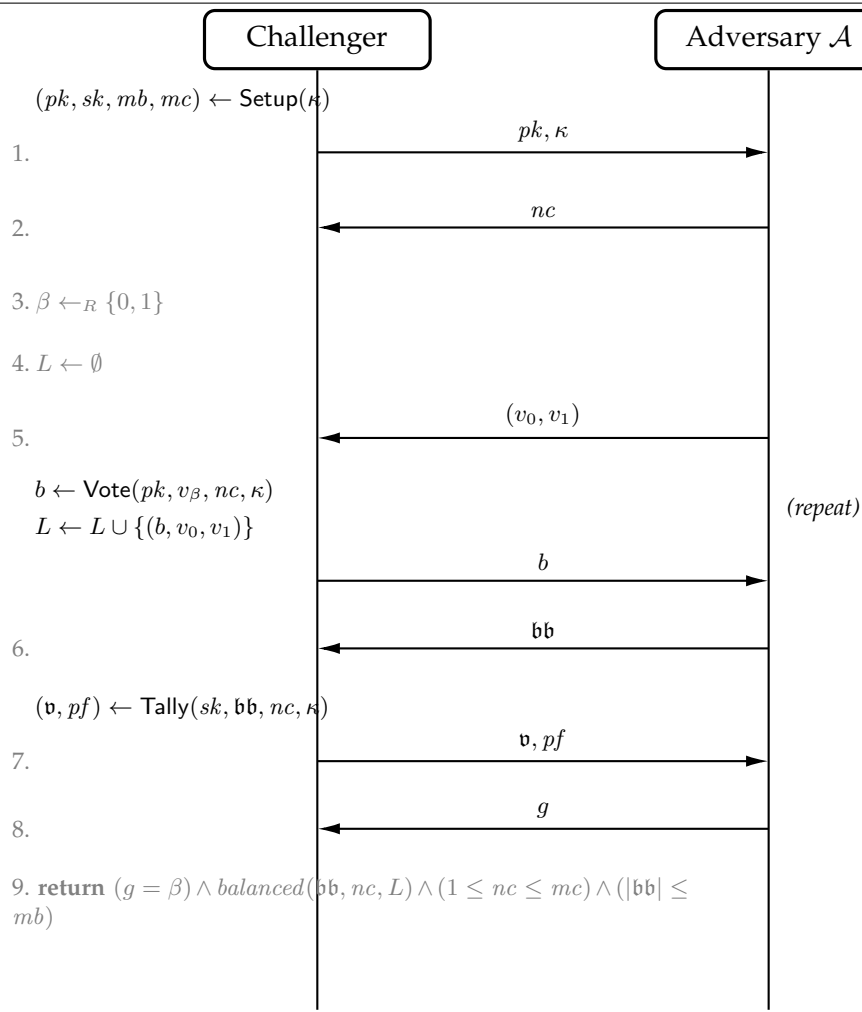
\begin{figure}
\centering
\begin{tikzpicture}[thick]
\tikzset{
    entity/.style={rectangle, rounded corners, draw=black, very thick, inner sep=0.5em, minimum width=3cm, minimum height=1.5em, text centered},
    arrowmsg/.style={-{Latex[length=3mm,width=1.5mm]}, thick},
    msglabel/.style={font=\footnotesize},
    steplabel/.style={font=\footnotesize, text=gray, anchor=west}
}

\node[entity] (chall) at (-3, 0) {Challenger};
\node[entity] (adv) at (3, 0) {Adversary $\adv$};

\draw[thick] (chall.south) -- ++(0, -12.5);
\draw[thick] (adv.south) -- ++(0, -12.5);

\node[msglabel, anchor=west] at (-6.9, -0.8) {$(\pk, \sk, \mB, \mC) \leftarrow \Setup$};
\node[steplabel] at (-7.2, -1.5) {1.};
\draw[arrowmsg] (-3, -1.5) -- node[above, msglabel] {$\pk, \kk$} (3, -1.5);

\node[steplabel] at (-7.2, -2.5) {2.};
\draw[arrowmsg] (3, -2.5) -- node[above, msglabel] {$\nC$} (-3, -2.5);

\node[steplabel] at (-7.2, -3.5) {3.\ $\beta \leftarrow_R \{0,1\}$};

\node[steplabel] at (-7.2, -4.3) {4.\ $L \leftarrow \emptyset$};

\node[steplabel] at (-7.2, -5.3) {5.};
\draw[arrowmsg] (3, -5.3) -- node[above, msglabel] {$(v_0, v_1)$} (-3, -5.3);
\node[msglabel, anchor=west] at (-6.9, -6.0) {$b \leftarrow \VoteSymb(\pk, v_\beta, \nC, \kk)$};
\node[msglabel, anchor=west] at (-6.9, -6.5) {$L \leftarrow L \cup \{(b, v_0, v_1)\}$};
\draw[arrowmsg] (-3, -7.2) -- node[above, msglabel] {$b$} (3, -7.2);
\node[msglabel, anchor=west] at (3.1, -6.25) {\textit{(repeat)}};

\node[steplabel] at (-7.2, -8.1) {6.};
\draw[arrowmsg] (3, -8.1) -- node[above, msglabel] {$\bb$} (-3, -8.1);

\node[msglabel, anchor=west] at (-6.9, -8.9) {$(\outcome, \tpf) \leftarrow \TallySymb(\sk, \bb, \nC, \kk)$};
\node[steplabel] at (-7.2, -9.6) {7.};
\draw[arrowmsg] (-3, -9.6) -- node[above, msglabel] {$\outcome, \tpf$} (3, -9.6);

\node[steplabel] at (-7.2, -10.5) {8.};
\draw[arrowmsg] (3, -10.5) -- node[above, msglabel] {$g$} (-3, -10.5);

\node[steplabel, text width=9cm] at (-7.2, -11.5) {9.\ $\Return\ (g = \beta) \land \balanced(\bb, \nC, L) \land (1 \leq \nC \leq \mC) \land (|\bb| \leq \mB)$};

\end{tikzpicture}
\medskip
\caption{The ballot secrecy game (Definition~\ref{def:ballotSecrecy}). The challenger generates keys and samples a secret bit $\beta$. The adversary repeatedly queries an oracle with vote pairs $(v_0, v_1)$ and receives a ballot encrypting $v_\beta$; the set $L$ records these queries. After submitting a bulletin board $\bb$, the adversary receives the tally and proof, then guesses $\beta$. The adversary wins if the guess is correct and the bulletin board is balanced (the multiset of left votes equals the multiset of right votes).}
\label{fig:ballot-secrecy-game}
\end{figure}

\noindent
Figure~\ref{fig:ballot-secrecy-game} illustrates the interactions between challenger and adversary in the ballot secrecy game.
An election scheme delivers ballot secrecy when the voting algorithm $\VoteSymb$ produces ballots that conceal voters' choices, and when the tallying algorithm $\TallySymb$ generates outcomes and proofs without exposing how individual ballots map to the final result.
In the game $\BallotSecrecy$, the adversary must construct a bulletin board using ballots that an oracle creates for votes from one of two possible lists. The adversary then attempts to identify which list was chosen, based solely on the election outcome and tallying proof derived from that board.
The choice between lists is determined by the result $\beta$ 
of a coin flip, and the oracle outputs a ballot for vote $v_\beta$ on input of a pair of votes
$v_0,v_1$.
Hence, the  oracle constructs ballots for votes in one of two possible lists,
where the lists are chosen by the adversary, and the bulletin board may 
contain those ballots in addition to 
ballots constructed by the adversary. 

Election schemes reveal the number of votes for each candidate, meaning the election outcome).
To avoid trivial distinctions in game $\BallotSecrecy$,
 we require that runs of the game are \emph{balanced}: inputs to the  oracle for $\beta = 0$ are equivalent to inputs for $\beta = 1$, when the corresponding 
outputs appear on the bulletin board.
For example, suppose the inputs to the oracle 
are $(v_{1,0},v_{1,1}),\dots, (v_{n,0},v_{n,1})$ and the corresponding
outputs are $b_1,\dots,b_n$, further suppose the bulletin board 
is $\{b_1,\dots,b_\ell\}$ such that $\ell \leq n$.
That game is balanced if the inputs for $\beta = 0$, namely $v_{1,0},\dots,v_{\ell,0}$, 
are a permutation of the inputs for $\beta = 1$, namely $v_{1,1},\dots,v_{\ell,1}$.
To avoid trivial wins, we require the balanced condition. Consider an adversary who constructs a bulletin board with only the ballot from oracle query $(1, 2)$. Although such an adversary could easily identify $\beta$ by observing whether candidate $1$ or $2$ received a vote, this configuration violates the balanced condition and therefore doesn't count as a valid win. (We leave formally defining a winning adversary as an exercise.)

We can easily see the intuition behind the game $\BallotSecrecy$. A winning adversary possesses a method for distinguishing ballots. Such an adversary can tell apart a voting system instance where voters submit certain votes from another instance where they submit a permutation of those same votes—thereby exposing voters' choices. Otherwise, the adversary is unable to distinguish between a voter casting a ballot for vote $v_0$ 
and another voter casting a ballot for vote $v_1$, hence, voters' votes cannot be 
revealed. 
\subsection{Example}
\label{sec:secrecy:example}
We outlined asymmetric encryption schemes in 
Section~\ref{sec:asym}, while more precise details are given in Section~\ref{sec:asym} of the appendix. Using an asymmetric scheme, we now
model our $\encToVoteSymb$ voting system (\S\ref{sec:intro}) as the 
following election scheme.

\begin{definition}[$\encToVoteSymb$\cite{Smyth15:BallotSecrecyFull}]\label{def:mv}
Given an asymmetric encryption scheme 
$\Pi = (\GenSymb,\allowbreak\EncSymb,\allowbreak\DecSymb)$, 
we define $\encToVote =\allowbreak (\SetupSymb,\allowbreak\VoteSymb,\allowbreak\TallySymb,\allowbreak\VerifySymb)$ 
such that:

\begin{itemize}
\item $\Setup$ computes $(\pk,\sk,\votespace)\leftarrow\Gen;\pk'\leftarrow(\pk,\votespace);\allowbreak\sk'\leftarrow(\pk,\sk)$, 
derives $\mC$ as the largest integer such that $\{0,\dots,\mC\}\subseteq\{0\} \cup\votespace$
and for all $m_0,m_1\in\{1,\dots,\mC\}$ we have $|m_0|=|m_1|$,
and outputs $(\pk',\allowbreak\sk',\allowbreak p(\kk),\mC)$, 
where $p$ is a polynomial function.

\item $\Vote[\pk',v,\nC,\kk]$ parses $\pk'$ as pair $(\pk,\votespace)$, 
outputting $\perp$ if parsing fails or
$(v\not\in\{1,\dots,\nC\}) \vee (\{1,\allowbreak\dots,\allowbreak\nC\}\not\subseteq\votespace)$,
computes $b\leftarrow\Enc{v}$, and outputs $b$.

\item $\Tally[\sk',\bb,\nC,\kk]$ initialises $\outcome$ as a zero-filled vector of length $\nC$,
parses $\sk'$ as pair $(\pk,\sk)$, outputting $(\outcome,\perp)$ if parsing fails,  
computes 
{ \DontPrintSemicolon
\llFor{$b\in\bb$}{%
$v\leftarrow\Dec{b};$
\llIf{$1\leq v\leq \nC$}{$\outcome[v] \leftarrow \outcome[v]+1$}%
}%
}, and outputs 
$(\outcome,\epsilon)$, where $\epsilon$ is a constant symbol.

\item $\Verify$ outputs 1. 
\end{itemize}

\end{definition}

Note that $\pk'$ and $\sk'$ have special structure in this construction: the election's public key $\pk' = (\pk, \votespace)$ bundles the encryption public key with the message space, while the election's private key $\sk' = (\pk, \sk)$ bundles both the encryption public key and private key.

\noindent
To ensure $\encToVote$ is an election scheme, we require asymmetric
encryption scheme $\Pi$ to produce distinct ciphertexts with overwhelming
probability, otherwise correctness cannot be satisfied, as the following
lemma demonstrates.

\begin{lemma}\label{lemma:enc2vote:correctnessIssue}
\begin{sloppypar}
There exists an asymmetric encryption scheme $\Pi$ such that 
$\encToVote$ is not an election scheme. 
\end{sloppypar}
\end{lemma}

\noindent
To prove our lemma, we show that colliding ciphertexts suffice to ensure
that $\encToVoteSymb$ cannot satisfy the correctness property 
of Definition~\ref{def:election}. 

\begin{proof}
Let $\encToVote = (\SetupSymb,\allowbreak	\VoteSymb,\allowbreak\TallySymb,\allowbreak\VerifySymb)$.
Suppose $(\pk',\allowbreak\sk',\allowbreak\mB,\allowbreak \mC)$ is an 
output of $\Setup$ and $b$ and $b'$ are outputs of $\Vote[\pk',v,\nC,\kk]$
such that $(2 \leq \mB ) \wedge (1\leq v \leq \nC \leq \mC)$, where
$\kk$ is a security parameter. Further suppose $\outcome$ is 
a zero-filled vector of length $\nC$, except for index $v$ which contains
the value $2$. 
Moreover, suppose $(\outcome',\tpf)$ is an output of 
$\Tally[\sk',\{b,b'\},\nC, \kk]$. If $b$ and $b'$ collide, then outcome 
$\outcome'$ is computed from the set $\{b,b'\} = \{b\}$, therefore, 
the correct outcome cannot have been computed, implying
$\outcome \not= \outcome'$ with a non-negligible probability, so correctness is not satisfied.
By definition of algorithm $\SetupSymb$, $\pk'$ is a pair, and, 
by definition of algorithm $\VoteSymb$, $b$ and $b'$ are ciphertexts
on plaintext $v$. Consequently, we now need to show that asymmetric encryption schemes 
can produce ciphertexts that collide. Indeed, they can: Consider an encryption scheme
$\Pi = (\GenSymb,\EncSymb,\allowbreak\DecSymb)$ such that $\Enc{m}$ outputs $m$ 
and $\Dec{c}$ outputs $c$. Although $\Pi$ is clearly not secure, it 
is straightforward to see that $\Pi$ satisfies correctness, because 
$\Dec{\Enc{m;r}} = m$ for all key pairs $(\pk,\sk)$, plaintexts $m$, and 
coins $r$.
\end{proof}

\noindent 
It follows from Lemma~\ref{lemma:enc2vote:correctnessIssue} that we must 
restrict the class of asymmetric encryption schemes used to 
instantiate $\encToVoteSymb$. We could consider a broad class of 
schemes that produce distinct ciphertexts with overwhelming
probability, but we favour the narrower class of non-malleable 
schemes, since we require non-malleability for ballot secrecy.
For non-malleability in general, the definitions are complex and proofs are relatively difficult. This motivates us to adopt the definition of indistinguishability 
under parallel attack ($\INDPA$) by Bellare \& Sahai~\cite{Bellare99:IND-k-CPA},
which is simpler, yet equivalent to their definition of comparison based 
non-malleability ($\CNMCPA$). We recall the definition of $\INDPA$ in Section~\ref{sec:INDPA} of
the appendix.

\begin{lemma}
If asymmetric encryption scheme $\Pi$ satisfies $\INDPA$, then
$\encToVote$ is an election scheme. 
\end{lemma}

\noindent
To prove our lemma, we show that $\encToVoteSymb$ satisfies the correctness
property of Definition~\ref{def:election} when ciphertexts do not 
collide.

\begin{proof}
Let $\encToVote = (\SetupSymb,\allowbreak	\VoteSymb,\allowbreak\TallySymb,\allowbreak\VerifySymb)$ and $\Pi = (\GenSymb,\EncSymb,\allowbreak\DecSymb)$.
Moreover, let $\kk$ be a security parameter, $\nB$ and $\nC$ be integers,
$v_1,\dots,v_{\nB}\in\{1,\dots,\nC\}$ be votes, and $\outcome$ be a 
zero-filled vector of length $\nC$. Suppose $(\pk',\allowbreak\sk',\allowbreak\mB,\allowbreak \mC)$
is an output of $\Setup$ such that $(\nB \leq \mB) \wedge (\nC\leq \mC)$. 
Further suppose we compute 
{\DontPrintSemicolon\llFor{$1 \leq i \leq \nB$}{%
  $b_i \leftarrow \Vote[\pk,v_i,\nC,\kk];
  \outcome[v_i] \leftarrow \outcome[v_i] + 1$%
}}.
Moreover, suppose $(\outcome',\tpf)$ is an output of 
$\Tally[\sk,\{b_1,\dots,b_{\nB}\},\nC, \kk]$.
To prove correctness, it suffices to prove $\outcome = \outcome'$, 
with overwhelming probability.

By definition of algorithm $\SetupSymb$, we have $\pk'$ is a pair $(\pk,\votespace)$
and $\sk'$ is a pair $(\pk,\sk)$ such that $(\pk,\sk,\votespace)$ was output by 
$\Gen$. Moreover, $\mC$ is the largest integer such that 
$\{0,\dots,\mC\}\subseteq\{0\} \cup\votespace$, hence, 
$\{1,\allowbreak\dots,\allowbreak\nC\}\subseteq\votespace$.
It follows by definition of algorithm $\VoteSymb$ that for each 
$i\in\{1,\dots,\nB\}$ we have $b_i$ is an output of $\Enc{v_i}$. 
By definition of algorithm $\TallySymb$, 
outcome $\outcome'$ is initialized as a zero-filled vector of length $\nC$ 
and computed as follows: 

{\DontPrintSemicolon
\llFor{$b\in\{b_1,\dots,b_{\nB}\}$}{%
$v\leftarrow\Dec{b};$
\llIf{$1\leq v\leq \nC$}{$\outcome'[v] \leftarrow \outcome'[v]+1$}%
}%
}. 

\noindent
Since $\Pi$ satisfies $\INDPA$, ciphertexts $b_1,\dots,b_{\nB}$
are distinct with overwhelming probability, hence, that computation 
is equivalent to the following:

{\DontPrintSemicolon
\llFor{$1 \leq i \leq \nB$}{%
$v\leftarrow\Dec{b_i};$
\llIf{$1\leq v\leq \nC$}{$\outcome'[v] \leftarrow \outcome'[v]+1$}%
}%
}.

\noindent
Moreover, by correctness of $\Pi$, we have $\Dec{b_i} = v_i$
for all $i\in\{1,\dots,\nB\}$, with 
overwhelming probability. Thus, the above computation is equivalent to 
computing

{\DontPrintSemicolon
\llFor{$1 \leq i \leq \nB$}{%
$\outcome'[v_i] \leftarrow \outcome'[v_i]+1$}%
},

\noindent
with overwhelming probability. It follows that outcomes $\outcome$ and $\outcome'$
are computed identically, hence, $\outcome=\outcome'$, as required, 
with overwhelming probability. 
\end{proof}

\noindent
Intuitively, election scheme $\encToVote$ satisfies ballot secrecy until 
tallying, because asymmetric encryption scheme $\Pi$ can ensure that voters'
votes are not revealed, and tallying maintains ballot secrecy by revealing 
only the election outcome.

\begin{proposition}\label{prop:encToVoteBS}
If an asymmetric encryption scheme $\Pi$ satisfies $\INDPA$, then
election scheme $\encToVote$ satisfies $\BallotSecrecy$.
\end{proposition}

\noindent
Proving this proposition and other ballot secrecy results is time consuming. 
Indeed, Quaglia \& Smyth's ballot-secrecy proof for our simple $\encToVoteSymb$
voting system fills over six and a half pages~\cite[Appendix C.6]{Smyth15:Hawk}
and 
Cortier \emph{et al.} devoted one person-year to their ballot-secrecy
proof for Helios~\cite{Cortier17:HeliosMachineChecked}. 
To reduce the expense of ballot-secrecy proofs, the following section introduces 
sufficient conditions that enable the simplification of game $\BallotSecrecy$, 
which gives way to simpler proofs. Indeed, we prove 
Proposition~\ref{prop:encToVoteBS} in just over a page.

\LinesNotNumbered
\subsection{Simplifying proofs}\label{sec:secrecy:simpleProofs}

Tallying proofs may reveal voters' votes.  For example, a variant of $\encToVoteSymb$ 
might define tallying proofs that map ballots to votes. Hence, such proofs are rightly
provided to the adversary in game $\BallotSecrecy$ (Line~\ref{alg:BallotSecrecyGame:adv:guess}).
Nevertheless, if tallying proofs
reveal nothing about the votes expressed in ballots on the bulletin board, then they 
can be omitted from the game. This precondition is ensured by election schemes that use 
zero-knowledge tallying proofs. Thus, the adversary need not be 
provided with such proofs in game $\BallotSecrecy$ when analysing such schemes, 
which achieves our first reduction in the expense of ballot-secrecy proofs. Our second
reduction involves modifying the computation of election outcomes. 

Game $\BallotSecrecy$ computes the election outcome from ballots constructed by the 
oracle and ballots constructed by the adversary (Line~\ref{alg:BallotSecrecyGame:tally}).
Intuitively, such an outcome can be equivalently computed as follows:

{\upshape
\begin{algorithm}[H]
$(\outcome, \tpf)\leftarrow \Tally[\sk,\bb\setminus\{b \mid (b,v_0,v_1)\in L\},\nC,\kk]$\;
$(\outcome', \tpf')\leftarrow \Tally[\sk,\bb\cap\{b \mid (b,v_0,v_1)\in L\},\nC,\kk]$\;
$\outcome \leftarrow \outcome + \outcome'$\;
\end{algorithm}
}

\noindent
Yet, a poorly designed tallying algorithm might not ensure equivalence. In particular,
ballots constructed by the adversary can cause the algorithm to behave unexpectedly.
(Such algorithms are nonetheless compatible with our correctness requirement, because
correctness does not consider an adversary.) Nevertheless, the equivalence holds when 
individual ballots are tallied correctly. Moreover, the above computation is equivalent 
to the following:

{\upshape
\begin{algorithm}[H]
$(\outcome, \tpf)\leftarrow \Tally[\sk,\bb\setminus\{b \mid (b,v_0,v_1)\in L\},\nC,\kk]$\;
\For{$(b \in \bb) \wedge ((b,v_0,v_1) \in L)$}{
	$(\outcome', \tpf')\leftarrow \Tally[\sk,\{b\},\nC,\kk]$\;
	$\outcome \leftarrow \outcome + \outcome'$\;
}
\end{algorithm}
}

\noindent
Furthermore, by correctness of the election scheme, the above for-loop 
can be equivalently computed as follows:

{\upshape
\begin{algorithm}[H]
\For{$(b \in \bb) \wedge ((b,v_0,v_1) \in L)$}{
	$\outcome [v_\beta] \leftarrow \outcome [v_\beta] +1$\;
}
\end{algorithm}
}

\noindent
Indeed, for each $(b \in \bb) \wedge ((b,v_0,v_1) \in L)$, we have $b$ is an output of 
$\Vote[\pk,\allowbreak v_\beta,\allowbreak\nC,\allowbreak\kk]$, hence, 
$\Tally[\sk,\{b\},\nC,\kk]$ outputs $(\outcome, \tpf)$ such that $\outcome$ 
is a zero-filled vector, except for index $v_\beta$ which contains one, 
and this suffices to ensure equivalence. In addition, for any adversary that wins 
game $\BallotSecrecy$, we are assured that $\balanced(\bb,\nC,L)$ holds, hence, 
the above for-loop can be computed as

{\upshape
\begin{algorithm}[H]
\For{$(b \in \bb) \wedge ((b,v_0,v_1) \in L)$}{
	$\outcome [v_0] \leftarrow \outcome [v_0] +1$\;
}
\end{algorithm}
}

\noindent
or

{\upshape
\begin{algorithm}[H]
\For{$(b \in \bb )\wedge ((b,v_0,v_1) \in L)$}{
	$\outcome [v_1] \leftarrow \outcome [v_1] +1$\;
}
\end{algorithm}
}

\noindent
without weakening the game. Thus, perhaps surprisingly, tallying ballots constructed 
by the oracle does not 
provide the adversary with an advantage (in determining whether $\beta = 0$ or 
$\beta = 1$) and we can omit such ballots from tallying in game $\BallotSecrecy$.
That is, we need only consider the game derived from $\BallotSecrecy$ by 
replacing $\adv(\outcome, \tpf)$
with $\adv(\outcome)$ and $\balanced(\bb,\nC,L)$ with $\{b\mid (b,v_0,v_1)\in L\} \cap \bb = \emptyset$, 
where the former
modification captures our first reduction in the expense of ballot-secrecy proofs
and the latter captures our second. 

Smyth~\cite{Smyth15:BallotSecrecyFull} further reduces the expense of ballot-secrecy 
proofs by removing the oracle in favour of a single challenge ballot:

\LinesNumbered
\begin{definition}[$\BallotIndependence$~\cite{Smyth15:BallotSecrecyFull}]\label{def:ind-independence}
Let $\Gamma = (\SetupSymb,\VoteSymb,\TallySymb,\VerifySymb)$ be 
an election scheme, $\adv$ be an adversary, $\kk$ be the security parameter,
and $\BallotIndependenceGame$ be the following game.

{\upshape
\begin{inlineexperiment}{$\BallotIndependenceGame$}
			$(\pk,\sk,\mB,\mC) \leftarrow \Setup$\;
			$(v_0,v_1,\nC) \leftarrow \adv(\pk,\kk)$\;
			$\beta \leftarrow_R \{0,1\}$\;
			$b \leftarrow \Vote[\pk,v_\beta,\nC,\kk]$\;
			$\bb \leftarrow \adv(b)$\;
			$(\outcome,\tpf) \leftarrow \Tally$\label{alg:BallotIndependenceGame:tally}\;
			$g\leftarrow \adv(\outcome)$\;
			\Return $(g = \beta)\mathrel\wedge\allowbreak
			(b\not\in\bb )\mathrel\wedge\allowbreak (1 \leq v_0,v_1 \leq \nC \leq \mC \mathrel)\wedge\allowbreak (|\bb| \leq \mB)$\;
\end{inlineexperiment}
}

\noindent
We say election scheme $\Gamma$ satisfies  $\BallotIndependence$, 
if for all probabilistic polynomial-time adversaries $\adv$,
there exists a negligible function $\negl$, such that for all security parameters $\kk$, 
the following inequality 
$$\Succ(\BallotIndependenceGame) \leq \frac{1}{2} + \neglK \,, $$
holds. 
\end{definition}

\LinesNotNumbered

\noindent
An election scheme satisfies $\BallotIndependence$ when algorithm $\VoteSymb$ outputs non-malleable ballots. 

The relationship between $\BallotSecrecy$ and $\BallotIndependence$ (also called IND-CVA, for \emph{indistinguishability under chosen-vote attack}) is central to simplifying ballot-secrecy proofs. In $\BallotSecrecy$, the adversary can query an oracle many times to obtain ballots for vote pairs of their choice, then must distinguish which votes were actually encrypted. In $\BallotIndependence$, the adversary receives only a \emph{single} challenge ballot and must determine which of two votes it contains---but may construct additional ballots themselves and place them on the bulletin board.

Why does this simplification work? The key insight is that if ballots are non-malleable, the adversary gains nothing from oracle access beyond what they could construct themselves. Malleable ballots would allow an adversary to transform an oracle-provided ballot into a related ballot, potentially leaking information about the encrypted vote. Non-malleability blocks this attack vector.

Smyth proves that game $\BallotSecrecy$ is strictly stronger than game $\BallotIndependence$, moreover, 
he proves that the games coincide for election schemes with zero-knowledge tallying proofs 
that tally individual ballots correctly~\cite[Theorems~1 \&~4]{Smyth15:BallotSecrecyFull}. Here strictly stronger means any scheme satisfying $\BallotSecrecy$ also satisfies $\BallotIndependence$, but not necessarily vice versa. However, the games coincide (yield equivalent security guarantees) when two conditions hold: (1) tallying proofs reveal nothing about individual votes (zero-knowledge), and (2) the tally correctly counts each ballot. Under these conditions, proving the simpler $\BallotIndependence$ game suffices to establish $\BallotSecrecy$.

Furthermore, Smyth proves that universally-verifiable election schemes tally 
individual ballots correctly~\cite[Lemmata~8 \&~24]{Smyth15:BallotSecrecyFull}.

\begin{mdframed}[style=guidebox]
\begin{guideline}
Ballots must be non-malleable.
\end{guideline}
\end{mdframed}

These results significantly reduce the expense of ballot-secrecy proofs
and we now use them to prove Proposition~\ref{prop:encToVoteBS}.

\begin{proof}[Proof of Proposition~\ref{prop:encToVoteBS}]
Suppose to the contrary that election scheme $\encToVote[\allowbreak\Pi]$ does not satisfy $\BallotSecrecy$.
Since $\BallotSecrecy$ is strictly stronger than 
$\BallotIndependence$~\cite[Theorem~1]{Smyth15:BallotSecrecyFull},
scheme $\encToVote$ does not satisfy $\BallotIndependence$ either, hence, 
there exists a probabilistic polynomial-time adversary $\adv$ that wins 
$\BallotIndependenceGame[\encToVote,\allowbreak\adv,\allowbreak\kk]$ with 
success greater than negligibly better than guessing, for some security 
parameter $\kk$. From $\adv$, we construct the following adversary $\Adv$ 
that wins $\INDPA$.

\begin{itemize}
\item $\Adv(\pk,\votespace,\kk)$ computes 
  $\pk'\leftarrow (\pk,\votespace);(v_0,v_1,\nC) \leftarrow \adv(\pk',\kk)$ and
  outputs $(v_0,v_1)$.

\item $\Adv(b)$ computes $\bb\leftarrow \adv(b)$, parses $\bb$ as a set
  $\{b_1,\dots,b_{|\bb|}\}$, and outputs vector $(b_1,\dots,b_{|\bb|})$.

\item $\Adv({\bf m})$ initialises $\outcome$ as a zero-filled vector of length $\nC$,
computes
{\DontPrintSemicolon
\llFor{$1 \leq i \leq |{\bf m}|$}{%
$v\leftarrow {\bf m}[i]$;
\llIf{$1\leq v\leq \nC$}{$\outcome[v] \leftarrow \outcome[v]+1;$}
$g \leftarrow \adv(\outcome)$%
}%
} and outputs $g$.
\end{itemize}

\noindent
We prove that the success of adversary $\Adv$ is equivalent to the success of 
adversary~$\adv$, which contradicts our assumption that $\Pi$ satisfies $\INDPA$.

Let $\Pi = (\GenSymb,\allowbreak\EncSymb,\allowbreak\DecSymb)$ 
and $\encToVote =\allowbreak (\SetupSymb,\allowbreak\VoteSymb,\allowbreak\TallySymb,\allowbreak\VerifySymb)$. 
Suppose $(\pk,\sk,\votespace)$ is an output of $\Gen$, $(v_0,v_1)$
is an output of $\Adv(\pk,\allowbreak\votespace,\allowbreak\kk)$, and $b$ is an output of
$\Enc{v_\beta}$, for some bit $\beta$ chosen uniformly at random.
By inspection of algorithms $\SetupSymb$ and $\VoteSymb$, and of 
adversary $\Adv$, it is straightforward to see that $\Adv$ 
simulates the challenger in $\BallotIndependence$ to adversary $\adv$. 
Indeed, adversary $\Adv$ couples public key $\pk$ with message space $\votespace$,
and inputs the resulting pair $\pk'$ to $\adv$, which
corresponds to the public key computed by algorithm $\SetupSymb$, hence, 
the public key input to $\adv$ by the challenger in $\BallotIndependence$.
Thus, adversary $\adv$ behaves as if playing game $\BallotIndependence$
and output $(v_0,v_1)$ is indistinguishable from outputs that 
would be observed while playing that game. Moreover,
since $\adv$ wins with success greater than negligibly better 
than guessing, we have $1 \leq v_0,v_1 \leq \nC \leq \mC$, 
furthermore, $\{1,\dots,\nC\}\subseteq\votespace$, where 
$\mC$ is the largest integer such that $\{0,\dots,\mC\}\subseteq{\{0\} \cup\votespace}$
and for all $m_0,m_1\in\{1,\dots,\mC\}$ we have $|m_0|=|m_1|$
(hence, $|v_0| = |v_1|$, which is required to win $\INDPA$), with the same probability.
It follows that outputs of $\Enc{v_\beta}$ and $\Vote[\pk,v_\beta,\nC,\kk]$ are indistinguishable. 
Suppose $(b_1,\dots,b_\ell)$ is an output of $\Adv(b)$. It is trivial to see that $\Adv$ 
simulates the challenger in $\BallotIndependence$ to adversary $\adv$, because the aforementioned
outputs of $\EncSymb$ and $\VoteSymb$ are indistinguishable. 
Since $\adv$ wins, we have $b\not\in \{b_1,\dots,b_\ell\}$, hence, 
$\bigwedge_{1\leq i \leq \ell} (b \not= b_i)$, again with the same 
probability.

Let ${\bf m} = (\Dec{b_1},\dots,\Dec{b_\ell})$ and suppose $g$ is 
an output of $\Adv({\bf m})$. By inspection of algorithm $\TallySymb$ 
and of adversary $\Adv$, we can see that $\Adv$ 
simulates the challenger in $\BallotIndependence$ to adversary $\adv$. 
Indeed, both the algorithm and adversary initialise $\outcome$ 
as a zero-filled vector of length $\nC$, then the adversary $\Adv$
computes

\begin{algorithm}[H]
\llFor{$1 \leq i \leq |{\bf m}|$}{%
$v\leftarrow {\bf m}[i]$;
{\DontPrintSemicolon
\llIf{$1\leq v\leq \nC$}{$\outcome[v] \leftarrow \outcome[v]+1$}%
}%
}
\end{algorithm}

\noindent
which is equivalent to algorithm $\TallySymb$ computing

\begin{algorithm}[H]
\llFor{$b\in\{b_1,\dots,b_\ell\}$}{%
$v\leftarrow\Dec{b};$
{\DontPrintSemicolon
\llIf{$1\leq v\leq \nC$}{$\outcome[v] \leftarrow \outcome[v]+1$}%
}%
}
\end{algorithm}

\noindent
because algorithm $\DecSymb$ is deterministic.
Thus, output $g$ is indistinguishable from outputs that would be 
observed in game $\BallotIndependence$. It follows that the success 
of adversary $\Adv$ is equivalent to the success of $\adv$, and we 
conclude our proof by~\cite[Theorem~4]{Smyth15:BallotSecrecyFull},
since Smyth proves that games $\BallotIndependence$ and $\BallotSecrecy$ 
coincide for election schemes with zero-knowledge tallying proofs that tally 
individual ballots correctly. (We omit proving that election scheme $\encToVote$ 
has zero-knowledge tallying proofs and tallies individual ballots correctly to avoid 
recalling formal definitions of those properties. Proving the former 
is trivial, because $\tpf$ is a constant, hence, it reveals nothing about
the votes expressed in ballots $b_1,\dots,b_\ell$. The latter
follows from the definition of algorithm $\TallySymb$, by correctness
of $\Pi$, and since $\Pi$ satisfies $\INDPA$, which is required only 
to ensure ciphertexts do not collide. Formally proving these details
is left as an exercise for the reader.) 
\end{proof}

\begin{figure}[h!]
\centering
\begin{tikzpicture}[thick, font=\small]
\tikzset{
    challenger/.style={rectangle, rounded corners, draw=black, very thick,
        inner sep=0.8em, minimum width=6.2cm, align=center},
    adversary/.style={rectangle, rounded corners, draw=black, thick,
        inner sep=0.6em, minimum width=4cm, align=center},
    simbox/.style={rectangle, rounded corners, draw=black!50, thick,
        inner sep=0.5em, minimum width=4.5cm, align=center,
        fill=light-gray},
    arrowfwd/.style={-{Latex[length=3mm,width=1.5mm]}, thick},
    labelstyle/.style={font=\footnotesize},
}

\node[challenger] (chall) at (0, 6.0)
    {$\INDPA$ Challenger\\[2pt]
     {\footnotesize generates $(\pk, \sk, \votespace)$,}\\
     {\footnotesize flips bit $\beta$, encrypts $v_\beta$}};

\node[adversary, minimum height=5.4cm, minimum width=7.8cm, dashed]
    (advB) at (0, 0.2) {};
\node[above] at (advB.north) {Reduction $\Adv$};

\node[simbox] (sim) at (0, 1.7)
    {Simulates $\BallotIndependence$ game\\[2pt]
     {\footnotesize sets $\pk' \leftarrow (\pk, \votespace)$,
      forwards challenge ballot,}\\
     {\footnotesize tallies $\bb$ from decrypted ciphertexts}};

\node[adversary] (advA) at (0, -1.5)
    {Adversary $\adv$\\[2pt]
     {\footnotesize (attacks $\BallotIndependence$ of $\encToVoteSymb$)}};

\draw[arrowfwd] ([xshift=-2.0cm]chall.south) -- ([xshift=-2.0cm]sim.north)
    node[midway, left=0.15cm, labelstyle, align=right]
    {$\pk, \votespace, \kk$\\[1pt]then $b \leftarrow \Enc{v_\beta}$};

\draw[arrowfwd] ([xshift=2.0cm]sim.north) -- ([xshift=2.0cm]chall.south)
    node[midway, right=0.15cm, labelstyle, align=left]
    {$(v_0, v_1)$\\[1pt]then $(b_1, \ldots, b_\ell)$\\[1pt]then guess $g$};

\draw[arrowfwd] ([xshift=-1.5cm]sim.south) -- ([xshift=-1.5cm]advA.north)
    node[midway, left=0.15cm, labelstyle, align=right]
    {$\pk' = (\pk, \votespace)$\\[1pt]then ballot $b$\\[1pt]then $\outcome$};

\draw[arrowfwd] ([xshift=1.5cm]advA.north) -- ([xshift=1.5cm]sim.south)
    node[midway, right=0.15cm, labelstyle, align=left]
    {$(v_0, v_1, \nC)$\\[1pt]then $\bb$\\[1pt]then guess $g$};

\node[labelstyle, align=center] at (0, -3.0)
    {If $\adv$ wins $\BallotIndependence$ $\Longrightarrow$
     $\Adv$ wins $\INDPA$\\[2pt]
     {\footnotesize Contradicts security of $\Pi$,
      so $\adv$ cannot exist.}};

\end{tikzpicture}
\caption{Reduction for Proposition~3. Adversary $\Adv$ receives an $\INDPA$ challenge (an encryption of either $v_0$ or $v_1$) and uses it as the challenge ballot in a simulated $\BallotIndependence$ game for~$\adv$.  When $\adv$ returns a bulletin board $\bb = \{b_1, \ldots, b_\ell\}$, adversary $\Adv$ submits these ciphertexts to its own challenger for decryption and tallies the result.  Since $\Adv$ perfectly simulates the $\BallotIndependence$ game, any advantage $\adv$ has transfers directly to $\Adv$, contradicting the $\INDPA$ security of~$\Pi$.}
\label{fig:reduction-enc2vote}
\end{figure}

\begin{sloppypar}
We can exploit Proposition~\ref{prop:encToVoteBS} to achieve our
fourth and final reduction in the expense of ballot-secrecy proofs.
Indeed, if an election scheme tallies sets of ballots correctly (rather
than individual ballots, as previously required), then we can compute 
the election outcome using function $\correcttally$ in game $\BallotIndependence$, 
rather than the tallying algorithm, that is, by replacing 
$(\outcome,\tpf) \leftarrow \Tally$ with $\outcome \leftarrow \correcttally(\pk,\nC,\bb,\kk)$. 
It follows that election scheme 
$(\SetupSymb,\allowbreak\VoteSymb,\allowbreak\TallySymb,\allowbreak\VerifySymb)$ 
satisfies $\BallotSecrecy$ if and only if 
$(\SetupSymb,\allowbreak\VoteSymb,\allowbreak\TallySymb',\allowbreak\VerifySymb')$ 
does, assuming algorithms $\TallySymb$ and $\TallySymb'$ both tally 
sets of ballots correctly. Proposition~\ref{prop:encToVoteBS} proves that
election scheme $\encToVote$ satisfies $\BallotSecrecy$, assuming 
$\Pi$ is an asymmetric encryption scheme satisfying $\INDPA$, and
Smyth~\cite[Lemma~12]{Smyth15:BallotSecrecyFull} proves that $\encToVote$ tallies 
sets of ballots correctly, under the additional assumption that $\Pi$ satisfies 
\emph{well-definedness}, that is, ill-formed ciphertexts are distinguishable from 
well-formed ciphertexts. 
Thus, $\BallotSecrecy$ is satisfied by any election scheme derived from 
$\encToVote$ by replacing its tallying and verification algorithms, assuming 
the replacement tallying algorithm tallies sets of ballots correctly and 
uses zero-knowledge tallying proofs~\cite[Theorem~13]{Smyth15:BallotSecrecyFull}. 
Moreover, Smyth proves that universally-verifiable election schemes tally sets of 
ballots correctly~\cite[Lemma~24]{Smyth15:BallotSecrecyFull}. It follows that proofs 
of ballot secrecy are trivial for a class of universally-verifiable,
encryption-based voting systems: Any universally-verifiable election scheme derived from 
$\encToVote$ satisfies  $\BallotSecrecy$ if $\Pi$ satisfies $\INDPA$ and well-definedness, 
and tallying proofs are zero-knowledge.
\end{sloppypar}

We will use our third simplification to prove that a variant of Helios satisfies 
$\BallotSecrecy$ and the fourth to prove that a variant of Helios Mixtnet does too
(the original schemes have known vulnerabilities and they do not satisfy $\BallotSecrecy$),
thereby demonstrating the application of these results.

\LinesNumbered

\subsection{Related definitions of ballot secrecy}
\label{sec:secrecy:related}

Discussion of ballot secrecy originates from Chaum~\cite{Chaum:1981} and the earliest 
definitions of ballot secrecy are due to Benaloh \emph{et al.}~\cite{Benaloh86,BT94:ReceiptFreeVoting,Benaloh96:Thesis}.
More recently, Bernhard \emph{et al.} propose a series of ballot secrecy definitions~\cite{Bernhard12:Helios,2013:BallotIndependenceESORICS,2014:BallotIndependenceEprint,BCGPW15}.
Smyth~\cite{Smyth15:BallotSecrecyFull} shows that these definitions do not 
detect vulnerabilities that arise when an adversary controls the bulletin 
board or the communication channel. 
By comparison, the definition of ballot secrecy that we consider 
(Definition~\ref{def:ballotSecrecy}) detects such vulnerabilities
and appears to be the strongest definition in the literature.

Beyond the computational model of security, Delaune, Kremer \& Ryan
formulate a definition of ballot secrecy in the 
applied pi calculus~\cite{DKR08} and Smyth~\emph{et al.} show that 
this definition is amenable to automated reasoning~\cite{Smyth08:Obs,Smyth10:ObsA,Smyth10:thesis,2016-verifying-observational-equivalence,2017-verifying-observational-equivalence}. An alternative definition is proposed by
Cremers \& Hirschi, along with sufficient conditions which are 
also amenable to automated reasoning~\cite{Hirschi17:SufficientConditionsSecrecyDKR}.
Albeit, the scope of automated reasoning is  limited by analysis tools 
(for example, ProVerif~\cite{2010-ProVerif-manual-version-1.96}), because the function symbols and equational theory used to model
cryptographic primitives might not be suitable for automated analysis (see~\cite{DKR11:TPM-Horn-clauses,BlanchetPaiola12:EQ-lists,Bursuc12:EQ-reenc}).

\subsection{Further notions of privacy}
\label{sec:secrecy:outlook}
\label{sec:related:outlook}

Ballot secrecy formalises a notion of free-choice assuming ballots are 
constructed and tallied in the prescribed manner. Moreover, 
Smyth's definition assumes 
the adversary's capabilities are limited to 
casting ballots on behalf of some voters and controlling the 
votes cast by the remaining voters. 
We have seen that 
voting system $\encToVoteSymb$ satisfies this definition, but ballot 
secrecy does not ensure free-choice when adversaries are able to 
communicate with voters nor when voters deviate from the prescribed 
voting procedure to follow instructions provided by adversaries.
Indeed, the coins used for encryption serve as proof of how a voter voted 
in  $\encToVoteSymb$ 
and the voter may communicate those coins to the adversary. 
Stronger notions of free-choice, such as receipt-freeness~\cite{Moran06,Kiayias15,Cortier16:BeleniosRF,Smyth19:receipt-freeness} and coercion resistance~\cite{JCJ05,Gardner09:CoercionDefinition,Unruh10,Kusters12:CoercionResistance,HainesSmyth19:CoercionSoK},
are needed in the presence of such adversaries.

\emph{Coercion resistance}, informally, is the property ensuring a voter cannot prove to a coercer how they voted, even if the coercer actively participates in the voting process. \emph{Receipt-freeness} is a related but weaker property: a voter cannot construct a receipt proving their vote to a third party. Formalising coercion resistance has proven difficult. Haines and Smyth~\cite{HainesSmyth19:CoercionSoK} survey four prominent definitions and find all but one to be unsuitable, demonstrating the challenges faced in capturing this property. See also K\"usters, Truderung and Vogt~\cite{Kusters12:CoercionResistance} for a game-based treatment.

Ballot secrecy does not provide assurances when deviations from the prescribed 
tallying procedure are possible.
Indeed, ballots can be tallied individually to reveal votes.
Hence, the tallier must be trusted.
Alternatively, we can design election schemes
that distribute the tallier's role amongst several talliers and ensure 
free-choice assuming at least one tallier tallies ballots in the prescribed manner.
Extending   results in this direction is an opportunity for future work.
Ultimately, we would prefer not to trust talliers. Unfortunately, this is only known to
be possible for decentralised voting systems, for example~\cite{Schoenmakers99:PVSS,KY02,G04,HRZ09,Smyth12:decentralised-voting-system,Khazaei18:DecentralisedVotingAnalysis}, which 
are designed such that ballots cannot be tallied individually, but are unsuitable 
for large-scale elections.

\subsubsection{Everlasting privacy}

A concern with current electronic voting systems is the longevity of vote privacy. Systems like Helios and Belenios rely on computational assumptions, specifically the hardness of the discrete logarithm problem, for ballot secrecy. Advances in quantum computing or mathematical breakthroughs could eventually render today's encrypted ballots decryptable, motivating \emph{everlasting privacy}: the guarantee that votes remain secret even against computationally unbounded adversaries~\cite{Haines23:EverlastingPrivacySoK}.

Achieving everlasting privacy while maintaining verifiability is challenging, since verifiability requires publishing cryptographic evidence that could eventually be broken. Existing approaches include perfectly hiding commitments, anonymous channels that provide unconditional anonymity, and verifiable secret sharing where destroyed key shares cannot be reconstructed~\cite{Demirel12:EverlastingPrivacy}. The trade-offs between everlasting privacy and efficiency remain an active research area.

\subsubsection{Post-quantum cryptography}

Quantum computers pose a threat to current electronic voting systems. Shor's algorithm~\cite{shor1999polynomial} can efficiently solve the discrete logarithm and integer factorisation problems underlying ElGamal and RSA encryption, problems that would take classical computers thousands of years to solve. While current quantum computers can only factor small numbers, and significant engineering challenges remain, the cryptography community is actively preparing for a post-quantum future.

Several research efforts have developed post-quantum secure electronic voting based on lattice assumptions~\cite{Chillotti16:LWEVoting,Aranha21:LatticeShuffle,Aranha23:LatticeMixnet}. Recent work by Farzaliyev et al.~\cite{Farzaliyev24:LatticeZKVoting} provides comprehensive lattice-based constructions for zero-knowledge proofs of ballot correctness, while Hough et al.~\cite{Hough23:NTRUVoting} achieve significant efficiency improvements using NTRU-based constructions.

The NIST post-quantum cryptography standardisation (2022) provides a foundation for building post-quantum electronic voting systems, but the specialised primitives required, such as homomorphic encryption, verifiable shuffles, and range proofs, need additional research. Transitioning deployed systems like Helios and Belenios to post-quantum security remains an open challenge.

\section{Case studies}\label{sec:casestudies}
Armed with the concepts and techniques from the previous sections, we can examine electronic election schemes that exist in the literature as case studies. Our aim is to demonstrate how the security definitions and proof techniques developed earlier apply to real voting systems, highlighting both their strengths and vulnerabilities.

Although there are various election schemes implemented by private companies, we cannot focus on these as they generally do not present their cryptographic designs in sufficient detail for rigorous analysis.

We will focus on the family of election schemes known as Helios. Helios is particularly well-suited for our case study because it is open-source, well-documented, has been subject to extensive academic scrutiny, and has been deployed in real elections including those of the International Association for Cryptologic Research (IACR).

\subsection{Case study I: Helios}\label{sec:helios}

Helios can be informally modelled as the following election scheme
(further details appear in Figure~\ref{fig:HeliosBallotsAndTallying}):

\newcommand{\helioSetupDescription}[1][]{generates a key pair for an asymmetric #1{}homomorphic encryption scheme, proves correct key generation in zero-knowledge, and outputs the key pair and proof}

\begin{description}

\item $\SetupSymb$ \helioSetupDescription[additively-].

\item $\VoteSymb$ 
enciphers 
the vote's bitstring encoding to a tuple of ciphertexts, proves in zero-knowledge 
that each ciphertext is correctly constructed and that the vote is selected from the sequence of
candidates, and outputs the ciphertexts coupled with the proofs.

\item $\TallySymb$ 
selects ballots from the bulletin board for which proofs hold, homomorphically combines the ciphertexts 
in those ballots,\footnotemark\ decrypts the homomorphic combination to reveal the election 
outcome, and announces the outcome, along with a zero-knowledge proof of correct decryption.

\footnotetext{In Section~\ref{sec:homomorph}, we discuss the importance of homomorphic encryption in election schemes.}

\item $\VerifySymb$ checks the proofs and accepts the outcome if these checks succeed.

\end{description}

\noindent
Helios was first released in 2009 as \emph{Helios~2.0},
and the current release is \emph{Helios~3.1.4}.
A new release was planned but never materialised.\footnote{%
	\url{http://documentation.heliosvoting.org/verification-specs/helios-v4}, published c. 2012, 
	accessed 21 Sep 2017. The specification document remains marked ``NOT FINAL'' and ``IN PROGRESS'' as of 2025, and the planned Fall 2012 release never occurred. Development activity on the main Helios repository has been minimal since approximately 2016.}
Henceforth, we'll refer to the planned release as \emph{\heliosspec}.

\begin{figure}[t!]
Algorithm $\VoteSymb$ inputs a vote $v$ selected from candidates $1,\allowbreak\dots,\allowbreak\nC$
and computes ciphertexts $c_1,\allowbreak\dots,\allowbreak c_{\nC-1}$ such that if $v<\nC$, then 
ciphertext $c_v$ contains plaintext	$1$ and the remaining ciphertexts contain plaintext $0$,
otherwise, all ciphertexts contain plaintext $0$. The algorithm also
computes proofs $\sigma_1,\allowbreak\dots,\allowbreak\sigma_\nC$ demonstrating correct computation.
Proof $\sigma_j$ demonstrates that ciphertext $c_j$ contains $0$ or $1$, where
$1\leq\allowbreak j \leq\allowbreak \nC-1$, and proof $\sigma_{\nC}$ demonstrates that the homomorphic combination	
of ciphertexts $c_1 \otimes\allowbreak\cdots\allowbreak\otimes\allowbreak c_{\nC-1}$ 
contains $0$ or $1$. The algorithm outputs the ciphertexts and proofs. 

Algorithm $\TallySymb$ inputs a bulletin board $\bb$; selects all the ballots 
$b_1,\dots, b_k\in\bb$ for which proofs hold, that is, ballots 
$b_i = \Enc{m_{i,1}},\dots,\Enc{m_{i,{\nC-1}}},\allowbreak\sigma_{i,1},\dots,\sigma_{i,\nC}$ 
such that proofs $\sigma_{i,1},\dots,\sigma_{i,\nC}$ hold, where $1 \leq i \leq k$;
forms a matrix of the encapsulated ciphertexts, that is, 
\begin{align*}
	\begin{array}{r@{\hspace{3mm}}c@{\hspace{3mm}}l}
		\cellcolor{light-gray} \Enc{m_{1,1}}, & \cellcolor{light-gray} \dots, & \cellcolor{light-gray} \Enc{m_{1,\nC-1}} \\
		 \multicolumn{1}{c}{\vdots}  &       & \multicolumn{1}{c}{\vdots} 	   \\
		\cellcolor{light-gray} \Enc{m_{k,1}}, & \cellcolor{light-gray} \dots, & \cellcolor{light-gray} \Enc{m_{k,\nC-1}};
\intertext{homomorphically combines the ciphertexts in each column to derive the encrypted outcome, that is,} 
		\Enc{\Sigma_{i=1}^k m_{i,1}}, & \dots, &  \Enc{\Sigma_{i=1}^k m_{i,\nC-1}};
\intertext{decrypts the homomorphic combinations to reveal the frequency of votes $1,\dots,\nC-1$, that is,}
		\Sigma_{i=1}^k m_{i,1},     & \dots, & \Sigma_{i=1}^k m_{i,\nC-1};
	\end{array}
\end{align*}
computes the frequency of vote $\nC$ by subtracting the frequency of any other vote from the 
number of ballots for which proofs hold, that is, $k - \sum_{j=1}^{\nC-1} \sum_{i=1}^k m_{i,j}$; and 
announces the outcome as those frequencies, along with a proof demonstrating correctness
of decryption.
\caption{Helios: Ballot construction and tallying}
\label{fig:HeliosBallotsAndTallying}
\end{figure}

The proofs $\sigma_1, \ldots, \sigma_{\nC}$ in Figure~\ref{fig:HeliosBallotsAndTallying} are \emph{disjunctive proofs} (or \emph{OR-proofs}): each proves that a ciphertext encrypts $0$ \emph{or} $1$ without revealing which. This is achieved by constructing a real proof for the true case and simulating a proof for the false case---possible because the underlying proof system is zero-knowledge. The Fiat-Shamir transformation (Section~\ref{sec:fiat-shamir}) makes these proofs non-interactive, allowing them to be posted alongside ballots on the bulletin board.

\subsubsection{Helios~2.0}\label{sec:helios:2}
Analysis by Cortier \& Smyth~\cite{Smyth12:Helios,Smyth11:Helios} demonstrates that Helios 2.0 fails to provide ballot secrecy. Building on the specification of Helios 2.0 developed by Smyth, Frink, and Clarkson~\cite{Smyth15:ElectionVerifiability},  Smyth~\cite{Smyth15:BallotSecrecyFull} establishes that $\BallotSecrecy$ cannot be satisfied.

\begin{theorem}\label{thm:Helios2Secrecy}
Helios 2.0 does not satisfy $\BallotSecrecy$. 
\end{theorem}

\noindent 

Cortier \& Smyth~\cite{Smyth12:Helios,Smyth11:Helios} trace the vulnerability to Helios tallying meaningfully related ballots. Specifically, Helios ballots are malleable: starting from a ballot $
  c_1,\allowbreak\dots,\allowbreak c_{\nC-1},\allowbreak
    \sigma_1,\allowbreak\dots,\allowbreak\sigma_{\nC}
$
, 
any permutation $\chi$ on $\{1,\dots,\nC-1\}$ gives
$
  c_{\chi(1)},\allowbreak\dots,\allowbreak c_{\chi(\nC-1)},\allowbreak
    \sigma_{\chi(1)},\allowbreak\dots,\allowbreak\sigma_{\chi(\nC-1)},\allowbreak\sigma_{\nC}
$. 
Thus, ballots are malleable, which is incompatible with ballot 
secrecy (\S\ref{sec:secrecy:simpleProofs}).

\begin{proof}[Proof sketch]
Suppose an adversary queries the oracle with 
(distinct) inputs $v_0,v_1\in\{1,\dots,\nC-1\}$ to derive a 
ballot for $v_\beta$, where integer $\nC \geq 3$ is chosen by the adversary and
bit $\beta$ is chosen by the challenger.
Further suppose the adversary picks a permutation $\chi$ on $\{1,\dots,\nC-1\}$,
abuses malleability to derive a related ballot $b$ for $\chi(v_\beta)$, and
outputs bulletin board $\{b\}$.
The board is balanced, because it does not contain the ballot output
by the oracle.
Suppose the adversary performs the following computation 
on input of election outcome $\outcome$: if $\outcome[\chi(v_0)] = 1$, then 
output $0$, otherwise, output $1$.
Since $b$ is a ballot for $\chi(v_\beta)$, it follows by correctness that
$\outcome[\chi(v_0)] = 1$ iff $\beta = 0$, and $\outcome[\chi(v_1)] = 1$ iff $\beta = 1$,
hence, the adversary wins the game.
\end{proof}

\noindent
For simplicity,  the  proof sketch considers an adversary that omits ballots 
from the bulletin board. Voters might detect such an adversary, because Helios
satisfies individual verifiability, hence, voters can discover if their ballot 
is omitted. The proof sketch can be extended to avoid such detection: 
Let $b_1$ be the ballot output by the oracle in the proof sketch and suppose 
$b_2$ is the ballot output by a (second) oracle query with inputs 
$v_1$ and $v_0$. Further suppose the adversary outputs (the balanced) bulletin 
board $\{b,b_1,b_2\}$ and performs the following computation on input of  
election outcome $\outcome$: 
if $\outcome$ corresponds to votes $v_0,v_1,\chi(v_0)$, then output $0$, otherwise, 
output $1$, where $\chi$ is the permutation chosen by the adversary.
Hence, the adversary wins the game.

Chang-Fong \& Essex show that Helios 2.0 does not satisfy universal
verifiability~\cite[Section~4.1]{Essex16:HeliosVerifiability}, and
Smyth, Frink \& Clarkson use their result to prove that the completeness aspect of
$\UV$ is not satisfied~\cite{Smyth15:ElectionVerifiability}.\footnote{%
  Chang-Fong \& Essex present a vulnerability~\cite[Section~4.2]{Essex16:HeliosVerifiability}
  that should violate $\Soundness$. We leave the proof of this result
  as an exercise for the reader.}

\begin{theorem}\label{thm:helios:completeness}
Helios 2.0 does not satisfy $\Completeness$.
\end{theorem}

\noindent
Chang-Fong \& Essex attribute the vulnerability to not checking the suitability of 
cryptographic parameters nor 
checking that ballots are constructed from such parameters.

\begin{proof}[Proof sketch]
Suppose an adversary computes a ciphertext and masks a term of that ciphertext. 
Moreover, suppose the adversary falsifies a proof of correct construction in a 
manner that hides malice. In particular, the adversary computes the proof such 
that an exponent will evaluate to zero during verification, which causes cancellation 
of the mask. (This is possible because verification does not check that ballots 
are constructed from suitable cryptographic parameters.)
Suppose the adversary computes a bulletin board containing the masked 
ciphertext and proof. Moreover, suppose that the challenger tallies that 
board. The masked ciphertext will be homomorphically combined with other
ciphertexts and decrypted, because the proof holds. Yet, the proof of correct 
decryption constructed by the challenger will fail, due to the masked
ciphertext, hence, the adversary wins the game.
\end{proof}

\noindent
The vulnerability was mitigated against in Helios 3.1.4 
by performing the necessary checks.

\subsubsection{Helios 3.1.4}\label{sec:helios:3}
Ballots remain malleable in Helios 3.1.4, hence, ballot secrecy is not satisfied,
and Smyth~\cite{Smyth15:BallotSecrecyFull} proves that $\BallotSecrecy$ is not satisfied, using the formal description of Helios~3.1.4 by Smyth, Frink \& Clarkson~\cite{Smyth15:ElectionVerifiability}.

\begin{corollary}\label{cor:Helios3Secrecy}
Helios 3.1.4 does not satisfy $\BallotSecrecy$.
\end{corollary}

\noindent
A proof of Corollary~\ref{cor:Helios3Secrecy} follows from Theorem~\ref{thm:Helios2Secrecy}, 
because 
Helios 3.1.4 does not address issues arising from related ballots.

Bernhard, Pereira \& Warinschi 
show that Helios~3.1.4 does not satisfy universal verifiability~\cite[Section~3]{Bernhard12:Helios},
and Smyth, Frink \& Clarkson use their result to prove that the soundness aspect of $\UV$ is not satisfied.\footnote{%
  Bernhard, Pereira \& Warinschi present a vulnerability~\cite[p632]{Essex16:HeliosVerifiability} 
  that should violate $\Completeness$. Again, the proof of this result 
  is left as another exercise for the reader.}

\begin{theorem}\label{thm:Helios2}
Helios 3.1.4 does not satisfy $\Soundness$. 
\end{theorem}

\noindent
Bernhard \emph{et al.}
attribute vulnerabilities to application of the Fiat-Shamir
transformation without inclusion of statements in hashes (that is, weak Fiat-Shamir).

\begin{proof}[Proof sketch]
Suppose an adversary partially computes a proof of ciphertext construction,
before computing a ciphertext and without computing a key pair. In particular,
suppose the adversary computes the challenge hash. (This is possible because
weak Fiat-Shamir does not include statements in hashes, hence, ciphertexts
are not included in hashes.) Further suppose the adversary computes a private
key as a function of that hash, challenges as functions of the hash and the 
private key, and responses as functions of the challenges and some coins.
Moreover, suppose the adversary computes a public key (from the private key) 
and a proof of correct key generation. That proof is valid, because the private 
key could have been correctly computed. Suppose the adversary encrypts some
plaintext $m$ (such that $m>1$) to a ciphertext, using the aforementioned coins. Further suppose
the adversary proves correct decryption of that ciphertext. That proof is valid,
because the ciphertext is well-formed. Finally, suppose the adversary claims
$(m,m-1)$ is the election outcome corresponding to the ballot containing the 
ciphertext and falsified proof of correct construction.
The verification procedure will accept that outcome, because all proofs hold,
yet the election outcome is clearly invalid, hence, the adversary wins the game.
\end{proof}

\subsubsection{\heliosspec}\label{sec:helios:spec}

\heliosspec\ is intended to mitigate against vulnerabilities.
In particular, this specification incorporates the Fiat-Shamir transformation, rather than weak Fiat-Shamir. There are plans to incorporate what is called \emph{ballot weeding}, that is, to omit meaningfully related ballots from tallying.
Smyth, Frink \& Clarkson show that \heliosspec\ does not satisfy universal verifiability~\cite{Smyth15:ElectionVerifiability}, and Smyth shows that ballot secrecy is not satisfied either~\cite{Smyth15:BallotSecrecyFull}.

\begin{remark}\label{remark:HeliosSpec}
\heliosspec\ does not satisfy $\Soundness$.
\end{remark}
\begin{proof}[Proof sketch]
Suppose an adversary constructs a ballot, 
abuses malleability to derive a related ballot, and tallies both ballots. 
Ballot weeding will omit at least
one of those ballots. (\heliosspec\ does not yet define a particular
ballot weeding mechanism, hence, the precise behaviour is unknown. 
Nonetheless, we are assured that at least one ballot will be omitted, 
because the ballots are related.) Hence, tallying produces an election
outcome that omits a vote, which soundness forbids, thus, the 
adversary wins the game.
\end{proof}

\begin{remark}\label{rem:helios3}
\heliosspec\ does not satisfy $\BallotSecrecy$.
\end{remark}

\begin{proof}[Proof sketch]
Neither ballot weeding nor the Fiat-Shamir transformation
eliminate the vulnerability we identified in Helios 3.1.4,
hence, we conclude by the proof sketch of Theorem~\ref{thm:Helios2Secrecy}.
\end{proof}

We point out that Remarks~\ref{remark:HeliosSpec} \&~\ref{rem:helios3} are stated informally, because
there is no formal description of \heliosspec. Such a description can be derived as a straightforward variant of Helios 3.1.4 that uses ballot weeding and applies the Fiat-Shamir transformation (rather than the weak Fiat-Shamir transformation). But arguably these details provide little value,
so we do not pursue them.

\subsubsection{\heliosnext}\label{sec:helios:next}

Smyth, Frink \& Clarkson~\cite{Smyth15:ElectionVerifiability} propose \heliosnext, 
a variant of Helios that uses the Fiat-Shamir transformation and non-malleable 
ballots, to overcome the aforementioned vulnerabilities. They prove that \heliosnext\ 
satisfies verifiability, and Smyth~\cite{Smyth15:BallotSecrecyFull} proves that ballot secrecy is satisfied too.

\begin{theorem}\label{thm:heliosNextIVUV}
\begin{sloppypar}
\heliosnext\ satisfies both $\IV$ and $\UV$.
\end{sloppypar}
\end{theorem}

\begin{proof}[Proof sketch]
Smyth \emph{et al.}~\cite{Smyth15:ElectionVerifiability,Smyth18:HeliosMixnetVerifiability}
prove that ElGamal
produces ciphertexts that do not collide for correctly generated keys.
Hence, \heliosnext\ ballots do not collide, because they contain
ElGamal ciphertexts constructed using such keys. Thus, \heliosnext\
satisfies $\IV$. Smyth, Frink \& Clarkson also prove that $\UV$
is satisfied. Their proof shows that tallying discards ill-formed 
ballots and that the remaining ballots all contain ciphertexts that 
encipher bitstring encodings of votes, hence, the homomorphic combination
of those ciphertexts contain the encrypted outcome, which is decrypted
to reveal the correct outcome ($\Soundness$). Moreover, they show that 
such outcomes are always accepted ($\Completeness$). 
\end{proof}

\begin{theorem}\label{thm:heliosNext}
\heliosnext\ satisfies $\BallotSecrecy$.
\end{theorem}

\begin{proof}[Proof sketch]
\begin{sloppypar}
Smyth proves that \heliosnext\ satisfies $\BallotIndependence$~\cite[Proposition~21]{Smyth15:BallotSecrecyFull}
and Smyth, Frink \& Clarkson prove that $\UV$ is satisfied too 
(Theorem~\ref{thm:heliosNextIVUV}), moreover, Smyth proves that 
\heliosnext\ uses zero-knowledge tallying proofs, which suffices 
for $\BallotSecrecy$ (\S\ref{sec:secrecy:simpleProofs}).\qedhere
\end{sloppypar}%
\end{proof}

\noindent
These results (summarised in Table~\ref{table:HeliosSummary}) provide strong 
motivation for future Helios releases being based upon \heliosnext, since it 
is the only variant of Helios which is proven to satisfy both ballot secrecy 
and verifiability. (As a side remark, beyond secrecy and verifiability, eligibility, which we mentioned in Section~\ref{sec:intro}, is 
  known not to be satisfied~\cite{2013-truncation-attacks-to-violate-beliefs,2014-truncation-attacks-to-violate-beliefs,Smyth16:helios-eligibility}.)

\begin{table}
\begin{center}
\begin{tabular}{ l || c | c | c| c}
        & Helios 2.0 & Helios 3.1.4 & \heliosspec & \heliosnext \\ \hline
Ballot secrecy 	          & \xmark & \xmark & \xmark & \cmark
\\ \hline
Individual verifiability  & \cmark & \cmark & \cmark & \cmark
\\ \hline
Universal verifiability   & \xmark & \xmark & \xmark & \cmark
\end{tabular}
\end{center}

\footnotesize{%
Cortier \& Smyth identify a secrecy vulnerability in Helios 2.0 
and Helios 3.1.4~\cite{Smyth12:Helios}, and Smyth shows the 
vulnerability is exploitable in \heliosspec\ when the adversary 
controls ballot collection~\cite{Smyth15:BallotSecrecyFull}. Moreover, 
Smyth proves that \heliosnext\ satisfies 
ballot secrecy. Bernhard, 
Pereira \& Warinschi identify universal-verifiability vulnerabilities 
in Helios 2.0 and Helios 3.1.4~\cite{Bernhard12:Helios}, Chang-Fong 
\& Essex identify vulnerabilities in Helios 2.0~\cite{Essex16:HeliosVerifiability}, 
and Smyth, Frink, \& Clarkson identify a vulnerability in 
\heliosspec~\cite{Smyth15:ElectionVerifiability}. Moreover, Smyth, 
Frink, \& Clarkson prove that \heliosnext\ satisfies individual and 
universal verifiability.
}
\caption{Summary of Helios security results}
\label{table:HeliosSummary}
\end{table}

\paragraph{A note on real-world deployment.} Despite the cryptographic vulnerabilities identified in the research literature, Helios has continued to be used for various elections. The International Association for Cryptologic Research (IACR) has used Helios annually since 2010 to elect board members. However, in 2025, a trustee lost their private key during the IACR election, making it impossible to decrypt the results. The IACR voided the election and subsequently implemented additional safeguards, including a 2-out-of-3 key threshold mechanism. This incident illustrates that even cryptographically sound systems can fail due to operational issues, highlighting the importance of key management and threshold cryptography in practice.

\subsection{Case study II: Helios Mixnet}\label{sec:heliosMixnet}

Helios Mixnet can be informally modelled as the following election scheme:

\begin{description}

\item $\SetupSymb$ 
  \helioSetupDescription.

\item $\VoteSymb$ enciphers the vote to a ciphertext, proves correct ciphertext construction in 
	zero-knowledge, and outputs the ciphertext coupled with the proof.

\item $\TallySymb$ 
	selects ballots from the bulletin board for which proofs hold, mixes the ciphertexts in those ballots, decrypts the ciphertexts output by the mix to reveal the election outcome (that is, the frequency distribution of votes) and any ill-formed votes (that is, votes that are not selected from the sequence of candidates), and announces that outcome, along with zero-knowledge proofs demonstrating correct decryption.

\item $\VerifySymb$ checks the proofs and accepts the outcome if these checks succeed.

\end{description}

\noindent
We have seen no implementation of Helios Mixnet by either Adida~\cite{Adida08} or Bulens, Giry \& Pereira~\cite{Pereira11:HeliosMixnet}.
Tsoukalas \emph{et al.}~\cite{Tsoukalas13:HeliosToZeus} released \emph{Zeus} as an independent variant (or fork) of Helios spliced with mixnet code to derive such an implementation. 
Building upon this, Yingtong Li released \emph{helios-server-mixnet} as an extension of Zeus with threshold asymmetric encryption and some other minor changes.

As noted before in Theorem~\ref{thm:helios:completeness}, Helios 2.0 does not satisfy completeness, which means that any implementations 
of Helios Mixnet did not satisfy completeness until Helios 
was patched because such implementations build off the Helios code and 
do not add code to check cryptographic parameters. In addition to the lack of completeness, Smyth~\cite{Smyth18:HeliosMixnetVerifiability}  
identified a soundness vulnerability in Helios Mixnet.  

\begin{remark}\label{rem:heliosM}
Zeus does not satisfy $\Soundness$.
\end{remark}

\noindent
Smyth traces the vulnerability to the weak Fiat-Shamir transformation. This vulnerability was reported to the developers of Zeus and helios-server-mixnet in 2018, who promptly adopted and deployed the proposed fix~\cite[Section~4]{Smyth18:HeliosMixnetVerifiability}.

\begin{proof}[Proof sketch]
We use the term subdistribution to refer to an incomplete frequency distribution of votes, meaning one that accounts for only a subset of the actual votes cast.
For the proof sketch, suppose an adversary constructs some ballots and mixes
the ciphertexts in those ballots. Further suppose the 
adversary decrypts the ciphertexts output by the mix to 
reveal the frequency distribution of votes and then selects some ciphertexts 
that decrypt to a (strict) subdistribution. The adversary can then prove correct 
decryption of those ciphertexts and falsify proofs of 
the remaining ciphertexts decrypting to arbitrary elements of 
the message space (by exploiting a vulnerability against
Helios~\cite{Bernhard12:Helios} due to the 
weak Fiat-Shamir transformation). 
Finally, suppose the adversary claims the 
subdistribution of votes is the election outcome. The verification 
procedure will accept that outcome, because all proofs hold,
yet the election outcome excludes votes, hence, the adversary 
wins the game.
\end{proof}

\noindent
Similarly, voting system helios-server-mixnet does not satisfy 
$\Soundness$ when a $(n,n)$-threshold is 
used~\cite{Smyth18:HeliosMixnetVerifiability}. 

\begin{sloppypar}
Smyth proposes a formal description of Helios Mixnet that uses the Fiat-Shamir
transformation and proves that $\BallotSecrecy$, $\IV$, and $\UV$ are 
satisfied~\cite{Smyth15:BallotSecrecyFull,Smyth17:HeliosMixnetVerifiability}.
\end{sloppypar}

\begin{theorem}
Helios Mixnet satisfies both $\IV$ and $\UV$.
\end{theorem}

\begin{proof}[Proof sketch]
\begin{sloppypar}
In the Helios system, ballots do not collide~\ref{thm:heliosNextIVUV}, because they 
contain ElGamal ciphertexts constructed using correctly generated keys, which means that $\IV$ is satisfied. 
This also means  $\UV$ is satisfied, 
because tallying discards ill-formed ballots and votes, hence the mix 
gives the correct outcome (corresponding to $\Soundness$), and such outcomes are 
always accepted (corresponding to $\Completeness$).\qedhere
\end{sloppypar}
\end{proof}

\begin{theorem}
Helios Mixnet satisfies $\BallotSecrecy$.
\end{theorem}

\begin{proof}[Proof sketch]
Smyth~\cite{Smyth18:HeliosMixnetVerifiability} demonstrates that Helios Mixnet can be 
derived from $\encToVote[\allowbreak\Pi]$ by using suitable tallying and verification 
algorithms. Smyth also proves that $\UV$ is satisfied, which is enough for 
$\BallotSecrecy$ (\S\ref{sec:secrecy:simpleProofs}), assuming $\Pi$ satisfies
$\INDPA$ and well-definedness.
\end{proof}
\subsection{Case study III: Belenios}\label{sec:belenios}

Belenios is a secret verifiable electronic voting system developed by Cortier, Gaudry, and Glondu~\cite{Cortier19:Belenios}, building upon the Helios codebase but addressing several of its security limitations. The system has been deployed in hundreds of real-world elections since its creation, including academic institutions, associations, and notably in experiments during French legislative elections~\cite{Cortier23:FrenchElections}.








\subsubsection{Key differences from Helios}

The most significant change in Belenios with respect to Helios is the addition of digital signatures to attest that ballots come from eligible voters associated with a given credential~\cite{Baloglu21:HeliosBelenios}. This addresses a fundamental limitation of Helios: vulnerability to ballot stuffing by a malicious bulletin board.

In Helios, a dishonest bulletin board could add ballots without anyone noticing, because there is no mechanism to verify that ballots originated from legitimate voters. Belenios provides unforgeability: nobody can forge a fake signature.

\subsubsection{Limitations and extensions}

Like Helios, Belenios is not coercion-resistant: voters may prove how they voted by revealing the randomness used to produce their ballot, or they may sell their voting credentials. 

Belenios does not ensure \emph{participation privacy}: the publicly available election data reveals whether a particular voter participated in the election. Although this information is typically available in traditional paper-based elections (anyone can observe people entering a polling station), an Internet voting system without participation privacy reveals voter identities on a much larger scale by publishing them online~\cite{KTV15:HeliosPrivacy}.

Several extensions to Belenios have been proposed. \emph{BeleniosVS}~\cite{Cortier19:BeleniosVS} provides secrecy and verifiability even against a corrupted voting device. Recent work has added \emph{cast-as-intended} verification mechanisms~\cite{Cortier24:BeleniosCastAsIntended}, allowing voters to verify that their device correctly encoded their intended vote. In 2024, Belenios underwent a certification campaign in France~\cite{Cortier24:BeleniosCertification}, providing additional assurance for its use in official elections.

\subsubsection{Comparison with Helios}

Table~\ref{table:HeliosBeleniosComparison} summarises the key differences between Helios and Belenios.

\begin{table}[h!]
\begin{center}
\begin{tabular}{ l || c | c }
        & Helios & Belenios \\ \hline
Ballot secrecy 	          & \xmark$^*$ & \cmark \\ \hline
Individual verifiability  & \cmark & \cmark \\ \hline
Universal verifiability   & \xmark$^*$ & \cmark \\ \hline
Eligibility verifiability & \xmark & \cmark \\ \hline
Resistance to ballot stuffing & \xmark & \cmark \\ \hline
Participation privacy & \xmark & \xmark \\ \hline
Receipt-freeness & \xmark & \xmark \\ \hline
Coercion resistance & \xmark & \xmark \\
\end{tabular}
\end{center}

\footnotesize{%
$^*$As discussed in Section~\ref{sec:helios}, existing Helios implementations (2.0, 3.1.4) do not satisfy these properties, though the proposed \heliosnext\ variant does.
}
\caption{Comparison of Helios and Belenios security properties}
\label{table:HeliosBeleniosComparison}
\end{table}

The primary advantage of Belenios over Helios is eligibility verifiability and resistance to ballot stuffing attacks. Both systems share similar limitations regarding coercion resistance and participation privacy, which remain active areas of research.

\section{Summary}\label{sec:summary}
We have introduced and detailed concepts from cryptography that can be used in designing election schemes. We have given a brief introduction to game-based cryptography, which we used to formulate elections as games. Our emphasis has been on studying definitions of secrecy and verifiability, which we formally defined with cryptography games.  

For definition and illustration purposes, we have examined a proposed election scheme based on four building blocks, namely the primitives  $\SetupSymb$, $\VoteSymb$, $\TallySymb$, and $\VerifySymb$, as illustrated in Figure~\ref{figure-simplevoting} and, more precisely, in Figure~\ref{figure-voting}. We have detailed these primitives or algorithms in Definition~\ref{def:election}, which we now briefly recall. $\SetupSymb$ randomly generates a public and private key pair, where the voters use the public key and the tallier uses the private key. $\VoteSymb$ allows the voters to vote using the public key. $\TallySymb$ adds up the encrypted votes, which is possible with a homomorphic asymmetric encryption scheme. $\VerifySymb$ uses the public key or keys from the primitive $\SetupSymb$ and a proof generated by the primitive $\TallySymb$ to verify the election outcome is correct.  

We have then stated what we mean exactly by secrecy and verifiability by introducing further concepts. 

We have defined election correctness as requiring that honestly executed elections produce outcomes that match the cast votes with overwhelming probability, ensuring the scheme functions as intended when all parties follow the protocol correctly.

We have defined soundness as requiring the $\VerifySymb$ algorithm to only accept correct outcomes; Definition~\ref{def:UVSoundness}. To complement soundness, we have defined completeness as requiring election outcomes (that are produced by algorithm $\TallySymb$) to be accepted by algorithm $\VerifySymb$;   Definition~\ref{def:UVComplete}. We have then defined  universal verifiability as a combination of soundness and completeness; 	Definition~\ref{def:UV}. 

Using the aforementioned concepts, we have then examined the voting schemes Helios, Helios Mixnet, and Belenios through detailed case studies.

\subsection{Other voting systems}
Beyond the Helios family, several other electronic voting systems have received formal security analysis.

\emph{JCJ}~\cite{JCJ05} introduced the first formal treatment of coercion-resistant electronic voting, allowing voters to produce fake credentials that are indistinguishable from real ones, thereby enabling them to deceive coercers. \emph{Civitas}~\cite{CCM08} implements the JCJ protocol with practical improvements, and its source code is publicly available. However, the JCJ/Civitas approach has quadratic complexity in the number of submitted ballots during the tallying phase, limiting its applicability to large-scale elections. Moreover, the original JCJ definition of coercion resistance appears to require a patch to be satisfiable~\cite{HainesSmyth19:CoercionSoK}, and recent work by Cortier et al.~\cite{Cortier24:JCJ} has identified potential weaknesses in the ballot cleansing procedure that may leak information to coercers.

\emph{Athena}~\cite{Smyth19:Athena} advances the JCJ paradigm by delivering a verifiable, coercion-resistant voting system with linear complexity, addressing the scalability limitations of JCJ/Civitas. The scheme preserves the fake credential approach that underpins coercion resistance while restructuring the cryptographic operations to avoid the quadratic cost of credential validation. This makes coercion-resistant remote voting more practical for elections with large numbers of voters.

\emph{Pr\^et \`a Voter}~\cite{ChaumRyanSchneider2005} uses paper-based cryptographic mechanisms with mix-network tallying and has been studied extensively for ballot secrecy and verifiability. \emph{Scantegrity}~\cite{Chaum08,Chaum08II} provides end-to-end verifiability for optical scan systems using invisible ink confirmation codes. We do not provide the same depth of analysis for these systems, but refer interested readers to the cited works and the systematization by Cortier et al.~\cite{Cortier16:VerifiabilitySoK}, which surveys verifiability notions across multiple voting protocols.

\subsection{Related tutorials and surveys}
Several other works provide introductions to electronic voting security. Bernhard et al.~\cite{Bernhard14:voting-tutorial} give a gentle introduction to cryptographic voting aimed at non-specialists. The systematizations of knowledge by Bernhard et al.~\cite{BCGPW15} on ballot privacy definitions and Cortier et al.~\cite{Cortier16:VerifiabilitySoK} on verifiability notions provide comprehensive treatments of the definitional landscape. For readers interested in the formal methods approach to analyzing voting protocols, Delaune et al.~\cite{DKR08} develop techniques based on the applied pi calculus.

\section{Conclusion}\label{sec:conclusion}
To better understand electronic voting, we have used a game-based approach coupled with formal definitions, such as soundness and completeness, that shows how third parties can use  these tools to detect subtle vulnerabilities in voting systems. We have seen how analysis drives development and ultimately leads to systems that are provably secure. This clearly demonstrates the need for security definitions coupled with analysis to ensure security of voting systems. 

For such a dynamic topic as electronic voting, it is impossible to give an up-to-date and complete overview, but this tutorial offers a good starting point for interested students, researchers, practitioners, and readers. We hope this \manuscript\ advances the reader's understanding of these topics and themes. Overall, we want to inspire more research and work in electronic voting, with the ultimate aim of enabling democratic institutions with secure voting schemes. 

\appendix

\section{List of Symbols}\label{app:notation}

The following tables collect the notation used throughout this \manuscript{}, grouped by theme.

\bigskip

\noindent\textbf{Cryptographic primitives}
\medskip

\noindent
\begin{tabular}{@{}l@{\quad}p{0.78\textwidth}@{}}
$\Pi$ & Encryption scheme, consisting of the tuple $(\GenSymb,\EncSymb,\DecSymb)$. \\[4pt]
$\GenSymb$ & Key generation algorithm for an encryption scheme. \\[4pt]
$\EncSymb$ & Encryption algorithm. \\[4pt]
$\DecSymb$ & Decryption algorithm. \\[4pt]
$\pk$ & Public key for an asymmetric encryption scheme. \\[4pt]
$\sk$ & Private (secret) key for an asymmetric encryption scheme. \\[4pt]
$\votespace$ & Message space (space of plaintexts) for an encryption scheme. \\[4pt]
\end{tabular}

\bigskip

\noindent\textbf{Election scheme}
\medskip

\noindent
\begin{tabular}{@{}l@{\quad}p{0.78\textwidth}@{}}
$\Gamma$ & Election scheme, a tuple $(\SetupSymb, \VoteSymb, \TallySymb, \VerifySymb)$. \\[4pt]
$\SetupSymb$ & Setup algorithm that generates keys and election parameters. \\[4pt]
$\VoteSymb$ & Vote algorithm that constructs a ballot from a vote. \\[4pt]
$\TallySymb$ & Tally algorithm that computes the election outcome. \\[4pt]
$\VerifySymb$ & Verify algorithm that audits the election outcome. \\[4pt]
$\encToVoteSymb$ & Construction that builds an election scheme from an encryption scheme. \\[4pt]
\end{tabular}

\bigskip

\noindent\textbf{Election parameters}
\medskip

\noindent
\begin{tabular}{@{}l@{\quad}p{0.78\textwidth}@{}}
$\kk$ & Security parameter. \\[4pt]
$\nC$ & Number of candidates in an election. \\[4pt]
$\mC$ & Maximum number of candidates. \\[4pt]
$\nB$ & Number of ballots in an election. \\[4pt]
$\mB$ & Maximum number of ballots. \\[4pt]
\end{tabular}

\bigskip

\noindent\textbf{Votes and ballots}
\medskip

\noindent
\begin{tabular}{@{}l@{\quad}p{0.78\textwidth}@{}}
$v$ & A vote (selected from candidates $1, \ldots, \nC$). \\[4pt]
$b$ & A ballot (encrypted vote). \\[4pt]
$\bb$ & Bulletin board (collection of ballots). \\[4pt]
$\outcome$ & Election outcome (vote tallies for each candidate). \\[4pt]
$\tpf$ & Tallying proof demonstrating correctness of the outcome. \\[4pt]
$\auditoutcome$ & Audit outcome. \\[4pt]
$\correcttally$ & Function mapping a bulletin board to the vector of vote counts per candidate. \\[4pt]
$\balanced$ & Predicate that holds when vote vectors are balanced (equal totals). \\[4pt]
\end{tabular}

\bigskip

\noindent\textbf{Assignments and general notation}
\medskip

\noindent
\begin{tabular}{@{}l@{\quad}p{0.78\textwidth}@{}}
$\leftarrow$ & Deterministic assignment. \\[4pt]
$\leftarrow_R$ & Uniformly random sampling from a finite set. \\[4pt]
$\adv$ & Adversary in a security game. \\[4pt]
$\Adv$ & Second adversary or reduction algorithm in a security proof. \\[4pt]
$\oracle$ & Oracle in a security game. \\[4pt]
$\Succ(\cdot)$ & Success probability of a game (probability of outputting $\top$). \\[4pt]
$\negl$ & Negligible function. \\[4pt]
\end{tabular}

\bigskip

\noindent\textbf{Security properties}
\medskip

\noindent
\begin{tabular}{@{}l@{\quad}p{0.78\textwidth}@{}}
$\INDCPA$ & IND-CPA: indistinguishability under chosen-plaintext attack. \\[4pt]
$\INDPA$ & IND-PA0: indistinguishability under plaintext-awareness. \\[4pt]
$\CNMCPA$ & CNM-CPA: ciphertext non-malleability under chosen-plaintext attack. \\[4pt]
$\BallotSecrecy$ & Ballot secrecy game. \\[4pt]
$\BallotIndependence$ & Ballot independence game (IND-CVA). \\[4pt]
$\IV$ & Individual verifiability. \\[4pt]
$\UV$ & Universal verifiability. \\[4pt]
$\Completeness$ & Completeness (honest outcomes pass verification). \\[4pt]
$\Soundness$ & Soundness (verification rejects incorrect outcomes). \\[4pt]
\end{tabular}

\section{Glossary of Voting Terms}\label{app:glossary}

The following table defines voting-specific terminology used throughout this \manuscript{}.

\bigskip

\noindent\textbf{Core verifiability concepts}
\medskip

\noindent
\begin{tabular}{@{}l@{\quad}p{0.78\textwidth}@{}}
\emph{Verifiability} & The ability to prove that no undue influence has occurred in an election and that the outcome correctly reflects the votes cast. \\[4pt]

\emph{Individual verifiability}~\cite{Smyth15:ElectionVerifiability} & A voter can check whether their ballot is collected and included in the election record. Formalized as the requirement that ballots do not collide. \\[4pt]

\emph{Universal verifiability}~\cite{Smyth15:ElectionVerifiability} & Anyone can check whether an election outcome corresponds to the votes expressed in collected ballots. Formalized as the combination of soundness and completeness. \\[4pt]
\end{tabular}

\bigskip

\noindent\textbf{Components of universal verifiability}
\medskip

\noindent
\begin{tabular}{@{}l@{\quad}p{0.78\textwidth}@{}}
\emph{Soundness}~\cite{Smyth15:ElectionVerifiability} & The verification algorithm only accepts outcomes that correspond to votes expressed in collected ballots. An adversary cannot produce a proof for an incorrect outcome that passes verification. \\[4pt]

\emph{Completeness}~\cite{Smyth15:ElectionVerifiability} & The verification algorithm accepts outcomes computed honestly by the tallying algorithm. Ensures that correctly computed election outcomes are not rejected. \\[4pt]

\emph{Injectivity}~\cite{Smyth15:ElectionVerifiability} & Ballots for distinct votes never collide. Ensures that each ballot can be interpreted as encoding at most one vote. \\[4pt]
\end{tabular}

\bigskip

\noindent\textbf{Core security properties}
\medskip

\noindent
\begin{tabular}{@{}l@{\quad}p{0.78\textwidth}@{}}
\emph{Ballot secrecy}~\cite{Smyth15:BallotSecrecyFull} & A voter's vote is not revealed to anyone, including election organizers. Formalized as the inability to distinguish between an instance of the voting system in which voters cast some votes from another instance in which the voters cast a permutation of those votes. \\[4pt]
\end{tabular}

\bigskip

\noindent\textbf{Fine-grained verifiability}
\medskip

\noindent
\begin{tabular}{@{}l@{\quad}p{0.78\textwidth}@{}}
\emph{End-to-end verifiability} & The combination of cast-as-intended, stored-as-cast, and tallied-as-stored verifiability. Provides comprehensive verification throughout the voting process. \\[4pt]

\emph{Cast-as-intended}~\cite{Cortier16:VerifiabilitySoK} & A voter can verify that their voting device correctly encoded their intended vote into a ballot. \\[4pt]

\emph{Stored-as-cast}~\cite{Cortier16:VerifiabilitySoK} & A voter can verify that their ballot was correctly stored on the bulletin board as they cast it. \\[4pt]

\emph{Tallied-as-stored}~\cite{Cortier16:VerifiabilitySoK} & Anyone can verify that the ballots stored on the bulletin board were correctly tallied to produce the election outcome. \\[4pt]
\end{tabular}

\bigskip

\noindent\textbf{Related properties}
\medskip

\noindent
\begin{tabular}{@{}l@{\quad}p{0.78\textwidth}@{}}
\emph{Unforgeability} & Only authorized voters can construct valid ballots. Prevents adversaries from injecting fraudulent ballots into the election. \\[4pt]

\emph{Eligibility verifiability}~\cite{Smyth15:ElectionVerifiability} & Deprecated term, used as a synonym for unforgeability in earlier work. Current terminology prefers unforgeability. \\[4pt]

\emph{Ballot independence} & Ballots constructed by honest voters do not leak information about votes cast by other honest voters. Related to the IND-CVA security game. \\[4pt]

\emph{Final-agreement}~\cite{Hirschi21:BulletinBoard} & A property ensuring that all participants eventually agree on the final state of the bulletin board. Necessary for verifiability with non-idealized bulletin boards. \\[4pt]
\end{tabular}

\section{Asymmetric encryption}\label{app:sec:asym}
We briefly described public-key or asymmetric encryption in Section~\ref{sec:asym}. Now we give a more precise definition, which can be found in cryptography textbooks, such as Katz and Lindell~\cite{katz2014introduction}. 

\begin{definition}[Asymmetric encryption scheme~\cite{katz2014introduction,Smyth15:ElectionVerifiability}]\label{def:enc}
An \emph{asymmetric encryption scheme} is a tuple of probabilistic polynomial-time algorithms 
$(\GenSymb,\EncSymb,\allowbreak\DecSymb)$, such that:

\begin{itemize}
\item \emph{\textbf{Gen}}, denoted $(\pk,\sk, \mathcal{M} )\leftarrow \Gen$, inputs a security parameter 
$\kk$ and outputs a key pair $(\pk,\sk)$ and message space $\mathcal{M}$.

\item \emph{\textbf{Enc}}, denoted $c\leftarrow \Enc{m}$, inputs a public key $\pk$ and message $m\in\mathcal{M}$, and outputs a ciphertext $c$. 

\item \emph{\textbf{Dec}}, denoted $m\leftarrow \Dec{c}$, inputs a private key $\sk$ and ciphertext $c$, and outputs a message $m$ or an error symbol. We assume $\DecSymb$ is deterministic.

\end{itemize}

\noindent
Moreover, the scheme must be \emph{correct}: there exists a negligible function
$\negl$, such that for all security parameters $\kk$ and messages $m$, we have
$$\Pr[
	(\pk,\sk,\mathcal{M}) \allowbreak\leftarrow\Gen; 
	c\leftarrow \Enc{m}:
	m\in\mathcal{M} \Rightarrow
	\Dec[\sk]{c} = {m}] > 1-\neglK  \,.$$
\end{definition}

\section{Definition of $\INDPA$}\label{sec:INDPA}
In Section~\ref{sec:indgame} we quickly discussed the importance of indistinguishability, which we defined with a game. Now we give a definition of \emph{indistinguishability under parallel attack}, as defined by Bellare \& Sahai~\cite{Bellare99:IND-k-CPA}, which we used in Section~\ref{sec:secrecy:example}. 

\begin{definition}[$\INDPA$~\cite{Bellare99:IND-k-CPA}]\label{def:INDPA}
Let $\Pi = (\GenSymb,\EncSymb,\allowbreak\DecSymb)$ be 
an asymmetric encryption scheme, $\adv$ be an adversary, $\kk$ be the security parameter,
and $\gameINDPA$ be the following game.

{\upshape
\begin{inlineexperiment}{$\gameINDPA$}
			$(\pk,\sk,\mathcal{M}) \leftarrow \Gen$\;
			$(m_0,m_1) \leftarrow \adv(\pk,\mathcal{M},\kk)$\;
			$\beta \leftarrow_R \{0,1\}$\;
			$c \leftarrow \Enc{m_\beta}$\;
			${\bf c} \leftarrow \adv(c)$\;
			${\bf m} \leftarrow (\Dec{{\bf c}[1]},\dots,\Dec{{\bf c}[|{\bf c}|]}$\;
			$g\leftarrow\adv({\bf m})$\;
			\Return $(g = \beta \mathrel) \wedge \bigwedge\limits_{1\leq i \leq |{\bf c}|} (c \not={\bf c}[i])$\;
\end{inlineexperiment}
}

\noindent
In the above game, we require the plaintexts $m_0,m_1\in\mathcal{M}$ and $|m_0| = |m_1|$.
We say that the encryption scheme $\Pi$ satisfies indistinguishability under parallel attack $\INDPA$, if for all probabilistic polynomial-time adversaries $\adv$,
there exists a negligible function $\negl$, such that for all security parameters $\kk$, the inequality
$$\Succ(\gameINDPA) \leq \frac{1}{2} + \neglK \,, $$
holds. 
\end{definition}

\ifdefined\acmformat
  \bibliographystyle{ACM-Reference-Format}
\else
  \bibliographystyle{plain}
\fi
\bibliography{main-tutorial}

\end{document}